\documentclass[pdflatex,sn-mathphys,Numbered,iicol]{sn-jnl}
\usepackage[T1]{fontenc}
\usepackage{ifthen}
\usepackage{tikz}
\usetikzlibrary{spy,backgrounds}
\usepackage{dsfont}
\usepackage{pgfplots}
\usepackage[ruled,linesnumbered,noend]{algorithm2e}
\usepackage{csquotes}
\usepackage{glossaries}
\newacronym{gmm}{{GMM}}{Gaussian mixure model}
\newacronym{sde}{{SDE}}{stochastic differential equation}
\newacronym{kde}{{KDE}}{kernel density estimate}
\newacronym{gmdm}{{GMDM}}{Gaussian mixture diffusion model}
\newacronym{ai}{{AI}}{Artificial Intelligence}
\newacronym{art}{{ART}}{Algebraic Reconstruction Technique}
\newacronym{sirt}{{SIRT}}{Simultaneous Iterative Reconstruction Technique}
\newacronym{sart}{{SART}}{Simultaneous Algebraic Reconstruction Technique}
\newacronym{bicav}{{BICAV}}{Block Iterative Component Averaging}
\newacronym{cgls}{{CGLS}}{Conjugate Gradient Least Squares}
\newacronym{ossqs}{{OS-SQS}}{Ordered Subset Separable Quadratic Surrogates}
\newacronym{ann}{{ANN}}{Artificial Neural Network}
\newacronym{bp}{{BP}}{Backprojection}
\newacronym{admm}{{ADMM}}{Alternating Direction Method of Multipliers}
\newacronym{cnn}{{CNN}}{convolutional neural network}
\newacronym{cg}{{CG}}{Conjugate Gradient}
\newacronym{cs}{{CS}}{Compressed Sensing}
\newacronym{ct}{{CT}}{computed tomography}
\newacronym{cpu}{{CPU}}{Central Processing Unit}
\newacronym{dicom}{{DICOM}}{Digital Imaging and Communications in Medicine}
\newacronym{dof}{{DOF}}{Degrees of Freedom}
\newacronym{ddr}{{DDR}}{Digitally Reconstructed Radiograph}
\newacronym{dsc}{{DSC}}{Dice Similarity Coefficient}
\newacronym{dti}{{DTI}}{Diffusion Tensor Imaging}
\newacronym{ecg}{{ECG}}{Electrocardiography}
\newacronym{em}{{EM}}{Expectation Maximization}
\newacronym{ft}{{FT}}{Fourier Transform}
\newacronym{fft}{{FFT}}{fast Fourier transform}
\newacronym{fista}{{FISTA}}{Fast Iterative Shrinkage and Thresholding Algorithm}
\newacronym{fbp}{{FBP}}{Filtered Back-Projection}
\newacronym{foe}{{FoE}}{fields-of-experts}
\newacronym{fov}{{FoV}}{Field of View}
\newacronym{gac}{{GAC}}{Geodesic Active Contours}
\newacronym{gan}{{GAN}}{generative adversarial network}
\newacronym{gd}{{GD}}{Gradient Descent}
\newacronym{gpu}{{GPU}}{Graphics Processing Unit}
\newacronym{hu}{{HU}}{Hounsfield Units}
\newacronym{ista}{{ISTA}}{Iterative Shrinkage and Thresholding Algorithm}
\newacronym{iipg}{{IIPG}}{Inertial Incremental Proximal Gradient}
\newacronym{ipalm}{{iPALM}}{inertial proximal alternating linearized minimization}
\newacronym{lsc}{{l.s.c.}}{lower-semicontinuous}
\newacronym{lista}{{LISTA}}{Learned Iterative Shrinkage and Thresholding Algorithm}
\newacronym{lbfgs}{{L-BFGS}}{Limited-Memory Broyden-Fletcher-Goldfarb-Shanno}
\newacronym{map}{{MAP}}{maximum a-posteriori}
\newacronym{mlp}{{MLP}}{Multi Layer Perceptron}
\newacronym{mr}{{MR}}{Magnetic Resonance}
\newacronym{ml}{{ML}}{Maximum Likelihood}
\newacronym{mri}{{MRI}}{magnetic resonance imaging}
\newacronym{mae}{{MAE}}{Mean Absolute Error}
\newacronym{mse}{{MSE}}{Mean Squared Error}
\newacronym{msssim}{{MS-SSIM}}{Multi-Scale Structural Similarity Index}
\newacronym{ncc}{{NCC}}{Normalized Cross Correlation}
\newacronym{nlm}{{NLM}}{Non-Local Means}
\newacronym{nufft}{{NUFFT}}{non-uniform fast Fourier transform}
\newacronym{nrmse}{{NRMSE}}{Normalized Root Mean Squared Error}
\newacronym{icp}{{ICP}}{Iterative Closest Point}
\newacronym{pat}{{PAT}}{Photoacoustic Tomography}
\newacronym{pca}{{PCA}}{Principal Component Analysis}
\newacronym{pet}{{PET}}{Positron Emission Tomography}
\newacronym{psf}{{PSF}}{Point Spread Function}
\newacronym{psnr}{{PSNR}}{peak signal-to-noise ratio}
\newacronym{rbf}{{RBF}}{Gaussian radial basis function}
\newacronym{relu}{{ReLU}}{Rectified Linear Unit}
\newacronym{roi}{{ROI}}{Region Of Interest}
\newacronym{snr}{{SNR}}{Signal-to-Noise Ratio}
\newacronym{sota}{SotA}{State-of-the-Art}
\newacronym{spect}{{SPECT}}{Single Photon Emission Computed Tomography}
\newacronym{ssim}{{SSIM}}{structural similarity}
\newacronym{tof}{{ToF}}{Time of Flight}
\newacronym{tgv}{{TGV}}{Total Generalized Variation}
\newacronym{tv}{{TV}}{total variation}
\newacronym{us}{{US}}{Ultrasound}
\newacronym{vn}{{VN}}{variational network}
\newacronym{sbp}{{SBP}}{Simple Back-Projection}
\newacronym{fdk}{{FDK}}{Feldkamp-Davis-Kress}
\newacronym{aec}{{AEC}}{Automatic Exposure Control}
\newacronym{bm3d}{{BM3D}}{Block Matching and 3D Filtering}
\newacronym{mcmc}{{MCMC}}{Markov chain Monte Carlo}
\newacronym{cd}{{CD}}{Contrastive Divergence}
\newacronym{poe}{{PoE}}{Products of Experts}
\newacronym{mala}{{MALA}}{Metropolis adjusted Langevin algorithm}
\newacronym{lmc}{{LMC}}{Langevin Monte Carlo}
\newacronym{ula}{{ULA}}{unadjusted Langevin algorithm}
\newacronym{tdv}{{TDV}}{Total Deep Variation}
\newacronym{ebm}{{EBM}}{energy-based model}
\newacronym{rss}{{RSS}}{root-sum-of-squares}
\newacronym{nmse}{{NMSE}}{normalized mean-squared error}
\newacronym{acl}{{ACL}}{auto-calibration lines}
\newacronym{mmse}{{MMSE}}{minimum mean-squared-error}
\newacronym{ebpa}{{EB-PA}}{empirical Bayes-patch averaging}
\newacronym{hqs}{{HQS}}{half-quadratic splitting}
\newacronym{epll}{{EPLL}}{expected patch log-likelihood}
\newacronym{pgc}{{PCG}}{proximal gradient continuation}
\newacronym{gsm}{{GSM}}{Gaussian scale mixture}
\newacronym{pogmdm}{{PoGMDM}}{product of Gaussian mixture diffusion model}
\newacronym{pde}{{PDE}}{partial differential equation}
\usepackage{amsmath}
\usepackage{graphicx}%
\usepackage{multirow}%
\usepackage{amsmath,amssymb,amsfonts}%
\usepackage{amsthm}%
\usepackage{mathrsfs}%
\usepackage[title]{appendix}%
\usepackage{xcolor}%
\usepackage{textcomp}%
\usepackage{manyfoot}%
\usepackage{booktabs}%
\usepackage[defaultcolor=brown,final]{changes}
\glsdisablehyper%

\pgfplotsset{compat=1.17}
\pgfplotsset{colormap={inferno}{%
rgb = (1.46200e-03, 4.66000e-04, 1.38660e-02)
rgb = (2.94320e-02, 2.15030e-02, 1.14621e-01)
rgb = (9.29900e-02, 4.55830e-02, 2.34358e-01)
rgb = (1.83429e-01, 4.03290e-02, 3.54971e-01)
rgb = (2.71347e-01, 4.09220e-02, 4.11976e-01)
rgb = (3.60284e-01, 6.92470e-02, 4.31497e-01)
rgb = (4.41207e-01, 9.93380e-02, 4.31594e-01)
rgb = (5.28444e-01, 1.30341e-01, 4.18142e-01)
rgb = (6.09330e-01, 1.59474e-01, 3.93589e-01)
rgb = (6.94627e-01, 1.95021e-01, 3.54388e-01)
rgb = (7.69556e-01, 2.36077e-01, 3.07485e-01)
rgb = (8.41969e-01, 2.92933e-01, 2.48564e-01)
rgb = (8.98192e-01, 3.58911e-01, 1.88860e-01)
rgb = (9.44285e-01, 4.42772e-01, 1.20354e-01)
rgb = (9.72590e-01, 5.29798e-01, 5.33240e-02)
rgb = (9.86964e-01, 6.30485e-01, 3.09080e-02)
rgb = (9.84865e-01, 7.28427e-01, 1.20785e-01)
rgb = (9.66243e-01, 8.36191e-01, 2.61534e-01)
rgb = (9.46392e-01, 9.30761e-01, 4.42367e-01)
rgb = (9.88362e-01, 9.98364e-01, 6.44924e-01)
}}
\newcommand{\drawcolorbar}{%
	\pgfplotscolorbardrawstandalone[
		scale=0.32, colormap={example}{samples of colormap = (8 of inferno)},
		colorbar horizontal,point meta max=0.2,colorbar style={ticks=none},
	]%
}
\newcommand{\drawcolorbarbw}{%
	\pgfplotscolorbardrawstandalone[
		scale=0.32, colormap={example}{samples of colormap = (8 of colormap/blackwhite)},
		colorbar horizontal,point meta max=0.2,colorbar style={ticks=none},
	]%
}

\usepackage{caption}
\usepackage{subcaption}
\usepackage{listings}%
\usepackage[title]{appendix}%
\usepackage{amssymb}
\usepackage{mathtools}
\usepackage[separate-uncertainty=true,detect-weight=true]{siunitx}
\usepackage[capitalize]{cleveref}
\usepackage{microtype}
\usepackage{xcolor}
\usepackage{booktabs}
\usepackage{multirow}
\definecolor{mplblue}{HTML}{1f77b4}
\definecolor{mplorange}{HTML}{ff7f0e}
\definecolor{mplgreen}{HTML}{2ca02c}
\definecolor{mplred}{HTML}{d62728}
\definecolor{mplpurple}{HTML}{9467bd}

\definecolor{coolwarm1}{HTML}{3A4CC0}
\definecolor{coolwarm2}{HTML}{8DB0FE}
\definecolor{coolwarm3}{HTML}{DEDDDC}
\definecolor{coolwarm4}{HTML}{F4997A}
\definecolor{coolwarm5}{HTML}{B40326}

\newboolean{review}

\newtheorem{theorem}{Theorem}
\newtheorem{corollary}{Corollary}
\setboolean{review}{false}
\ifthenelse{\boolean{review}}{
	\author{%
		First Author\inst{1}\orcidID{0000-1111-2222-3333} \and%
		Second Author\inst{1}\orcidID{1111-2222-3333-4444}
	}
	\authorrunning{F. Author et al.}
	\institute{%
		Princeton University, Princeton NJ 08544, USA
	}
}{
	\author*[1]{\fnm{Martin} \sur{Zach}}\email{martin.zach@icg.tugraz.at}
	\author[2]{\fnm{Erich} \sur{Kobler}}\email{kobler@uni-bonn.de}
	\author[3]{\fnm{Antonin} \sur{Chambolle}}\email{antonin.chambolle@ceremade.dauphine.fr}
	\author[1]{\fnm{Thomas} \sur{Pock}}\email{pock@icg.tugraz.at}
	\affil*[1]{%
		\orgdiv{Institute of Computer Graphics and Vision}, %
		\orgname{Graz University of Technology}, %
		\orgaddress{\country{Austria}}%
	}
	\affil[2]{%
		\orgdiv{Klinik für Neuroradiologie}, %
		\orgname{Universitätsklinikum Bonn}, %
		\orgaddress{\country{Germany}}%
	}
	\affil[3]{%
		\orgdiv{CEREMADE, CNRS and Paris Dauphine University (PSL), and Mokaplan, INRIA Paris},
		\orgaddress{\country{France}}%
	}
}

\title{Product of Gaussian Mixture Diffusion Models}

\newcommand{\R}{\mathbb{R}}
\newcommand{\T}{\top}
\newcommand{\Id}{\mathrm{Id}}
\newcommand{\argm}{\,\cdot\,}
\newcommand{\grad}[1]{\nabla_{\mkern-4mu#1}}

\DeclareMathOperator{\tr}{\mathrm{Tr}}
\DeclareMathOperator{\proj}{proj}
\DeclareMathOperator*{\argmin}{arg\,min}
\DeclareMathOperator*{\argmax}{arg\,max}
\DeclareMathOperator{\diag}{diag}
\DeclarePairedDelimiter\norm{\lVert}{\rVert}
\DeclarePairedDelimiter\detm{\lvert}{\rvert}

\begin{document}
\abstract{%
	In this work we tackle the problem of estimating the density \( f_X \) of a random variable \( X \) by successive smoothing, such that the smoothed random variable \( Y \) fulfills the diffusion partial differential equation \( (\partial_t - \Delta_1)f_Y(\argm, t) = 0 \) with initial condition \( f_Y(\argm, 0) = f_X \).
	We propose a product-of-experts-type model utilizing Gaussian mixture experts and study configurations that admit an analytic expression for \( f_Y (\argm, t) \).
	In particular, with a focus on image processing, we derive conditions for models acting on filter-, wavelet-, and shearlet\added{-}responses.
	Our construction naturally allows the model to be trained simultaneously over the entire diffusion horizon using empirical Bayes.
	We show numerical results for image denoising where our models are competitive while being tractable, interpretable, and having only a small number of learnable parameters.
	As a byproduct, our models can be used for reliable noise \added{level} estimation, allowing blind denoising of images corrupted by heteroscedastic noise.
}
\maketitle
\section{Introduction}
The problem of estimating the probability density \( f_X : \mathcal{X} \to \R \) of a random variable \( X \) in \( \mathcal{X} \), given a set of data samples \( \{ x_i \}_{i=1}^N \) drawn from \( f_X \) has received significant attention in the recent years~\cite{hinton_training_2002,sohl-dickstein15,song_generative_2019,du_implicit_2019,ho_ddpm_2020}.
The applications range from purely generative purposes~\cite{rombach2021highresolution,ho_ddpm_2020}, over classical image restoration problems~\cite{zach_explicit_2023,ozdenizci2023,lugmayr_repaint_2022} to medical image reconstruction~\cite{zach_stable_2023,CHUNG2022102479,zach_computed_2021}.
This is a challenging problem in high dimension (e.g.\ for images of size \( M \times N \), i.e.\ \( \mathcal{X} = \R^{M \times N} \)), due to extremely sparsely populated regions~\cite{bengio_representation_2013}.
A fruitful approach is to estimate the density at different times when undergoing a diffusion process~\cite{song_generative_2019,ho_ddpm_2020}.
Intuitively, the diffusion equilibrates high- and low-density regions over time, thus easing the estimation problem.

Let \( Y_t \) (carelessly) denote the random variable whose distribution is defined by diffusing \( f_X \) for some time \( t \).
We denote the density of \( Y_t \) by \( f_Y(\argm, t) \), which fulfills the diffusion \gls{pde} \( (\partial_t - \Delta_1)f_Y(\argm, t) = 0 \) with initial condition \( f_Y(\argm, 0) = f_X \).
The empirical Bayes theory~\cite{robbins_empirical_1956} provides a machinery for reversing the diffusion \replaced{\gls{pde}}{process}:
Given an instantiation of the random variable \( Y_t \), the Bayesian least-squares estimate of \( X \) can be expressed solely using \( f_Y(\argm, t) \).
Importantly, this holds for all positive \( t \), as long as \( f_Y \) is properly constructed.

In practice we wish to have a parametrized, trainable model of \( f_Y \), say \( f_\theta \) where \( \theta \) is a parameter vector, such that \( f_Y(x, t) \approx f_\theta(x, t) \) for all \( x \in \mathcal{X} \) and all \( t \in [0, \infty) \).
Recent choices~\cite{song_generative_2019,song_scorebased_2021} for the family of functions \( f_\theta(\argm, t) \) were of practical nature:
Instead of an analytic expression for \( f_\theta \) at any time \( t \), authors proposed a time-conditioned network in the hope that it can learn to behave as if it had undergone \replaced{the}{a} diffusion \replaced{\gls{pde}}{process}.
Further, instead of worrying about the normalization \( \int_\mathcal{X} f_Y(\argm, t) = 1 \) for all \( t \in [0, \infty) \), usually they directly estimate the \emph{score} \( -\grad{1} \log f_Y(\argm, t) : \mathcal{X} \to \mathcal{X} \) with some network~\( s_\theta(\argm, t) : \mathcal{X} \to \mathcal{X} \).
This has the advantage that normalization constants vanish, but usually the constraint \( \partial_j (s_\theta(\argm, t))_i = \partial_i (s_\theta(\argm, t))_j \) is not enforced in the architecture of \( s_\theta \).
Thus, \( s_\theta(\argm, t) \) is in general not the gradient of a scalar function (the negative-log-density it claims to model).

In contrast to this line of works, in this paper we pursue a more principled approach.
Specifically, we leverage products of \deleted{one-dimensional} \gls{gmm} experts to model the distribution of responses of transformations acting on natural images.
\added{%
	Here, an \emph{expert} is a one-dimensional distribution modeling certain characteristics of the random variable \( Y_t \) (the terminology is borrowed from \cite{hinton_training_2002}).%
}
In particular, we derive conditions under which \( f_Y(\argm, t) \) can be expressed analytically from \( f_Y(\argm, 0 ) \).
We call our model \gls{pogmdm} to reflect the building blocks: products of \gls{gmm} experts and diffusion.
The conditions arising for \replaced{a model acting on filter-responses}{the filter bank of a patch-based model} naturally lead to the consideration of the wavelet transformation.
Analyzing the conditions for a convolutional model naturally leads to the shearlet transformation.
Thus, we present three models that utilize transformations that are extremely popular in image processing:
Linear filters, the wavelet transformation and the shearlet transformation.
\added{%
	To the best of our knowledge, this paper is the first in proposing strategies to learn patch-based and convolutional priors in a unified framework.%
}
\subsection{Contributions}
This paper constitutes an extension to our previous conference publication~\cite{zach_explicit_2023}, in which we introduced the idea of explicit diffusion models and showed preliminary numerical results.
In this work, we present two additional explicit diffusion models and derive conditions under which they fulfill the diffusion \gls{pde}.
In particular, the derived conditions naturally lead to models that can leverage transformations that are well known and popular in image processing: Wavelets and shearlets.
For all models, we show how the associated transformation, along with the diffusion model for the density, can be learned.
We provide numerical results for denoising and extend our analysis regarding \replaced{noise level estimation and blind heteroscedastic denoising}{blind noise estimation}.
Our contributions can be summarized as follows:
\begin{itemize}
	\item We derive conditions under which products of \gls{gmm} experts acting on filter-, wavelet-, and shearlet-responses \replaced{obey}{respect} the \deleted{linear isotropic} diffusion \gls{pde}.
	\item We show how the \gls{gmm} experts, along with the corresponding transformations, can be learned and provide algorithms to solve the arising sub-problems.
	\item We evaluate the learned models on the prototypical image restoration problem: denoising.
\end{itemize}
Code for training, validation, and visualization, along with pre-trained models is available at \href{https://github.com/VLOGroup/PoGMDM}{https://github.com/VLOGroup/PoGMDM}.

This paper is organized as follows:
In~\cref{sec:background}, we give background information on \deleted{isotropic} diffusion and how it can be used for parameter estimation of learned densities.
This section also encompasses an overview of related work.
In \cref{sec:methods}, we introduce the backbone of our models and derive conditions under which they obey the diffusion \gls{pde}.
We demonstrate the practical applicability of our models in \cref{sec:numerics} with numerical experiments.
We explore alternative parametrizations and possible extensions of our models in~\cref{sec:discussion} and finally conclude the paper, providing future research directions, in \cref{sec:conclusion}.
\subsection{Notation and Preliminaries}
For the sake of simplicity, throughout this article, we assume that all distributions admit a density with respect to the Lebesgue measure, although the numerical experiments only assume access to an empirical distribution.
Thus, we use the terms \emph{distribution} and \emph{density} interchangeably.
In~\cref{sec:methods}, we define normal distributions that are supported \emph{on a subspace} (e.g.\ the zero-mean subspace \added{$\{ x \in \R^n : \langle \mathds{1}_{\R^n}, x \rangle_{\R^n} = 0 \}$}).
In this case, we restrict our analysis to the support, which is theoretically supported by the disintegration theorem~\cite{Rao1973}.
\added{We use the symbols \( \R_{+} \) and \( \R_{++} \) to denote the non-negative real numbers \( \{ x \in \R : x \geq 0 \} \) and positive reals numbers \( \{ x \in \R : x > 0 \} \) respectively.}
We denote with \( \langle \argm, \argm \rangle_{\R^n} \added{ : \R^n \times \R^n \to \R : (x, y) \mapsto \sum_{i=1}^n x_i y_i} \) the standard inner product in the Euclidean space \( \R^n \), and \replaced{with \( \norm{\argm}^2 : \R^n \to \R_+ \) the map \( x \mapsto \langle x, x\rangle_{\R^n}\)}{\( \norm{\argm}^2 = \langle \argm, \argm \rangle_{\R^n} \)}.
\added{In addition, \( (\argm \otimes \argm) : \R^n \times \R^n \to \R^{n \times n} \) is the standard outer product in \( \R^n \): \( (x \otimes y)_{ij} = x_i y_j \).}
\added{$\operatorname{conj}$ denotes element-wise complex conjugation.}
Let \( \mathcal{Q} \subset \mathcal{H} \) be a (not necessarily convex) subset of a Hilbert space \( \mathcal{H} \).
We define by \( \proj_\mathcal{Q} : \mathcal{H} \to \replaced{\mathcal{H}}{2^\mathcal{H}} \) the orthogonal projection onto the set \( \mathcal{Q} \).
\added{%
	With slight abuse of notation, we ignore that this is a multivalued map in general.%
}
\( L^2(\Omega) \) denotes the standard Lebes\replaced{g}{q}ue space on a domain \( \Omega \subset \R^n \).
Lastly, \( \mathrm{Id}_{\mathcal{H}} \) and \( \mathds{1}_\mathcal{H} \) denote the identity map and the one-vector in \( \mathcal{H} \) respectively.
\section{Background}%
\label{sec:background}
In this section, we first emphasize the importance of \deleted{the} diffusion \deleted{process} in density estimation (and sampling) in high dimensions.
Then, we detail the relationship between diffusing the density function, empirical Bayes, and denoising score matching~\cite{vincent_connection_2011}.
\subsection{Diffusion Eases Density Estimation and Sampling}
Let \( f_X \) be a density on \( \mathcal{X} \subset \R^d\).
A major difficulty in estimating \( f_X \) with parametric models is that \( f_X \) is extremely sparsely populated in high dimensional spaces\footnote{Without any reference to samples \( x_i \sim f_X \), an equivalent statement may be that \( f_X \) is (close to) zero almost everywhere (in the layman --- not measure-theoretic --- sense).}, i.e., \( d\gg1 \).
This phenomenon has many names, e.g.\ the curse of dimensionality or the manifold hypothesis~\cite{bengio_representation_2013}.
Thus, the learning problem is difficult, since meaningful gradients are rare.
Conversely, let us for the moment assume \replaced{access to}{we have} a model \( \tilde{f}_X \) that approximates \( f_X \) well.
In general, it is still very challenging to generate a set of points \( \{ x_i \}_{i=1}^I \) such that we can confidently say that the associated empirical density \( \frac{1}{I} \sum_{i=1}^I \delta_{x_i} \) approximates \( \tilde{f}_X \) \added{(let alone $f_X$) well}.
This is because, in general, there does not exist a procedure to directly draw \added{samples} from \( \tilde{f}_X \), and \deleted{modern} Markov chain Monte Carlo methods \replaced{are prohibitively slow in practice, especially for multimodal distributions in high dimensions~\cite{song_generative_2019}}{rely on the estimated gradients of \( \tilde{f}_X \) and, in practice, only works well for unimodal distributions~\cite{song_scorebased_2021}}.

The \deleted{isotropic} diffusion \replaced{\gls{pde}}{process} or heat equation
\begin{equation}
	(\partial_t - \Delta_1) f(\argm, t) = 0\ \text{with initial condition}\ f(\argm, 0) = f_X
	\label{eq:diff}
\end{equation}
equilibrates the density \( f_X \), thus mitigating the challenges outlined above.
Here, \( \partial_t \) denotes the standard partial derivative with respect to time \( \frac{\partial}{\partial t} \) and \( \Delta_1 = \tr \circ \grad{1}^2 \) is the Laplace operator, where the \( 1 \) indicates its application to the first argument.
We detail the evolution of \( f_X \) under this \replaced{diffusion \gls{pde}}{process} and relations to empirical Bayes in~\cref{ssec:diffusion empirical bayes}.

\emph{Learning} \( f(\argm, t) \) for \( t \geq 0 \) is more stable since the diffusion \enquote{fills the space} with meaningful gradients~\cite{song_generative_2019}.
Of course, this assumes that for different times \( t_1 \) and \( t_2 \), the models of \( f(\argm, t_1) \) and \( f(\argm, t_2) \) are somehow related to each other.
As an example of this relation, the recently popularized noise-conditional score-network~\cite{song_generative_2019} shares convolution filters over time, but their input is transformed through a time-conditional instance normalization.
In this work, we make this relation explicit by considering a family of functions \( f(\argm, 0) \) for which \( f(\argm, t) \) can be expressed analytically.

For \emph{sampling}, \( f(\argm, t) \) for \( t > 0 \) can help by gradually moving samples towards high-density regions of \( f_X \), regardless of initialization.
To utilize this, a very simple idea with relations to simulated annealing~\cite{Kirkpatrick1983} is to have a pre-defined time schedule \( t_T > t_{T-1} > \ldots > t_{\replaced{0}{1}} > 0 \) and sample \( f(\argm, t_i) \), \( i = T, \dotsc, 0 \) (e.g.\ with Langevin Monte Carlo~\cite{roberts_exponential_1996}) successively~\cite{song_generative_2019}.
In~\cite{song_scorebased_2021}, instead of considering discrete time steps, the authors propose to model the sampling procedure as a continuous-time stochastic differential equation.
We note that the \deleted{isotropic} diffusion \gls{pde}~\eqref{eq:diff} on the densities corresponds to the stochastic differential equation
\begin{equation}
	\mathrm{d}X = \sqrt{2} \mathrm{d}W\deleted{,}
	\label{eq:sde view}
\end{equation}
on the random variables, where \( W \) is the standard Wiener process.
This is known as the variance exploding stochastic differential equation in the literature.
\subsection{Diffusion, Empirical Bayes, and Denoising Score Matching}%
\label{ssec:diffusion empirical bayes}
In this section, similar to the introduction, we again adopt the interpretation that the evolution in~\eqref{eq:diff} defines the density of a \deleted{smoothed} random variable \( Y_t \).
That is, \( Y_t \) is a random variable with probability density \( f_Y(\argm, t) \), which fulfills \( (\partial_t - \Delta_1) f_Y(\argm, t) = 0 \) and \( f_Y(\argm, 0) = f_X \).
It is well known that Green's function of~\eqref{eq:diff} is a Gaussian (see e.g.~\cite{cole_green_2010}) with zero mean and \added{co}variance \( \deleted{\sigma^2(t) =} 2t\Id\).
In other words, for \( t > 0 \) we can write \( f_Y(\argm, t) = G_{0,2t\Id_{\added{\mathcal{X}}}} * f_X \), where 
\begin{equation}
	G_{\mu,\Sigma}(x) = |2 \pi \Sigma|^{-1/2} \exp\bigl( - \norm{x - \mu}_{\Sigma^{-1}}^2 / 2 \bigr).
\end{equation}
Thus, the diffusion \replaced{\gls{pde}}{process} constructs a (linear) \emph{scale space in the space of probability densities} \added{and we refer to \( Y_t \) (respectively \( f_{Y_t} \)) as the \emph{smoothed} random variable (respectively density)}.
Equivalently, in terms of the random variables, we can write \( Y_t = X + \sqrt{2t}N \) where \( N \) is a random variable with normal distribution \( \mathcal{N}(0, \Id_{\mathcal{X}}) \).
Next, we show how to estimate the corresponding instantiation of \( X \) which has \enquote{most likely} spawned an instantiation of \( Y_t \) using empirical Bayes.

In the school of empirical Bayes~\cite{robbins_empirical_1956}, we try to estimate a clean random variable given a corrupted instantiation, using only knowledge about the \replaced{density of the random variable corresponding to the corrupted instance}{corrupted density}.
In particular, for our setup we have a corruption model
\begin{equation}
	y_t = x + \sqrt{2t} \eta,
	\label{eq:corruption}
\end{equation}
\added{with \( x \sim f_X \) and \( \eta \sim \mathcal{N}(0, \Id_{\mathcal{X}})\).}
It is well known that the Bayesian \gls{mmse} estimate is the conditional mean, i.e.\ the map \( y_t \mapsto \int x f_{X \mid Y_t}(x \mid y_t)\,\mathrm{d}x \).
In classical Bayes theory, such a map would be constructed by utilizing Bayes theorem, i.e. writing 
\( f_{X \mid Y_t} = \frac{f_{Y_t\mid X}f_{X}}{f_{Y_t}} \) and choosing an appropriate prior \( f_{X} \).
However, a classical result from empirical Bayes estimation reveals that a map \( y_t \mapsto \int x f_{X \mid Y_t}(x \mid y_t)\,\mathrm{d}x \) can be constructed \emph{only} assuming access to the \emph{\replaced{smoothed}{measurement} density} \( f_{Y_t} \) without any reference to the prior \( f_X \).
This result is known as the Miyasawa estimate~\cite{miyasawa_empirical_1961} or Tweedie's formula~\cite{efron_tweedie_2011,raphan_least_2011} and we derive it here for completeness.

First, by the corruption model~\eqref{eq:corruption} we can write
\begin{equation}
	f_{Y_t \mid X}(y \mid x) = \detm{2\pi\sigma^2\Id}^{-\frac{1}{2}} \exp\Bigl( -\frac{\norm{y - x}^2}{2\sigma^2} \Bigr),
\end{equation}
where we use the relation \( \sigma = \sqrt{2t} \), and thus by Bayes theorem it follows that
\begin{equation}
	\begin{aligned}
		f_{Y_t}(y) &= \int f_{Y_t\mid X}(y\mid x) f_{X}(x)\,\mathrm{d}x \\
				   &= \int \detm{2\pi\sigma^2\Id}^{-\frac{1}{2}} \exp\Bigl( -\frac{\norm{y - x}^2}{2\sigma^2} \Bigr) f_{X}(x)\,\mathrm{d}x.
	\end{aligned}
	\label{eq:py}
\end{equation}
Taking the gradient w.r.t.\ \( y \) and multiplying by \( \sigma^2 \) yields
\begin{equation}
	\begin{aligned}
		&\sigma^2 \nabla f_{Y_t}(y) \\
		&= \int (x - y)\detm{2\pi\sigma^2\Id}^{-\frac{1}{2}} \exp\Bigl( -\frac{\norm{y - x}^2}{2\sigma^2} \Bigr) f_{X}(x)\,\mathrm{d}x \\
		&= \int (x - y) f_{X, Y_t}(x, y)\,\mathrm{d}x\\
		&= \int x f_{X, Y_t}(x, y)\,\mathrm{d}x - yf_{Y_t}(y)
	\end{aligned}
\end{equation}
and after dividing by \( f_{Y_t} \) it follows that
\begin{equation}
	y + \sigma^2 \frac{\nabla f_{Y_t}(y)}{f_{Y_t}(y)} = \int xf_{X\mid Y_t}(x\mid y)\,\mathrm{d}x,
\end{equation}
where we used the definition of conditional densities, i.e.\ that \( f_{X \mid Y_t} = \frac{f_{X, Y_t}}{f_{Y_t}} \).
Finally, by noting that \( \frac{\nabla f_{Y_t}}{f_{Y_t}} = \nabla \log f_{Y_t} \), the above can be rewritten as
\begin{equation}
	y + \sigma^2 \nabla \log f_{Y_t}(y) = \int xf_{X\mid Y_t}(x\mid y)\,\mathrm{d}x.
	\label{eq:tweedie}
\end{equation}
We refer to the work of Raphan and Simoncelli~\cite{raphan_least_2011} for an empirical Bayes theory encompassing a more general family of corruptions.
They refer to this type of estimator more generally as non-parametric empirical Bayes least-squares (NEBLS).

We illustrate the idea of empirical \gls{mmse} estimation on a toy example in~\cref{fig:diffusion toy example}, where the data distribution consists of Dirac measures \( f_X = \sum_{i=1}^6 w_i \delta_{x_i} \) with
\begin{equation}
	\begin{pmatrix} x_1^\top \\ \vdotswithin{x_1^\top} \\ x_6^\top \end{pmatrix} = \begin{pmatrix}
		0.588 & 0.966 \\
		0.289 & 0.112 \\
		-0.313 &-0.924 \\
		-0.696 & 0.990 \\
		-0.906 & 0.030 \\
		-0.516 & 0.039
	\end{pmatrix}%
	\text{ and }
	w = \begin{pmatrix} 0.23 \\ 0.1 \\ 0.09 \\ 0.19 \\ 0.29 \\ 0.08 \end{pmatrix}.%
	\label{eq:example}
\end{equation}
The figure illustrates that that \( f_{Y_t} \) approaches a simple form as \( t \) \replaced{approaches infinity}{$\to \infty$}.
Indeed, it has been shown~\cite{kobler2023learning} that \( f_{Y_t} \) is log-concave for large enough \( t \), and \replaced{$-\log f_{Y_t}$ approaches a quadratic function}{in particular \( \nabla \log f_{Y_t} \) approaches a linear vector field}.
\begin{figure*}
	\centering
	\includegraphics[width=.33\textwidth]{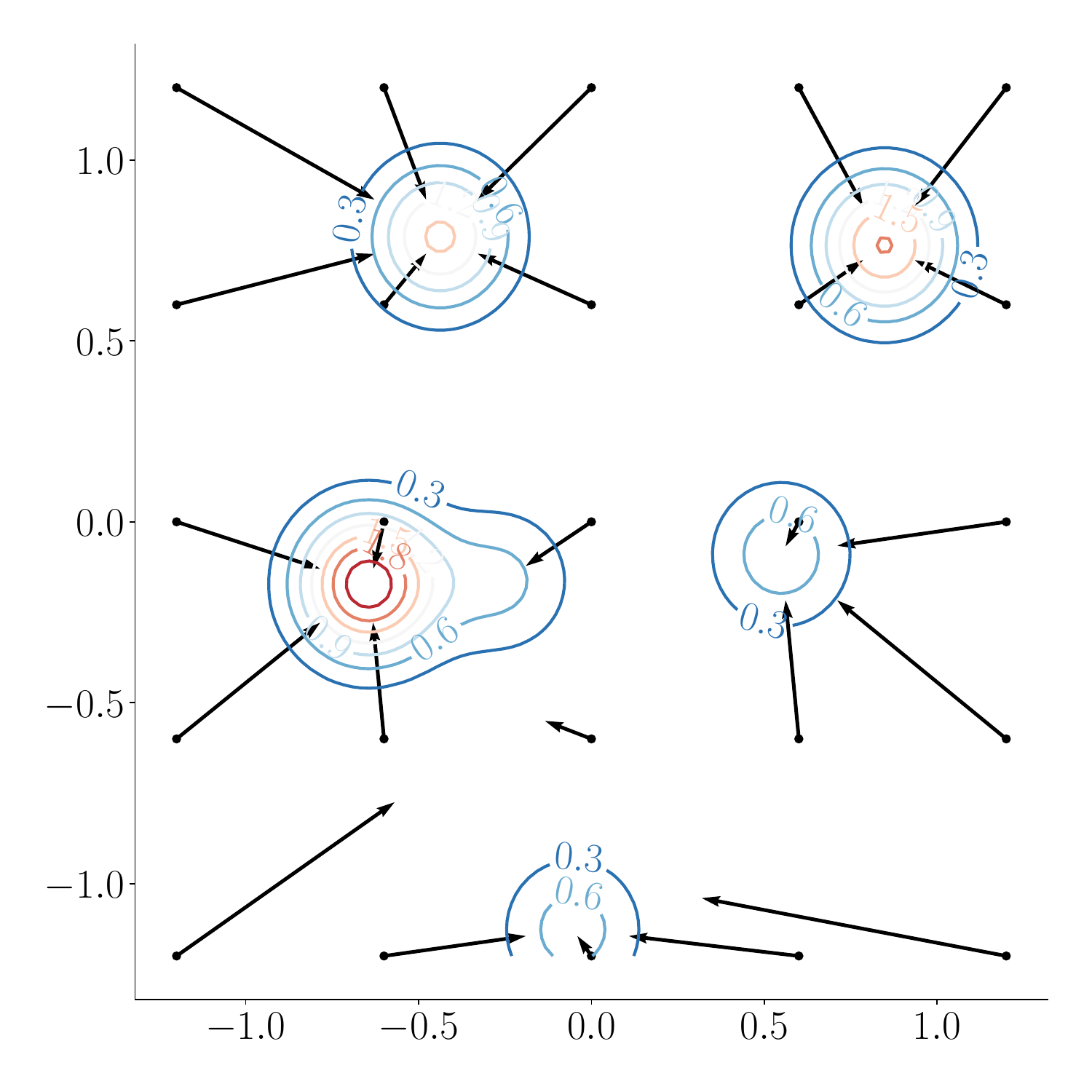}
	\includegraphics[width=.33\textwidth]{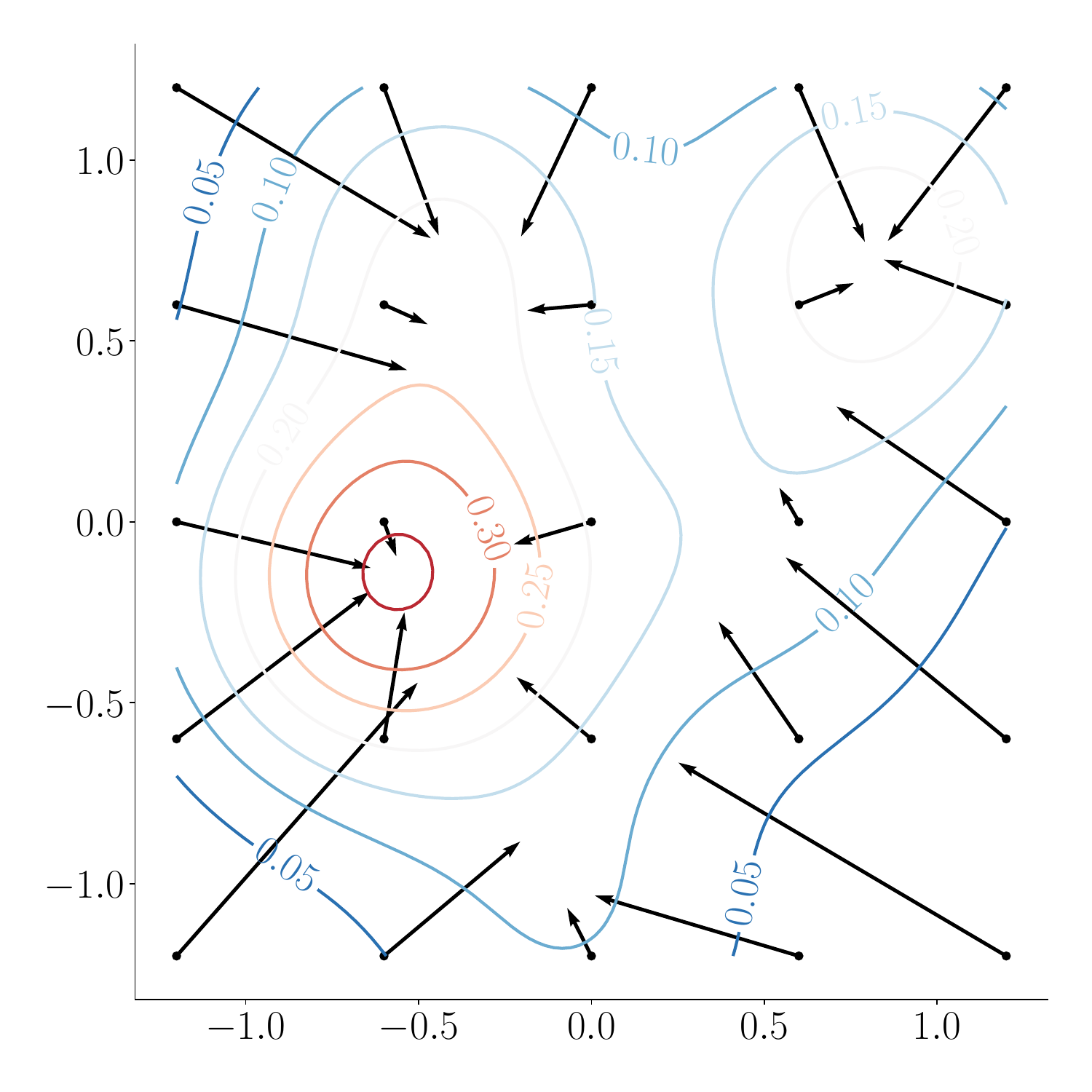}
	\includegraphics[width=.33\textwidth]{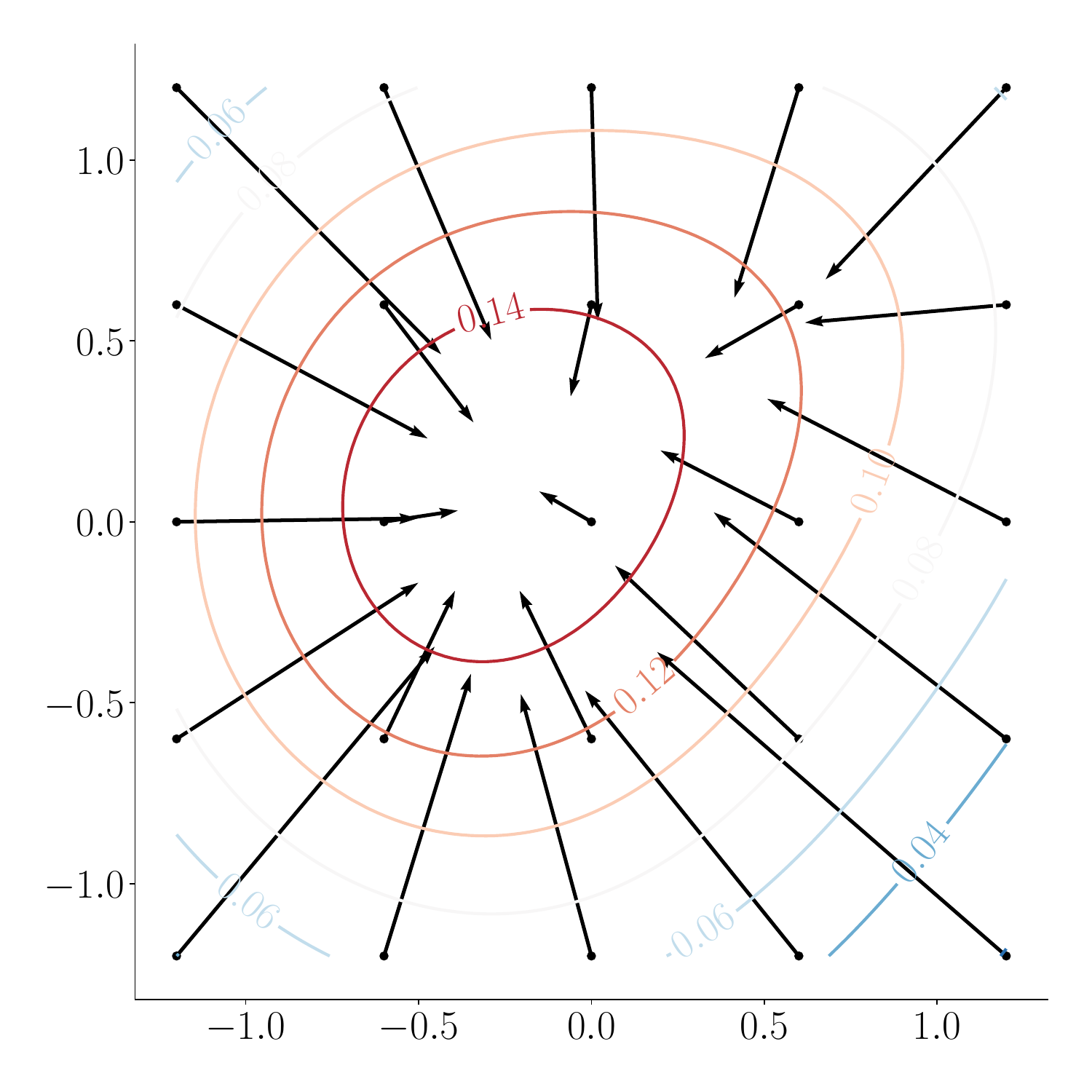}
	\caption{%
		\deleted{Linear isotropic} Diffusion of an empirical density consisting of the weighted Dirac measures \( \sum_{i=1}^6 w_i \delta_{x_i} \) specified in~\eqref{eq:example} at times \( t \in \{ \num{0.008}, \num{0.078}, \num{0.3} \} \).
		The arrows show the empirical Bayes estimate \( y \mapsto y + 2t \nabla \log \replaced{f}{p}_{Y_t}(y) \).
		As \( t \) \replaced{approaches infinity}{$\to \infty$}, \( \replaced{f}{p}_{Y_t} \) becomes log-concave and \added{$-\log f_{Y_t}$} approaches a quadratic function.
}%
	\label{fig:diffusion toy example}
\end{figure*}

Recently,~\eqref{eq:tweedie} has been used for parameter estimation~\cite{song_scorebased_2021,vincent_connection_2011}:
Let \( \{ x_i \}_{i=1}^I \) be a dataset of \( I \)~samples drawn from \( f_X \) and let \( Y_t \) be governed by diffusion.
Additionally, let \( f_\theta : \mathcal{X} \times [0, \infty) \to \R_+ \) denote a parametrized model for which we wish that \( f_\theta(\argm, t) \approx f_{Y_t} \), for all \( t > 0 \).
Then, both the left- and right-hand side of~\eqref{eq:tweedie} are \emph{known} --- in expectation.
This naturally leads to the loss function
\begin{equation}
	\min_{\theta \in \Theta} \int_{(0, \infty)} \mathbb{E}_{(x, y_t) \sim f_{X, Y_t}} \norm{x - y_t - \sigma^2(t) \grad{1} \log f_\theta(y_t, t)}^2 \,\mathrm{d}t
\label{eq:score}
\end{equation}
for estimating \( \theta \) such that \( f_\theta(\argm, t) \approx f_{Y_t} \) for all \( t > 0 \).
Here, \( f_{X, Y_t} \) denotes the joint distribution of the clean and \replaced{smoothed}{corrupted} random variables and \( \Theta \) describes the set of feasible parameters.
This learning problem is known as denoising score matching~\cite{hyvarinen_estimation_nodate,song_scorebased_2021,vincent_connection_2011}.
\section{Methods}%
\label{sec:methods}
In this section, we first introduce one-dimensional \glspl{gmm} as the backbone of our model and recall some properties that are needed for the analysis in the following subsections.
Then, we detail \added{how we can utilize} \glspl{pogmdm} based on filter\added{-}, wavelet\added{-}, and shearlet\added{-}responses to model the distribution of natural images.
For all models, we present assumptions under which they obey the diffusion \replaced{\gls{pde}}{process}.

The backbone of our models is the one-dimensional \gls{gmm} expert~\( \psi_j : \R \times \triangle^L \times [0, \infty) \to \R_+ \) with \( L \) components of the form
\begin{equation}
	\psi_j(x,w_j,t) = \sum_{l=1}^{L} w_{jl} G_{\mu_l,\sigma_j^2(t)}(x).
	\label{eq:expert}
\end{equation}
The weights of each expert \( w_j = (w_{j1}, \dotsc, w_{jL})^\T \) must satisfy the unit simplex constraint, i.e., \( w_j \in \triangle^L \), \( \triangle^L = \{ x \in \R^L : x \geq 0, \langle \mathds{1}_{\R^L}, x \rangle_{\R^L} = 1 \} \).
Although not necessary, we assume for simplicity that all experts \( \psi_j \) have the same number of components and the discretization of their means \( \mu_l \) over the real line is shared and fixed a priori (for details see~\cref{ssec:implementation details}).

\replaced{%
	The main contribution of our work is to show that, under certain assumptions, it suffices to adapt the variances of the individual experts to implement the diffusion of a model built through multiplying experts of the form~\eqref{eq:expert}.
	In detail, we show that the variance \( \sigma_j^2 : [0, \infty) \to \R_+ \) of the \( j \)-th expert can be modeled as
}{%
	Further, the variances of all components within each expert are shared and are modeled as
}
\begin{equation}
	\sigma_j^2(t) = \sigma_0^2 + c_j 2t,
	\label{eq:variances}
\end{equation}
where \( \sigma_0 \added{> 0} \) is chosen \added{a-priori} to support the uniform discretization of the means~\( \mu_l \) and \( c_j \in \R_{++} \) are derived from properties of the product model such that \replaced{it}{the product model} obeys the diffusion \gls{pde}~\eqref{eq:diff}.

In the following subsections we exploit two well-known properties of \glspl{gmm} to derive how to \replaced{adapt}{model} the variance \( \sigma_j^2(t) \) of each expert \replaced{with diffusion time \( t \), such that the product model obeys the diffusion \gls{pde}}{to obey the diffusion \gls{pde}}:
First, \added{up to normalization}, the product of \glspl{gmm} is again a \gls{gmm}, see e.g.~\cite{1591840}.
This allows us to work on highly expressive models that enable efficient \emph{evaluations} due to factorization.
Second, we use the fact that there exists an analytical solution to the diffusion \gls{pde} if \( f_X \) is a \gls{gmm}:
Green's function associated with the linear isotropic diffusion \gls{pde}~\eqref{eq:diff} is a Gaussian with isotropic covariance \( 2t \mathrm{Id} \).
Due to the linearity of the convolution, it suffices to analyse the convolution of individual components of the product model, which is just the convolution of two Gaussians.
Using previous notation, if \( X \) is a random variable with normal distribution \( \mathcal{N}(\mu_X, \Sigma_X) \), then \( Y_t \) follows the distribution \( \mathcal{N}(\mu_X, \Sigma_X + 2t\Id) \).
In particular, the mean remains unchanged and it \replaced{suffices to}{is sufficient to only} adapt the covariance matrix with the diffusion time.

\replaced{%
	In what follows, we discuss three product models whose one-dimensional \gls{gmm} experts act on filter-, wavelet-, and shearlet-responses.
	In particular, we present conditions under which the diffusion of the product model can be implemented by adapting the variances of the one-dimensional \gls{gmm} experts.
	For all three models, we give an analytic expression for the constants \( c_j \) in~\eqref{eq:variances}.
}{%
	In what follows, we describe product of one-dimensional Gaussian mixture experts for which the diffusion of the covariance matrix can be efficiently implemented.
	In particular, we present three models acting on filter-, wavelet-, and shearlet responses for which we derive conditions under which the diffusion is accounted for by adapting the variances of the one-dimensional Gaussian mixture experts~\eqref{eq:expert}.
}
\subsection{Patch Model}%
\label{ssec:patch model}
In this section, we approximate the distribution \deleted{of image} of image patches~\( p \in \R^a \) of size~\( a = b\times b \) by a product of \( J \in \mathbb{N} \) \gls{gmm} experts \replaced{acting on}{modeling the distribution of} filter\added{-}responses.
In detail, the \deleted{filter-based} model is of the form
\begin{equation}
	f^{\mathrm{filt}}_\theta(p, t) = Z(\{ k_j \}_{j=1}^J\added{, \sigma_0, t})^{-1}\prod_{j=1}^J \psi_j(\langle k_j, p \rangle_{\R^a}, w_j, t).
	\label{eq:gmdm patch}
\end{equation}
Each \gls{gmm} expert~\( \psi_j : \R \times \triangle^L \times [0, \infty) \to \R_+ \) for \( j=1,\ldots,J \) models the distribution of filter\added{-}responses \( \mathbb{E}_{p_t\sim f_{Y_t}} \added{\bigl[} \delta (\argm - \langle k_j, p_t \rangle_{\R^a}) \added{\bigr]} \) of the associated filters~\( k_j \in \R^a \) for all~\( t > 0 \).
We denote with \( Z(\{ k_j \}_{j=1}^J\added{,\sigma_0, t}) \) the partition function such that \( f^{\mathrm{filt}}_\theta \) is properly normalized\deleted{, which we later show to be independent of \( t \)}.
In this model, we can summarize the learnable parameters \( \theta = \{ (k_j, w_j) \}_{j=1}^J \).

First, the following theorem establishes the exact form of~\eqref{eq:expert} as a \gls{gmm} on \( \R^a \).
\added{%
	The covariance matrix and the means are endowed with the subscript \( \R^a \) to emphasize that the resulting \gls{gmm} models patches of this size;
	The models based on wavelet- and shearlet-responses discussed later (\cref{ssec:wavelet model} and~\cref{ssec:conv model} respectively) can be applied to images of arbitrary size, which we emphasize by using the subscript \( \R^n \).
}
We denote with \( \hat{l} : \{ 1, \dotsc, J \} \to \{ 1, \dotsc, L \}  \) a fixed but arbitrary selection from the index set \( \{ 1, \dotsc, L \} \).
\begin{theorem}
	\( f^{\mathrm{filt}}_\theta(\argm, 0) \) is a homoscedastic \gls{gmm} on \( \R^a \) with \( L^J \) components and precision matrix
	\begin{equation}
			(\Sigma_{\R^a})^{-1} = \frac{1}{\sigma_0^2} \sum_{j=1}^J (k_j \otimes k_j).
			\label{eq:precision}
	\end{equation}
	The mean of the component identified by the choice of the index map \( \hat{l} \) has the form
	\begin{equation}
		\mu_{\R^a,\hat{l}} = \Sigma_{\R^a} \sum_{j=1}^J k_j\mu_{\hat{l}(j)}.
	\end{equation}%
	\label{th:gmm}
\end{theorem}
\begin{proof}
	By definition,
	\begin{equation}
		\begin{aligned}
			&\prod_{j=1}^J \psi_j(\langle k_j, p \rangle_{\R^a}, w_{j}, 0) \\
			&= \prod_{j=1}^J \sum_{l=1}^{L} \frac{w_{jl}}{\sqrt{2\pi\sigma_0^2}} \exp\left( -\frac{1}{2\sigma_0^2}{(\langle k_j, p \rangle_{\R^a} - \mu_l)}^2 \right).
		\end{aligned}
	\end{equation}
	The general component of the above is uniquely identified by the choice of the map \( \hat{l} \) as
	\begin{equation}
		\begin{aligned}
			&(2\pi\sigma_0^2)^{-\frac{J}{2}} \biggl( \prod_{j=1}^J w_{j\hat{l}(j)} \biggr) \\& \times  \exp\Biggl( -\frac{1}{2\sigma_0^2} \sum_{j=1}^J (\langle k_j, p \rangle_{\R^a} - \mu_{\hat{l}(j)})^2 \Biggr).
		\end{aligned}
	\end{equation}
	To find \( (\Sigma_{\R^a})^{-1} \), we \replaced{match the gradient of the familiar quadratic form}{complete the square}:
	Motivated by \( \grad{p} \norm{p - \mu_{\R^a,\hat{j}}}^2_{\Sigma_{\R^a}^{-1}} / 2 = \Sigma_{\R^a}^{-1} (p - \mu_{\R^a,\hat{l}}) \) we find that \( \grad{p} \bigl( \frac{1}{2\sigma_0^2} \sum_{j=1}^J (\langle k_j, p \rangle_{\R^a} - \mu_{\hat{l}(j)})^2 \bigr) = \frac{1}{\sigma_0^2} \sum_{j=1}^J \bigl( (k_j \otimes k_j) p - k_j \mu_{\hat{l}(j)}\bigr) \).
	From the first term, we immediately identify \( (\Sigma_{\R^a})^{-1} = \frac{1}{\sigma_0^2} \sum_{j=1}^J (k_j \otimes k_j) \), and we find \( \mu_{\R^a, \hat{l}} \) by left-multiplying \( \Sigma_{\R^a} \) onto \( \sum_{j=1}^J k_j \mu_{\hat{l}(j)} \).
\end{proof}
The next theorem establishes a tractable analytical expression for the diffusion process under the assumption of pair-wise orthogonal filters, that is
\begin{equation}
	\langle k_j, k_i \rangle = \begin{cases}
		0 & \text{if}\ i \neq j, \\
		\norm{k_j}^2 & \text{else},
	\end{cases} \text{ for all } i, j \in \{ 1,\dotsc,J \}.
\label{eq:ortho}
\end{equation}
\begin{theorem}[Patch diffusion]
	Under assumption~\eqref{eq:ortho}, \( f^{\mathrm{filt}}_\theta(\argm, t) \) satisfies the diffusion \gls{pde} \( (\partial_t - \Delta_1) f^{\mathrm{filt}}_\theta(\argm, t) = 0 \) if \( \sigma_j^2(t) = \sigma_0^2 + \norm{k_j}^2 2t \).
	\label{th:diff local}
\end{theorem}
\begin{proof}
	Assuming~\eqref{eq:ortho}, the Eigendecomposition of the precision matrix can be trivially constructed.
	In particular, \( (\Sigma_{\R^a})^{-1} = \sum_{j=1}^J \frac{\norm{k_j}^{2}}{\sigma_0^2} (\frac{k_j}{\norm{k_j}} \otimes \frac{k_j}{\norm{k_j}}) \), hence \( \Sigma_{\R^a} = \sum_{j=1}^J \frac{\sigma_0^2}{\norm{k_j}^{2}} (\frac{k_j}{\norm{k_j}} \otimes \frac{k_j}{\norm{k_j}}) \).
	As discussed in~\cref{ssec:diffusion empirical bayes}, \( \Sigma_{\R^a} \) evolves as \( \Sigma_{\R^a} \mapsto \Sigma_{\R^a} + 2t\mathrm{Id}_{\R^a} \) under diffusion.
	Equivalently, on the level of Eigenvalues, \( \frac{\sigma_0^2}{\norm{k_j}^{2}} \mapsto \frac{\sigma_0^2 + 2t\norm{k_j}^{2}}{\norm{k_j}^{2}} \) for all \( j = 1, \ldots, J \) .
	Recall that \( \sigma_0^2 \) is just \( \sigma_j^2(0) \).
	Thus, \( f^{\mathrm{filt}}_\theta(\argm,t) \) satisfies the diffusion \gls{pde} if \( \sigma_j^2(t) = \sigma_0^2 + \norm{k_j}^2 2t \).
\end{proof}
\begin{corollary}
	With assumption~\eqref{eq:ortho}, the \replaced{experts}{the potential functions} \( \psi_j(\argm, w_j, t) \) in~\eqref{eq:gmdm patch} model the marginal distribution of the random variable \( U_{j, t} = \langle k_j, Y_t \rangle \).
	In addition, \(\replaced{f^{\text{filt}}_{\theta}}{\tilde{f}_\theta} \) is normalized when \( Z(\{ k_j \}_{j=1}^J\added{, \sigma_0, t}) = \Bigl((2\pi)^a \prod_{j=1}^J \frac{\sigma_0^2 + 2t\norm{k_j}^2}{\norm{k_j}^2}\Bigr)^{\frac{1}{2}} \).
	\label{cor:marginal}
\end{corollary}
\begin{proof}
	\added{We first show that \( \psi_j(\argm, w_j, t) \) model the marginal distribution of the random variable \( U_{j, t} = \langle k_j, Y_t \rangle \).}
	Consider one component of the resulting homoscedastic \gls{gmm}: \( \hat{Y}_t \sim \mathcal{N}(\mu_{\R^a,\hat{l}}, \Sigma_{\R^a} + 2t\mathrm{Id}_{\R^a}) \).
	The distribution of \( \hat{U}_{j, t} = \langle k_j, \hat{Y}_t \rangle \) is  \( \hat{U}_{j, t} \sim \mathcal{N}(k_j^\top \mu_{\R^a,\hat{l}}, k_j^\top (\Sigma_a + 2t\mathrm{Id}_a) k_j) \) (see e.g.\ \cite{Gut2009} for a proof).
	Under our orthogonality assumptions, this simplifies to \( \mathcal{N}(\mu_{\hat{l}(j)}, \sigma_0^2  + 2t\norm{k_j}^2) \).
	The claim follows from the linear combination of the different components.

	\added{%
		The normalization constant is the classical normalization of a Gaussian, which requires the pseudo-determinant of \( \Sigma_{\R^a} \)~\cite{Rao1973}.
		The pseudo-determinant is easily calculated by the product of the Eigenvalues outlined in~\cref{th:diff local}.
	}
\end{proof}

\subsection{Wavelet Model}%
\label{ssec:wavelet model}
The key ingredient in the previous section was the orthogonality of the filters.
In other words, the filter bank \( \{ k_j \}_{j=1}^J \) forms an orthogonal (not necessarily orthonormal) basis for (a subspace of) \( \R^a \).
In this section, we discuss the application of explicit diffusion models in another well-known orthogonal basis: Wavelets.
In what follows, we briefly discuss the main concepts of the discrete wavelet transformation needed for our purposes.
For the sake of simplicity, we stick to the one-dimensional case but note that the extension to two dimensions is straight forward, see e.g.~\cite[Chapter 4.4]{Bredies2018}.
The following is largely adapted from~\cite{Bredies2018}, we refer the reader to this and~\cite{mallat_multiresolution_1989,vetterli1995wavelets} for information on the extension to two-dimensional signals as well as efficient implementations using the fast wavelet transformation.
\subsubsection{The discrete wavelet transformation}
Let \( \omega \in L^2(\R) \) be a wavelet satisfying the admissibility condition
\begin{equation}
	0 < \int_0^\infty \frac{|(\mathcal{F}\omega)(\zeta)|^2}{\zeta}\,\mathrm{d}\zeta < \infty.
\end{equation}
The set of functions
\begin{equation}
	\left\{ \omega_{j, k} = 2^{-j/2} \omega(2^{-j}\argm - k) : j, k \in \mathbb{Z} \right\}
	\label{eq:l2 basis}
\end{equation}
forms an orthonormal basis of \( L^2(\R) \) under certain conditions that we now recall.
Let \( (V_j)_{j\in\mathbb{Z}} \) be a multiscale analysis with \emph{generator} or \emph{scaling} function \( \phi \in V_0 \), i.e. \( \{ T_k \phi : k \in \mathbb{Z} \} \) form an orthonormal basis of \( V_0 \) (\( T_k \) is a translation operator \( (T_k \phi)(x) = \phi(x + k) \)).
The scaling property
\begin{equation}
	u \in V_j \iff D_{1/2}u\in V_{j+1} \quad((D_s u)(x) = u(sx))
\end{equation}
of the multiscale analysis \( (V_j)_{j\in\mathbb{Z}} \) implies that the functions \( \phi_{\replaced{j}{h}, k} = 2^{-j/2}\phi(2^{-j}\argm - k) \), \( k \in \mathbb{Z} \) form an orthonormal basis of \( V_j \).
Further, the scaling property implies that \( \phi \in V_{-1} \) and since \( \phi_{-1, k} \) form an orthonormal basis of \( V_{-1} \), we have that
\begin{equation}
	\phi(x) = \sqrt{2}\sum_{k\in\mathbb{Z}} h_k \phi(2x - k)
\end{equation}
with \( h_k = \langle \phi, \phi_{-1, k} \rangle_{L^2(\R)} \).
We define the \emph{detail} or \emph{wavelet spaces} \( W_j \) as the orthogonal complements of the \emph{approximation spaces} \( V_j \) in \( V_{j-1} \), i.e.
\begin{equation}
	V_{j-1} = V_j \oplus W_j, \quad V_j \perp W_j.
	\label{eq:orthogonal spaces}
\end{equation}
From this follows that \( V_j = \bigoplus\limits_{m\geq j+1} W_m \) and due to the completeness of \( V_j \), that \( \bigoplus\limits_{m \in \mathbb{Z}} W_m = L^2(\R) \).
By the orthogonality, we have that \( \proj_{V_{j-1}} = \proj_{V_j} + \proj_{W_j} \) and hence \( \proj_{W_j} = \proj_{V_{j-1}} - \proj_{V_j} \).
Thus, any \( u \in L^2(\R) \) can be represented as
\begin{equation}
	u = \sum_{j\in \mathbb{Z}} \proj_{W_j} u = \proj_{V_m} u + \sum_{j\leq m}\proj_{W_j} u
\end{equation}
justifying the name multiscale analysis.
Then (see \cite[Theorem 4.67]{Bredies2018} for details) \( \omega \in V_{-1} \) defined by
\begin{equation}
	\omega(x) = \sqrt{2}\sum_{k\in\mathbb{Z}} (-1)^{k} h_{1-k} \phi(2x-k)
\end{equation}
is a wavelet, \( \{ \omega_{j,k} : k \in \mathbb{Z} \} \) is an orthonormal basis of \( W_j \) and in particular the construction~\eqref{eq:l2 basis} is an orthonormal basis of \( L^2(\R) \).

\subsubsection{Modeling Wavelet Coefficients}
In this section we describe how we can utilize a product of \gls{gmm} experts to model the distribution of wavelet-responses.
For the subsequent analysis, first observe that by~\eqref{eq:orthogonal spaces} the detail spaces (and the approximation spaces) are orthogonal.
Utilizing the shorthand notation
\begin{equation}
	\mathcal{W}_j = \proj_{W_j},
\end{equation}
since \( \mathcal{W}_j \) is an orthogonal projection, it satisfies the properties
\begin{alignat}{3}
	&\text{(self-adjoint)} &&(\mathcal{W}_j)^\ast &&= \mathcal{W}_j, \nonumber \\
	&\text{(idempotency)} &&\mathcal{W}_j \circ \mathcal{W}_j &&= \mathcal{W}_j,\ \text{and} \label{eq:projection}\\
	&\text{(identity on subspace)}\ \ \ &&\mathcal{W}_j|_{W_j} &&= \mathrm{Id}_{W_j}\nonumber 
\end{alignat}
where \( \mathcal{W}_j|_{W_j} \) denotes the restriction of \( \mathcal{W}_j \) to \( W_j \).

As in the previous section, we model the \replaced{wavelet-responses}{responses of the wavelet filters} with Gaussian mixture \replaced{experts}{activations}.
In detail, let \( x \) be a signal in \( \R^n \) and thus, \( \mathcal{W}_j : \R^n \to \R^n \).
Then, the model reads\footnote{For simplicity, we discard the partition function \( Z \)\deleted{, which is independent of \( t \)}.}
\begin{equation}
	f^{\mathrm{wave}}_\theta(x, t) \propto \prod_{j=1}^{J} \prod_{i=1}^{n} \psi_j ((\mathcal{W}_j x)_{i}, w_j, t).
	\label{eq:wavelet gmm}
\end{equation}

Following the approach utilized in~\cref{th:diff local}, we first describe the exact form of~\eqref{eq:wavelet gmm} as a \gls{gmm} on \( \R^n \).
We denote with \( \hat{l} : \{1, \dotsc, n \} \to \{ 1, \dotsc, L \} \) a fixed but arbitrary selection from the index set \( \{ 1, \dotsc, L \} \).
In addition, the notation \( \sum_{\hat{l} = 1}^{L^n} \) indicates the summation over all \( L^n \) possible selections and for the following proof we define \( \mu_{\R^n}(\hat{l}) \coloneqq (\mu_{\hat{l}(1)}, \mu_{\hat{l}(2)},\dotsc,\mu_{\hat{l}(n)})^\top \in \R^n \).
\begin{theorem}
	\( f^{\mathrm{wave}}_\theta(\argm, t) \) is a homoscedastic \gls{gmm} on \( \R^n \) with precision matrix
	\begin{equation}
		(\Sigma_{\R^n})^{-1} = \frac{1}{\sigma_0^2} \sum_{j=1}^J \mathcal{W}_j.
	\end{equation}
\end{theorem}
\begin{proof}
By definition for~\( t = 0\), we have
\begin{equation}
	\begin{aligned}
		&f^{\mathrm{wave}}_\theta(x, 0) \\
		&\propto \prod_{i=1}^n \prod_{j=1}^J \sum_{l=1}^{L} \frac{w_{jl}}{\sqrt{2\pi\sigma_0^2}} \exp\left(-\frac{((\mathcal{W}_j x)_i - \mu_l)^2}{2\sigma_0^2}\right).
	\end{aligned}
	\label{eq:conv gmm}
\end{equation}
First, we expand the product over the pixels
\begin{equation}
	\begin{aligned}
		&f^{\mathrm{wave}}_\theta(x, 0) \propto \prod_{j=1}^J \sum_{\hat l = 1}^{L^n} (2\pi\sigma_0^2)^{-\frac{n}{2}} \overline{w}_{j\hat l} \\
		&\times\exp\left(-\frac{\norm{(\mathcal{W}_j x) - \mu_{\R^n}(\hat{l})}^2}{2\sigma_0^2}\right)
	\end{aligned}
\end{equation}
using the index map~\( \hat{l} \) and \(\overline{w}_{j\hat{l}} = \prod_{i=1}^I w_{j\hat{l}(i)} \).
Further, expanding over the features results in
\begin{equation}
	\begin{aligned}
		&f^{\mathrm{wave}}_\theta(x, 0) \propto \sum_{\hat\imath = 1}^{(L^n)^J}(2\pi\sigma_0^2)^{-\frac{nJ}{2}} \overline{\overline{w}}_{\hat\imath(i,j)} \\
		&\times\exp\left(-\frac{1}{2\sigma_0^2}\sum_{j=1}^J \norm{(\mathcal{W}_j x) - \mu_{\R^n,\hat{\imath}(i, j)}}^2\right),
	\end{aligned}
	\label{eq:expanded}
\end{equation}
where \( \overline{\overline{w}}_{\hat\imath(i,j)}=\prod_{j=1}^{J}\prod_{i=1}^{I} w_{\hat\imath(i,j)} \).
Notice that~\eqref{eq:expanded} describes a homoscedastic \gls{gmm} on \( \R^n \) with precision matrix
\begin{equation}
	(\Sigma_{\R^n})^{-1} = \frac{1}{\sigma_0^2} \sum_{j=1}^J \bigl(\mathcal{W}_j\bigr)^\ast \mathcal{W}_j.
	\label{eq:prec rn}
\end{equation}
Using the properties of a projection~\eqref{eq:projection} and the orthogonality property \( W_i \perp W_j \) for \( i \neq j \), this simplifies to
\begin{equation}
	(\Sigma_{\R^n})^{-1} = \frac{1}{\sigma_0^2} \sum_{j=1}^J \mathcal{W}_j.
\end{equation}
\end{proof}
\begin{theorem}[Wavelet diffusion]
	\( f^{\mathrm{wave}}_\theta(\argm, t) \) satisfies the diffusion \gls{pde} \( (\partial_t - \Delta_1) f^{\mathrm{wave}}_\theta(\argm, t) = 0 \) if \( \sigma_j^2(t) = \sigma_0^2 + 2t \).
	\label{th:wavelet diff}
\end{theorem}
\begin{proof}
	Notice that, using the properties of a projection~\eqref{eq:projection}, \( (\Sigma_{\R^n})^{-1}|_{\oplus_{j=1}^J W_j} = \frac{1}{\sigma^2}\mathrm{Id}_{\oplus_{j=1}^J W_j} \).
	Thus, on \( \bigoplus\limits_{j=1}^J W_j \), in analogy to the \replaced{model based on filter-responses}{patch model}, it suffices to adapt the variance of the one-dimensional \glspl{gmm} \( \psi_j \) with \( \sigma_0^2 \mapsto \sigma_0^2 + 2t \).
\end{proof}
We can endow the different \replaced{sub-bands}{levels} of the wavelet transformation with scalars to weight their influence as follows:
Replacing \( \mathcal{W}_j \) with \( \lambda_j \mathcal{W}_j \) in~\eqref{eq:prec rn} (the derivation does not change up to this point), we find that \( (\Sigma_{\R^n})^{-1} = \sum_{j=1}^J \frac{\lambda_j^2}{\sigma_0^2} \mathcal{W}_j \).
Then, the diffusion \gls{pde} is satisfied when \( \sigma_j^2(t) = \sigma_0^2 + 2t \lambda_j^2 \).
Thus, the scaling parameters \( \lambda_{\added{j}} \) are analogous to the filter-norm \( \norm{k_j} \) in~\cref{th:diff local}.

We briefly discuss the extension to two-dimensional signals:
Let \( x \in \R^{n \times n} \) be a two-dimensional signal.
\( \mathcal{W}_j^{d} : \R^{n \times n} \to \R^{n \times n} \) is a linear operator corresponding to the \( j \)-th detail level (\( j \in \{ 1, \dotsc, J \} \) where \( J \in \mathbb{N} \) is the coarsest scale in the decomposition) in the wavelet decomposition in direction \( d \).
We denote the (now three) \emph{detail spaces} at scale \( j \) as \( W_{j}^{d} \), where \( d \in \{ \mathbf{v}, \mathbf{h}, \mathbf{d} \} \) indexes the direction (\(\mathbf{v}\)ertical, \(\mathbf{h}\)orizontal, and \(\mathbf{d}\)iagonal.).
Our model accounts for the directional sub-bands with individual \gls{gmm} \replaced{experts}{potentials}, i.e.\ every \( \psi_j \) is replaced by a triplet \( \psi_j^{d} \) endowed with weights \( w_j^d \in \triangle^{n_w} \) for \( d \in \{ \mathbf{v}, \mathbf{h}, \mathbf{d} \} \).
Then, the entire previous discussion holds, where in particular \( W_i^{d} \perp W_j^{\tilde{d}} \) for all \( d, \tilde{d} \in \{ \mathbf{v}, \mathbf{h}, \mathbf{d} \} \) and all \( i \neq j \).
Since the operators \( \mathcal{W}_j^d \) are derived from generating sequence \( h \in \R^k \) (see~\cref{sssec:learning wavelets}), the learnable parameters are summarized as \( \theta = \{ h, \{\lambda_j^d\}_{j,d}, \{ w_j^d \}_{j,d} \} \).
\subsubsection{Interpretation as Diffusion Wavelet Shrinkage}
Wavelet shrinkage is a popular class of denoising algorithms.
Starting from the seminal work of~\cite{donoho_ideal_1994,donoho_adapting_1995,donoho_denoising_1995}, a vast literature is dedicated to finding optimal shrinkage parameters for wavelet-based denoising (see, e.g.\ \cite{simoncelli_noise_1996,chambolle_nonlinear_1998,chipman_adaptive_1997,clyde_multiple_1998,crouse_wavelet-based_1998,JANSEN199733} and the references therein).
In what follows, we briefly describe \deleted{the wavelet transformation and} historical approaches to estimating shrinkage \replaced{parameters}{coefficients}.
\deleted{%
	For a more detailed discussion about the two-dimensional wavelet transformation, we refer the reader to.
}

\replaced{%
	The key motivation behind wavelet shrinkage denoising algorithms is the observation that wavelet coefficients of natural images are sparse, wheres the wavelet coefficients of noisy images are densely filled with \enquote{small} values.
	Thus, a straight forward denoising algorithm might be to calculate the wavelet coefficients, \enquote{shrink} small coefficients towards zero, and calculate the inverse wavelet transform of the shrank coefficients.%
}
{%
	The denoised signal is given by \( \hat{x} = \mathcal{W}^\ast \mathcal{T}_\tau \mathcal{W} x \) where \( \mathcal{T}_\tau : \R^{n \times n} \to \R^{n \times n}\) is a shrinkage operator with thresholding parameter \( \tau > 0 \).
	The observation that the wavelet coefficients \( y = \mathcal{W}x \) are sparse for noise-free natural images motivates wavelet thresholding:%
}
Popular shrinkage operators include the soft-shrinkage \( x \mapsto \operatorname{sgn}(x) \max \{ |x| - \tau, 0 \} \) and the hard-shrinkage \( x \mapsto x \replaced{\chi}{\mathds{1}}_{\{|x| > \tau\}} \).
It is easy to see that these operators promote sparsity in the wavelet coefficients, as they correspond to the proximal maps w.r.t.\ \( \tau\norm{\argm}_1 \) and \( \tau\norm{\argm}_{0} \) respectively.
\added{%
	Here, \( \tau > 0 \) is a thresholding parameter that has to be chosen depending on the noise level.%
}

Historically, research for wavelet shrinkage models has focused on finding the optimal shrinkage parameter \( \tau \) (w.r.t.\ some risk, e.g.\ the squared error), assuming a particular choice of the shrinkage operator (e.g.\ the soft-shrinkage).
Popular selection methods include \emph{VisuShrink}~\cite{donoho_ideal_1994} and \emph{SureShrink}~\cite{donoho_adapting_1995}.
The former is signal independent and the threshold is essentially determined by the dimensionality of the signal as well as the (assumed known) noise \replaced{level}{variance}.
In contrast, the latter chooses the thresholding parameter depending on the energy in a particular sub-band and does not depend on the dimensionality of the signal explicitly.
The \emph{BayesShrink}~\cite{chang_adaptive_2000} method is also sub-band adaptive, and the authors provide expressions (or at least good approximations) for the optimal thresholding parameter under a generalized Gaussian prior on the wavelet coefficients.
In particular, they rely on classical noise \added{level} estimation techniques to fit the generalized Gaussian to the wavelet coefficients (of the noisy image) and arrive at a simple expression for a sub-band dependent threshold.

The general methodology outlined in the previous section allows us to take a different approach:
Instead of fixing the thresholding function and estimating the threshold solely on the corrupted image, we instead propose to learn the distribution of wavelet coefficients in different sub-bands for all noise \replaced{levels}{scales} \( \sigma > 0 \).
Notice that an empirical Bayes step on the wavelet coefficients under our model corresponds to applying a point-wise non-linearity.

In contrast to the traditional wavelet shrinkage, our model does not prescribe a shrinkage function for which an optimal parameter has to be estimated for different noise levels.
Rather, by learning the distribution of the wavelet coefficients at \enquote{all} noise \replaced{levels}{scales}, we have access to an \gls{mmse} optimal \enquote{shrinkage} function view of the empirical Bayes step on the \replaced{experts}{activation}.
In addition, our wavelet prior can be used in more general inverse problems whereas classical shrinkage methods are only applicable to denoising (although the denoising engine could be used in regularization by denoising~\cite{romano_little_2017} or plug-and-play~\cite{venkatakrishnan_plug-and-play_2013} approaches).

\subsection{Convolutional Model}%
\label{ssec:conv model}
The \replaced{model based on filter-responses discussed in~\cref{ssec:patch model}}{patch-based model~\eqref{eq:gmdm patch}} can not account for the correlation of overlapping patches when used for whole image restoration~\cite{zoran_learning_2011,RoBl09}.
Similarly, the \replaced{model based on wavelet-responses}{wavelet-based model} is limited in expressiveness since it only models the distribution of a scalar random variable per sub-band.
In what follows, we describe a convolutional \gls{pogmdm} that avoids the extraction and combination of patches in patch-based image priors and can account for the local nature of low-level image features.
The following analysis assumes vectorized images \( x \in \R^n \) with \( n \)~pixels; the generalization to higher dimensions is \replaced{straight forward}{straightforward}.

In analogy to the product-of-experts\replaced{-type model acting on filter-responses}{patch-based model}, here we extend the fields-of-experts model~\cite{RoBl09} to our considered diffusion setting by accounting for the diffusion time~\( t \) and obtain
\begin{equation}
	f_\theta^{\mathrm{conv}}(x, t) \propto \prod_{i=1}^n \prod_{j=1}^J \psi_j((K_j x)_i, w_{j}, t).
	\label{eq:gmdm}
\end{equation}
Here, each expert~\( \psi_j \) models the density of convolution \emph{features} extracted by convolution kernels~\( {\{ k_j \}}_{j=1}^J \) of size \( a = b \times b \), where
\( {\{ K_j \}}_{j=1}^J \subset \R^{n \times n}\) are the corresponding matrix representations.
Further, \( w_j \in \triangle^L \) are the \replaced{weights of the components of the \( j \)-th expert \( \psi_j \) (see~\eqref{eq:expert})}{weight the components of each expert~\( \psi_j \) as in~\eqref{eq:expert}}.
As \replaced{with the models based on filter- and wavelet-responses}{in the filter- and wavelet-based models}, it is sufficient to adapt the variances~$\sigma_j^2(t)$ by the diffusion time as the following analysis shows.

\eqref{eq:gmdm} again describes a homoscedastic \gls{gmm} on \( \R^n \) with precision \( (\Sigma_{\R^n})^{-1} = \frac{1}{\sigma_0^2} \sum_{j=1}^J K_j^\top K_j \).
This can be seen by essentially following the derivation of~\eqref{eq:prec rn} in~\cref{th:wavelet diff}, at which point we did not yet exploit the special structure of \( \mathcal{W}_j : \R^n \to \R^n \) (i.e.\ it may be an arbitrary linear operator).

In order to derive conditions under which~\eqref{eq:gmdm} fulfills the diffusion \gls{pde}, we begin by fixing the convolutions as cyclic, i.e., \( K_j x \equiv k_j *_n x \), where \( *_n \) denotes a 2-dimensional convolution with cyclic boundary conditions.
Due to the assumed boundary conditions, the Fourier transformation \( \mathcal{F} \) diagonalizes the convolution matrices: \( K_j = \mathcal{F}^* \diag(\mathcal{F}k_j) \mathcal{F} \).
Thus, the precision matrix can be expressed as
\begin{equation}
	(\Sigma_{\R^n})^{-1} = \mathcal{F}^*\diag\biggl(\sum_{j=1}^J \frac{|\mathcal{F}k_j|^2}{\sigma^2}\biggr) \mathcal{F}
	\label{eq:fourier diagonalization}
\end{equation}
where we used the fact that \( \mathcal{F}\mathcal{F}^* = \mathrm{Id}_{\R^n} \) and \( \replaced{\operatorname{conj}(z)}{\bar{z}}z = |z|^2 \) (here \( |\argm| \) denotes the complex modulus acting element-wise on its argument).
To get a tractable analytic expression for the variances \( \sigma_j^2(t) \), we further assume that the spectra of \( k_j \) have disjoint support, i.e.
\begin{equation}
	\Gamma_i \cap \Gamma_j = \emptyset\ \text{ if }\ i\neq j,
	\label{eq:disjoint}
\end{equation}
where \( \Gamma_j = \operatorname{supp} \mathcal{F}k_j \).
\deleted{%
	Note that, in analogy to the pair-wise orthogonality of the filters in the patch model~\eqref{eq:ortho}, from this it immediately follows that \( \langle \mathcal{F}k_j, \mathcal{F}k_i \rangle_{\added{\mathbb{C}^n}} = 0 \) when \( i \neq j \).
	Thus, in some sense, the convolutional model becomes a patch-based model in Fourier space.
}
In addition, we assume that the magnitude is constant over the support, i.e.
\begin{equation}
	|\mathcal{F} k_j| = \xi_j \replaced{\chi_{\Gamma_j}}{\mathds{1}_{\Gamma_j}},
	\label{eq:constant}
\end{equation}
where \added{$\xi_j \in \R$ is the magnitude and} \( \replaced{\chi_A}{\mathds{1}_A} \) is the characteristic function of the set \( A \)\added{
	\begin{equation}
		\chi_A(x) = \begin{cases}
			1 & \text{if}\ x \in A,\\
			0 & \text{else}.
		\end{cases}
	\end{equation}
}
\begin{theorem}[Convolutional Diffusion]
	Under assumptions~\eqref{eq:disjoint} and~\eqref{eq:constant}, \( f^{\mathrm{conv}}_\theta(\argm, t) \) satisfies the diffusion \gls{pde} \( (\partial_t - \Delta_1) f^{\mathrm{conv}}_\theta(\argm, t) = 0 \) if \( \bar{\sigma}_j^2(t) = \sigma_0^2 + \xi_j^2 2t \).
\end{theorem}
\begin{proof}
	In analogy to~\cref{th:diff local}, with~\eqref{eq:fourier diagonalization} \( \mathcal{F}^*\diag\big(\sum_{j=1}^J \frac{\sigma^2}{|\mathcal{F}k_j|^2} \big) \mathcal{F} \mapsto \mathcal{F}^*\diag\left(\frac{\sigma^2 + 2t\sum_{j=1}^J|\mathcal{F}k_j|^2 }{\sum_{j=1}^J |\mathcal{F}k_j|^2}\right) \mathcal{F} \) under diffusion.
	The inner sum decomposes as
	\begin{equation}
		\frac{\sigma_0^2 + 2t \sum_{j=1}^J |\mathcal{F}k_j|^2}{\sum_{j=1}^J |\mathcal{F}k_j|^2} = \sum_{j=1}^J \frac{\sigma_0^2 + 2t |\mathcal{F}k_j|^2}{|\mathcal{F}k_j|^2}
	\end{equation}
	using~\eqref{eq:disjoint}, and with~\eqref{eq:constant} the numerator reduces to \( \sigma_0^2 + 2t\xi_j^2 \).
\end{proof}

We emphasize that the convolutional model~\eqref{eq:gmdm} is distinctly different from the \replaced{model based on filter-responses discussed in~\cref{ssec:patch model}}{patch model~\eqref{eq:gmdm patch}}.
In particular, the one-dimensional \gls{gmm} experts \( \psi_j(\argm, w_j, t) \) do \emph{not} model the distribution of the filter\added{-}responses of their corresponding filter kernels \( k_j \), but account for the non-trivial correlation of overlapping patches.
We discuss this in more detail in~\cref{ssec:conv v patch}.
\subsubsection{Shearlets}
In the previous section, we derived abstract conditions under which a product of one-dimensional \gls{gmm} experts, with each expert modeling the distribution of convolutional features, can obey the diffusion \gls{pde}.
In particular, we derived that the spectra of the corresponding convolution filters must be non-overlapping and constant on their support.
Naturally, the question arises how such a filter bank can be constructed.
Luckily, the shearlet transformation~\cite{kutyniok_shearlets_2012} (and in particular the non-separable version of~\cite{lim_nonseparable_2013}) fulfills these conditions.
As an extension to the wavelet transformation, the shearlet transformation~\cite{kutyniok_shearlets_2012} can represent directional information in multidimensional signals via shearing.
Here, we consider the non-separable digital shearlet transformation~\cite{lim_nonseparable_2013}, whose induced frequency tiling is shown schematically in~\cref{fig:shearlet partitioning}.
In particular, the frequency plane is partitioned into non-overlapping cones indexed by the scaling and shearing parameters described in the next paragraph.
\begin{figure}
	\centering
	\begin{tikzpicture}
		\draw[->] (-3, 0) -- (3, 0);
		\draw[->] (0, -3) -- (0, 3);
		\fill[draw,gray, opacity=.5] (1, -.3) --  (1, .3) -- (2.5, .6) -- (2.5, -.6) -- cycle;
		\fill[draw,gray, opacity=.5] (-1, -.3) --  (-1, .3) -- (-2.5, .6) -- (-2.5, -.6) -- cycle;
		\fill[draw,gray!20!black, opacity=.5] (1, .3) --  (1, .9) -- (2.5, 1.8) -- (2.5, .6) -- cycle;
		\fill[draw,gray!20!black, opacity=.5] (-1, -.3) --  (-1, -.9) -- (-2.5, -1.8) -- (-2.5, -.6) -- cycle;
	\end{tikzpicture}
	\caption{%
		Schematic illustration of the cone-like frequency tiling of the non-separable shearlet transformation~\cite{lim_nonseparable_2013}.
		The shearlets are constructed such that their spectra are non-overlapping.
	}%
	\label{fig:shearlet partitioning}
\end{figure}
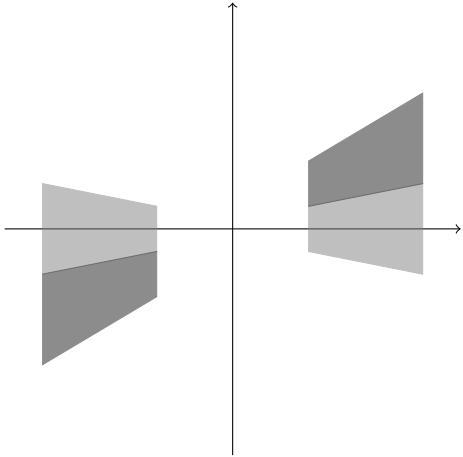

We briefly describe our setup but refer the interested reader to~\cite{lim_nonseparable_2013,kutyniok_shearlab_2016} for more details.
We construct a digital shearlet system, specified by the positive scaling integer \( j = 0,\dotsc,J \) the translations \( m \in \mathbb{Z}^2 \) and shearing \( |k| \leq \lceil 2^{\lfloor \frac{j}{2} \rfloor}\rceil \).
The system is constructed by a one-dimensional low-pass filter \( h_1 \) and a two-dimensional directional filter \( P \).
Given one-dimensional filters \( h_{J-j/2} \) and \( g_{J-j} \)  derived from \( h_1 \) in a wavelet multiresolution analysis, let \( W_j = g_{J-j} \otimes h_{J-j/2} \) and let \( p_j \) be the Fourier coefficients of \( P \) at scaling level \( j \).
Then, the system is constructed by
\begin{equation}
	\gamma_{j,k} = \Bigl[ \Bigl( S_k\bigl( (p_j * W_j)_{\uparrow_{2^{j/2}}} *_1 h_{j/2} \bigr) \Bigr) *_1 \replaced{\overleftarrow{h}_{j/2}}{\bar{h}_{j/2}} \Bigr]_{\downarrow_{2^{j/2}}}.
\end{equation}
Here, \( \uparrow_{a} \) and \( \downarrow_{a} \) are \( a \)-fold up- and down-sampling operators, and \( \replaced{\overleftarrow{(\cdot)}}{\bar{(\cdot)}} \) indicates sequence reversal \( \replaced{\overleftarrow{(\cdot)}}{\bar{(\cdot)}}(n) = (\cdot)(-n) \), and \( S_k \) is a shearing operator.
The digital shearlet transformation of an image \( x \in \R^{n \times n} \) is then given by
\begin{equation}
	\lambda_{j,k} \operatorname{conj} (\gamma_{j, k}) * x
\end{equation}
where \( \lambda_{j,k} > 0 \) are learnable weights \added{that reflect} the importance of the respective scale and shear level.
The learnable weights \( \lambda_{j,k} > 0 \) are easily accounted for in the diffusion models by adapting \( \xi_{j, k} \) in~\eqref{eq:constant} (where we have swapped the index \( j \) for a two-index \( j, k \) to account for the scales and shearing levels).
Thus, we can summarize the learnable parameters for the \replaced{model based on shearlet-responses}{shearlet model} as \( \theta = \{ h_1, P, \{\lambda_{j,k}\}_{j,k}, \{w_{j,k}\}_{j,k} \} \).
\section{Numerical Results}%
\label{sec:numerics}
In this section, we first detail the setup for numerical optimization.
In particular, we discuss how we can learn the one-dimensional \gls{gmm} experts along with the corresponding transformation (filters, wavelets, and shearlets) jointly.
Then, we show results for denoising utilizing a simple one-step empirical Bayes scheme as well as denoising algorithms derived for classical diffusion models.
In addition, we show that we can use our models for noise \added{level} estimation and blind heterosceda\added{s}tic denoising, and exploit~\cref{cor:marginal} to derive a direct sampling scheme.
\subsection{Numerical Optimization}%
\label{ssec:implementation details}
For the numerical experiments, \( f_X \) reflects the distribution of rotated and flipped \( b \times b \) patches from the \num{400} gray-scale images in the BSDS 500~\cite{martin_database_2001} training and test set, with each pixel \added{value} in the interval \( [0, 1] \).
We optimize the score matching objective function~\eqref{eq:score} using projected AdaBelief~\cite{zhuang2020adabelief} for \( \num{100000} \) steps.
We approximate the infinite-time diffusion \replaced{\gls{pde}}{process} by uniformly drawing \( \sqrt{2t} \) from the interval \( [\num{0}, \num{0.4}] \).
For the denoising experiments, we utilize the validation images from~\cite{RoBl09} (also known as \enquote{Set68}).
Due to computational constraints, we utilize only the first \num{15} images of the dataset according to a lexicographic ordering of the filenames.
In addition, our wavelet- and shearlet-toolboxes only allow the processing of square images.
To avoid boundary artifacts arising through padding images to a square, we only utilize the central region of size \( \num{320} \times \num{320} \).

For all experiments, \( \psi_j \) is a \( L = \num{125} \) component \gls{gmm}, with equidistant means \( \mu_l\) in the interval \( [-\eta_j, \eta_j] \) (we discuss the choice of \( \eta_j \) for the different models in their respective sections).
To support the uniform discretization of the means, the standard deviation of the \( j \)-th experts is \( \sigma_{0} = \frac{2\eta_j}{L - 1} \).
In the one-dimensional \gls{gmm} backbone of all models, we have to project a weight vector \( w \in \R^L \) onto the unit simplex \( \triangle^L \).
We realize the projection \( \proj_{\triangle^{L}} : \R^L \to \R^L \) with the sorting-based method proposed by~\cite{Held1974}, which is summarized in~\cref{alg:simplex proj}.
In addition, we further assume that the one-dimensional \gls{gmm} experts are symmetric around \( 0 \), i.e.\ that the weights \( w_j \) are in the set \( \{ x \in \R^L : (x \in \triangle^L) \wedge (x = \replaced{\overleftarrow{x}}{\bar{x}}) \} \).
We implement by storing only \( \lceil L / 2 \rceil \) weights, and mirroring the tail of \( \lceil L / 2 \rceil - 1 \) elements prior to the projection algorithm and function evaluations.
To ensure that the one-dimensional \gls{gmm} experts are sufficiently peaky around zero, we always choose \( L \) to be odd.

In the next sections, we detail the constraints the building blocks of the learned transformations have to fulfill and how to satisfy them in practice.
\subsubsection{Learning Orthogonal Filters}
Let \( K = [k_1, k_2, \dotsc, k_J] \in \R^{a \times J} \) denote the matrix obtained by horizontally stacking the filters.
We are interested in finding
\begin{equation}
	\operatorname{proj}_{\mathcal{O}}(K) \replaced{=}{\in} \argmin_{M \in \mathcal{O}} \norm{M - K}_F^2
\end{equation}
where \( \mathcal{O} = \{ X \in \R^{a \times J} : X^\top X = D^2 \} \), \( D = \diag(\lambda_1,\lambda_2,\dotsc,\lambda_J) \) is diagonal, and \( \norm{\argm}_F \) is the Frobenius norm.
Since \( \operatorname{proj}_{\mathcal{O}}(K)^\top \operatorname{proj}_{\mathcal{O}}(K) = D^2 \) we can represent it as \( \operatorname{proj}_{\mathcal{O}}(K) = OD \) with \( O \) semi-unitary (\( O^\top O = \mathrm{Id}_{\R^J} \)).
Other than positivity, we do not place any restrictions on \( \lambda_1, \dotsc, \lambda_J \), as these are related to the precision in our model.
Thus, we rewrite the objective
\begin{equation}
	\proj_{\mathcal{O}}(K) \replaced{=}{\in} \argmin_{\substack{O^\top O = \mathrm{Id}_J \\ D = \diag(\lambda_1,\dotsc,\lambda_J)}} \mathcal{E}(O, D)
\end{equation}
where
\begin{equation}
	\mathcal{E}(O, D) \coloneqq \norm{OD - K}_F^2  = \norm{K}_F^2 - 2 \langle K, OD \rangle_F + \norm{D}_F^2,
\end{equation}
with \( \langle \argm, \argm \rangle_F \) denoting the Frobenius inner product.

We propose the following alternating minimization scheme for finding \( O \) and \( D \).
The solution for the reduced sub-problem in \( O \) can be computed by setting \( O = U \), using the polar decomposition of \( DK^\top = UP \), where \( U \in \R^{J \times a}\) is semi-unitary (\( U^\top U = \mathrm{Id}_{\R^a} \)) and \( P = P^\top \succeq 0 \).
The sub-problem in \( D \) is solved by setting \( D_{i,i} = \bigl((O^\top K)_{i,i}\bigr)_{+} \).
The algorithm is summarized in~\cref{alg:orthogonalizing}, where we have empirically observed fast convergence; \( B = 3 \) steps already yielded satisfactory results.
A preliminary theoretical analysis of the algorithm is presented in the supplementary material of our conference paper~\cite{zach_explicit_2023}.
\begin{algorithm}[t]
	\DontPrintSemicolon
	\SetKwInOut{Output}{Output}
	\SetKwInOut{Input}{Input}
	\Input{\( K = [k_1, \dotsc, k_J] \in \R^{a \times J} \), \( B \in \mathbb{N} \), \( D^{(1)} = \mathrm{Id}_J \)}
	\Output{\( O^{(B)}D^{(B)} = \proj_{\mathcal{O}}(K) \)}
	\For{\( b \in 1, \dotsc, B - 1 \)}{
		\( U^{(b)}P^{(b)} = D^{(b)}K^\top \)\tcp*{Polar decomposition}
		\( O^{(b+1)} = U^{(b)} \)\;
		\( D^{(b+1)}_{i,i} = \bigl(((O^{(b+1)})^\top K)_{i, i} \bigr)_+ \)\;
	}
	\caption{%
		Algorithm for orthogonalizing a set of filters \( K \).
	}%
	\label{alg:orthogonalizing}
\end{algorithm}

Assuming a patch size of \( a = b \times b \) we use \( J = b^2 - 1 \) filters spanning the space of zero-mean patches \( \mathfrak{Z} = \{ x \in \R^a : \langle \mathds{1}_{\R^a}, x \rangle_{\R^a} = 0 \} \).
We found that implementing \( \proj_{\mathcal{O} \cap \mathfrak{Z}} \) as \( \proj_{\mathfrak{Z}} \circ \proj_{\mathcal{O}} \), both constraints were always almost exactly fulfilled.
To ensure the correct projection, an alternative would be to utilize Dykstra's projection algorithm~\cite{Boyle1986}.
The filters are initialized by independently drawing their entries from a zero-mean Gaussian distribution with standard deviation \( b^{-1} \).
Since the filters can be freely scaled, we simply choose \( \eta_j = 1 \) for all \( j = 1, \dotsc, J \).

To visually evaluate whether our learned model matches the empirical marginal densities for any diffusion time \( t \), we plot them in~\cref{fig:patch results}.
At the top, the learned \( 7 \times 7 \) orthogonal filters \( k_j \) are depicted.
The filters bare striking similarity to the Eigenimages of the covariance matrices of~\cite[Fig. 6]{zoran_learning_2011}, who learn a \gls{gmm} directly on the space of image patches (i.e.\ without any factorizing structure).
This comes as no surprise, since the construction of the patch-model~\eqref{eq:gmdm patch} can be interpreted as \enquote{learning the Eigendecomposition}, see~\cref{th:gmm} and the proof of~\cref{th:diff local}.
The \deleted{associated} learned potential functions \( -\log \psi_j(\argm, w_j, t) \) and \replaced{activation functions \( -\nabla \log \psi_j(\argm, w_j, t) \)}{their gradient}\footnote{\added{%
		We refer to the (gradient of the) negative-log experts as \enquote{potential functions} (\enquote{activation functions}).
		This nomenclature is often used in the neural network literature, but clashes with what is typically used in the context of graphical models and Markov random fields.
}}%
\added{associated with the \( j \)-th filter} are shown below \added{the filters in~\cref{fig:patch results}}.
Indeed, the learned \replaced{potential functions}{potentials} match the \added{negative-log} empirical marginal \replaced{response histograms}{responses}
\begin{equation}
	\deleted{h_j(z, t) =} -\log \mathbb{E}_{p \sim \replaced{f_{Y_t}}{f_X}} \added{\bigl[} \delta(z - \langle k_j, p \rangle) \added{\bigr]}
\end{equation}
visualized at the bottom almost perfectly even at extremely low-density tails.
This supports the theoretical argument that \deleted{the} diffusion \deleted{process} eases the problem of density estimation outlined in the introductory sections.
\begin{figure*}
	\newlength{\hheight}
	\settoheight{\hheight}{A}
	\centering
	\includegraphics[trim=4.5cm .7cm 4cm .6cm, clip, width=\textwidth]{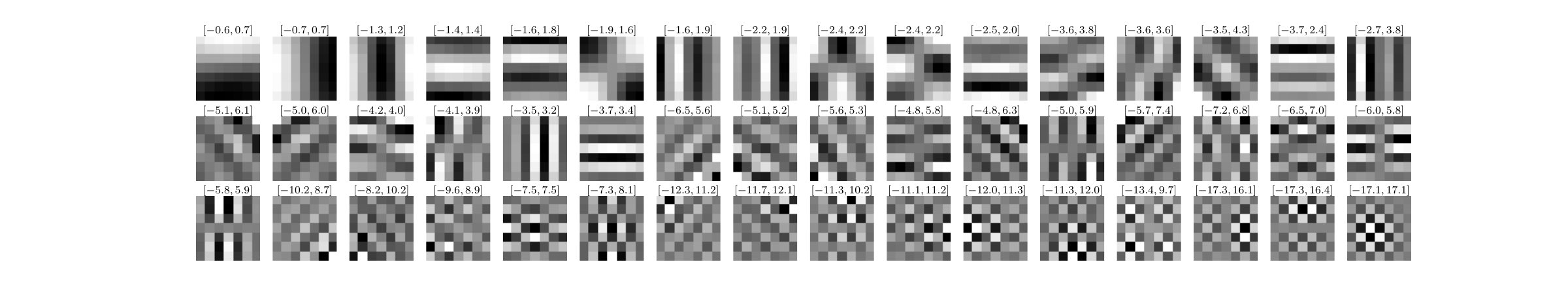}
	\includegraphics[trim=4.5cm .3cm 4cm .7cm, clip, width=\textwidth]{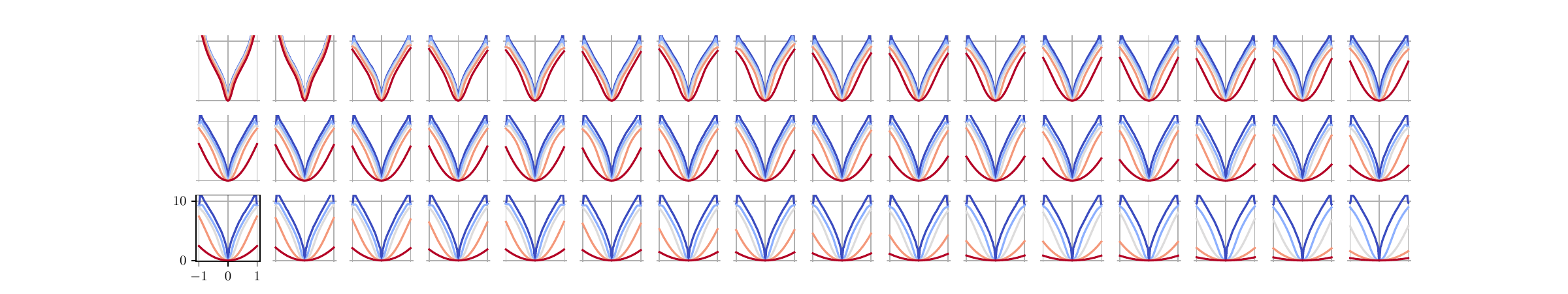}
	\includegraphics[trim=4.5cm .3cm 4cm .7cm, clip, width=\textwidth]{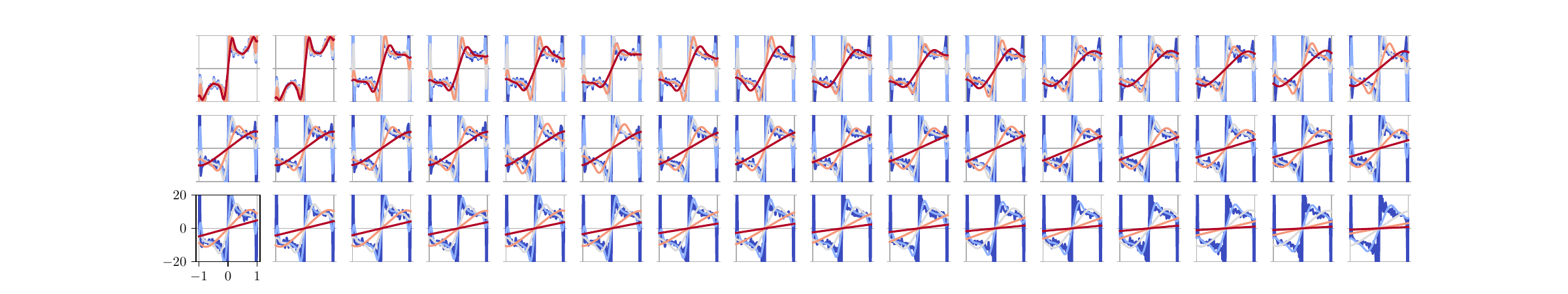}
	\rule[3mm]{\textwidth}{.3mm}\vspace*{-3mm}
	\vspace*{-3mm}
	\includegraphics[trim=4.5cm .3cm 4cm .7cm, clip, width=\textwidth]{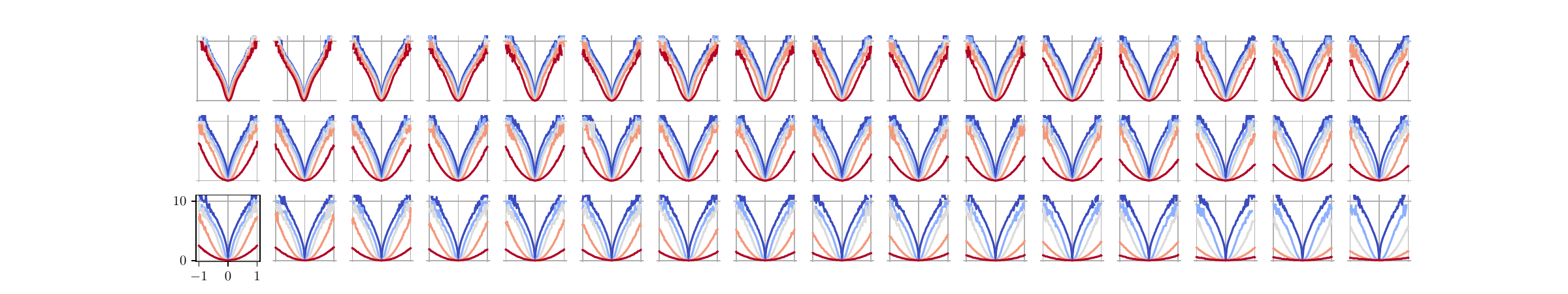}
	\vspace*{3mm}
	\hspace*{8.2cm}{\resizebox{6cm}{!}{\(\sqrt{2t}=\)
		\foreach \ccolor\ssigma in {coolwarm1/0, coolwarm2/0.025, coolwarm3/0.05, coolwarm4/0.1, coolwarm5/0.2}
		{
			\tikz[baseline=\hheight*0.8]{\draw[\ccolor, ultra thick](0,0.3) -- (.5,.3);}\num{\ssigma}
		}
	}}
	\caption{%
		Learned filters \( k_j \) (top, the intervals show the values of black and white respectively, amplified by a factor of \num{10}), potential functions \( -\log \psi_j(\argm, w_j, t) \) and activation functions \( -\nabla \log \psi_j(\argm, w_j, t) \).
		On the bottom, the \added{negative-log} empirical marginal \deleted{filter} response histograms are drawn.
	}%
	\label{fig:patch results}
\end{figure*}
\subsubsection{Learning Wavelets}%
\label{sssec:learning wavelets}
The discrete wavelet transformation is characterized by the sequence \( h \in \R^K \).
In addition to learning the parameters of the one-dimensional \gls{gmm}, we follow~\cite{grandits_optimizing_2018} and also learn \( h \).
From the sequence \( h \), the scaling-function \( \phi \) and wavelet-function \( \omega \) are defined by
\begin{equation}
	\phi(x) = \sum_{k=1}^K h_k \sqrt{2} \phi(2x - k)
\end{equation}
and
\begin{equation}
	\omega(x) = \sum_{k=1}^K\replaced{\bigl(g(h)\bigr)_k}{g_k}\sqrt{2}\phi(2x - k)
\end{equation}
where \( \replaced{\bigl(g(h)\bigr)_k}{g_k} = (-1)^kh_{K - k - 1} \).
For \( \omega \) to be a wavelet, it must follow the admissibility criterion
\begin{equation}
	\int_0^\infty \frac{|(\mathcal{F}\omega)(\zeta)|^2}{\zeta}\,\mathrm{d}\zeta < \infty,
\end{equation}
cf~\cite{mallat_multiresolution_1989}, from which it immediately follows that \( (\mathcal{F}\omega)(0) = \int_\R \omega = 0 \).
For the transformation to be unitary, we need that \( \int_\R \phi = 1 \), and
\begin{equation}
	\int_\R \phi(x)\phi(x-n)\,\mathrm{d} x = \delta_{n} \text{ for all } n \in \mathbb{Z}.
\end{equation}
From these constraints, the feasible set of wavelet-generating sequences is described by
\begin{equation}
	\begin{aligned}
		\Omega = \{ &h \in \R^K : \langle \mathds{1}_{\R^K}, g(h) \rangle_{\R^K} = 0,\\
					&\langle \mathds{1}_{\R^K}, h \rangle_{\R^K} = \sqrt{2},\\
					&\langle h, \circlearrowleft_{2n} h \rangle_{\R^K} = \delta_n\ \text{ for all } n \in \mathbb{Z} \}.
	\end{aligned}
\end{equation}
Here \( \circlearrowleft_{n} : \R^K \to \R^K \) rolls its argument by \( n \) entries, i.e. \( \circlearrowleft_{n} x = (x_{K-n+1}, x_{K-n+2},\dotsc,x_{K},x_1,x_2,\dotsc,x_{K-n})^\top \).
Observe that the orthogonality condition encodes \( K / 2 \) constraints (we assume that \( K \) is even), since \( \circlearrowleft_{0} = \circlearrowleft_{K} = \mathrm{Id}_{\R^K} \).
To project onto \( \Omega \), we write the projection problem
\begin{equation}
	\proj_{\Omega} (\bar{x}) \replaced{=}{\in} \argmin_{x \in \Omega} \frac{1}{2} \norm{x - \bar{x}}_2^2
\end{equation}
in its Lagragian form using \( \mathcal{L} : \R^K \times \R \times \R \times \R^{K/2} \to \R : \)
\begin{equation}
	\begin{aligned}
		&(x, \Lambda_{\mathrm{scal}}, \Lambda_{\mathrm{adm}}, \Lambda)\\
		&\mapsto \frac12 \norm{x - \bar{x}}_2^2\\
		&+ \Lambda_{\mathrm{scal}} \bigl( \langle \mathds{1}_{R^K}, h \rangle_{\R^K} - \sqrt{2} \bigr)
		+ \Lambda_{\mathrm{adm}} \bigl( \langle \mathds{1}_{\R^K}, g(h) \rangle_{\R^K} \bigr)\\
		&+ \sum_{n = 0}^{\frac{K}{2}-1} \Lambda_{n+1} \bigl( \langle h, \circlearrowleft_{2n} h \rangle_{\R^K} - \replaced{\chi_{\{0\}}(n)}{\delta_n} \bigr).
	\end{aligned}
\end{equation}
and find stationary points by solving the associated non-linear \replaced{least-squares }{least-squared} problem
\begin{equation}
	\norm{\nabla \mathcal{L}(x, \Lambda_{\mathrm{scal}}, \Lambda_{\mathrm{adm}}, \Lambda)}_2^2/2 = 0
\end{equation}
using \num{10} iterations of Gauss-Newton.
To facilitate convergence, we warm start the Lagrange multipliers \( \Lambda_{\mathrm{scal}}, \Lambda_{\mathrm{adm}}, \Lambda \) with the solution from the previous outer iteration.
We initialize the sequence \( h \) with the generating sequences of the \texttt{db2}- (\( K = 4 \)) and \texttt{db4}-wavelet (\( K = 8 \)).
For both, we utilize \( J = 2 \) levels.
We use the \texttt{pytorch\_wavelets}~\cite{cotter_complex_2020} implementation of the discrete wavelet transformation.

In contrast to the \replaced{model based on filter-responses}{patch model}, the \replaced{model based on wavelet-responses}{wavelet model} does not have the freedom to adapt the scaling of filters\deleted{that respond well to noise}.
To overcome this, we discretize the means over the real line individually for each sub-band.
In detail, for the \( j \)-th level and \( d \)-th direction, \( d \in \{ \mathbf{h}, \mathbf{v}, \mathbf{d} \} \), we choose \( \eta^d_j = 1.1 q^{d}_j \), where \( q^d_j \) is the \( .999 \)-quantile of corresponding responses.

The initial and learned generating sequences, their corresponding scaling- and wavelet-functions, along with the learned potential functions and \gls{mmse}-shrinkage are shown in~\cref{fig:wavelet pot}.
In these figures, it is apparent that our chosen parametrization is sub-optimal.
In particular, in order to represent the heavy tails (especially for level \( j = 1\)), many intermediate weights are set to \( 0 \).
This leads to the \gls{mmse} shrinkage functions becoming step-like.
We emphasize that this is a practical problem of choosing the appropriate parametrization; we discuss alternatives to our equispaced \gls{gmm} in~\cref{sec:discussion}.
\begin{figure*}
	\newlength{\hhheight}
	\settoheight{\hhheight}{A}
	\begin{tabular}{c|c}
		\includegraphics[width=.23\textwidth]{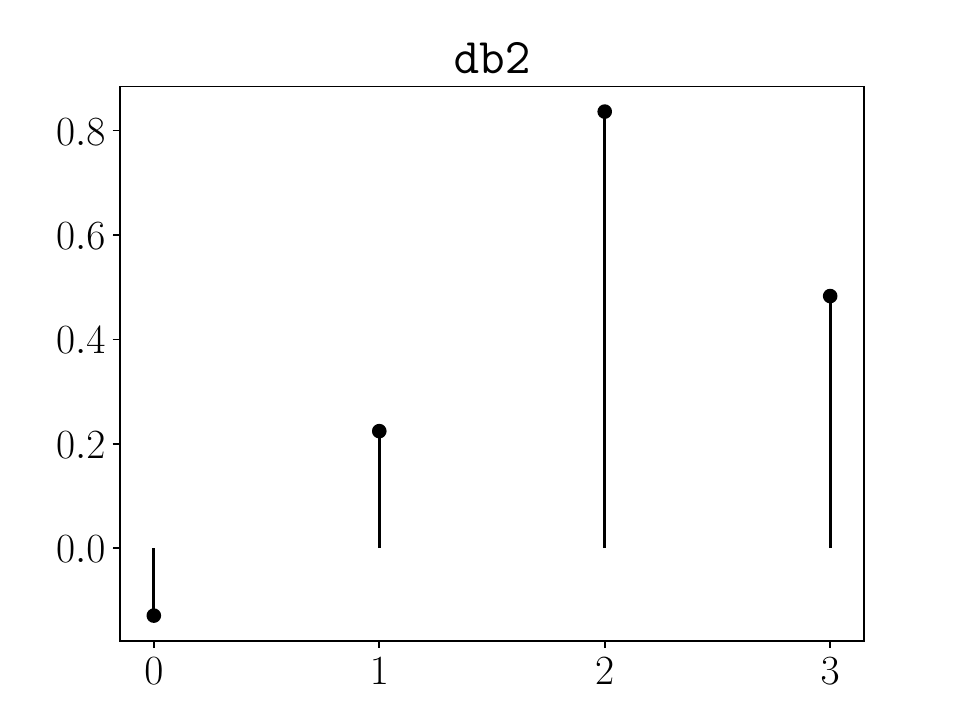}\includegraphics[width=.23\textwidth]{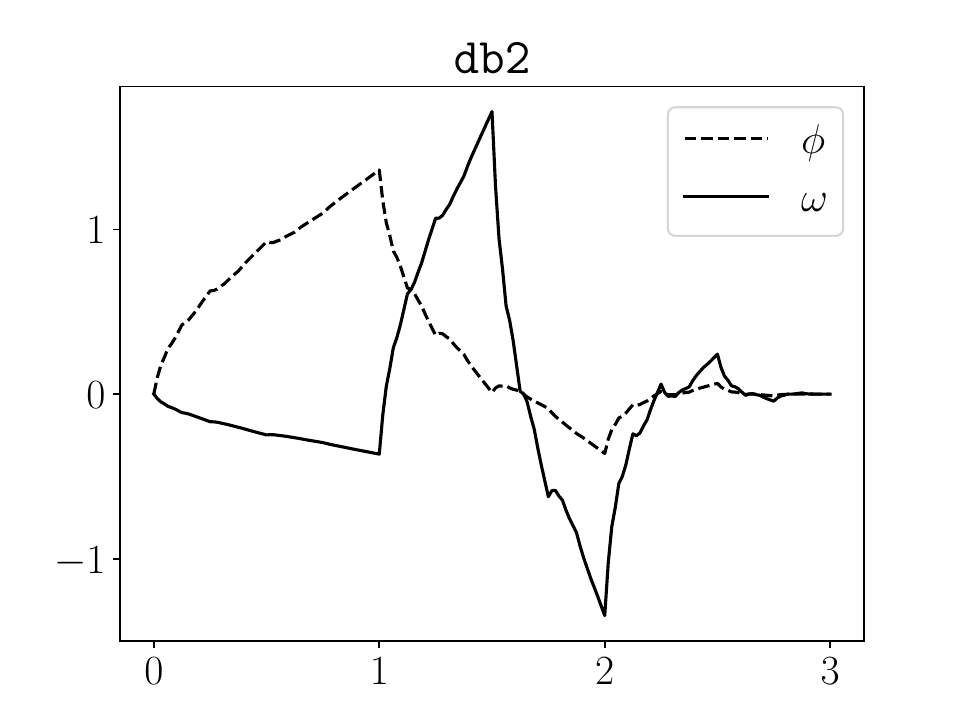} & \includegraphics[width=.23\textwidth]{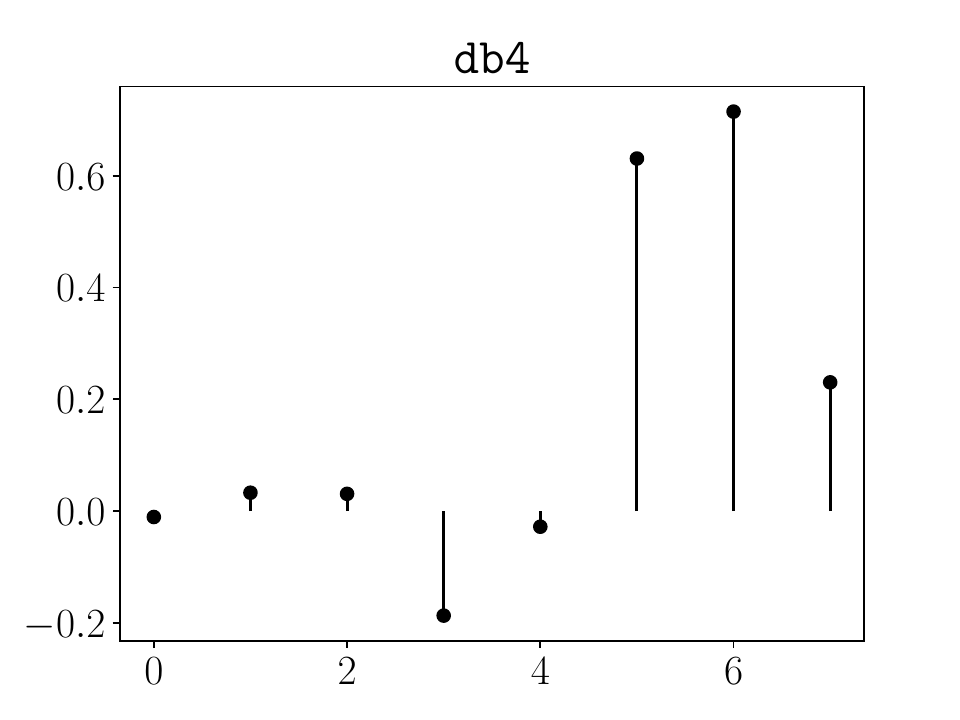} \includegraphics[width=.23\textwidth]{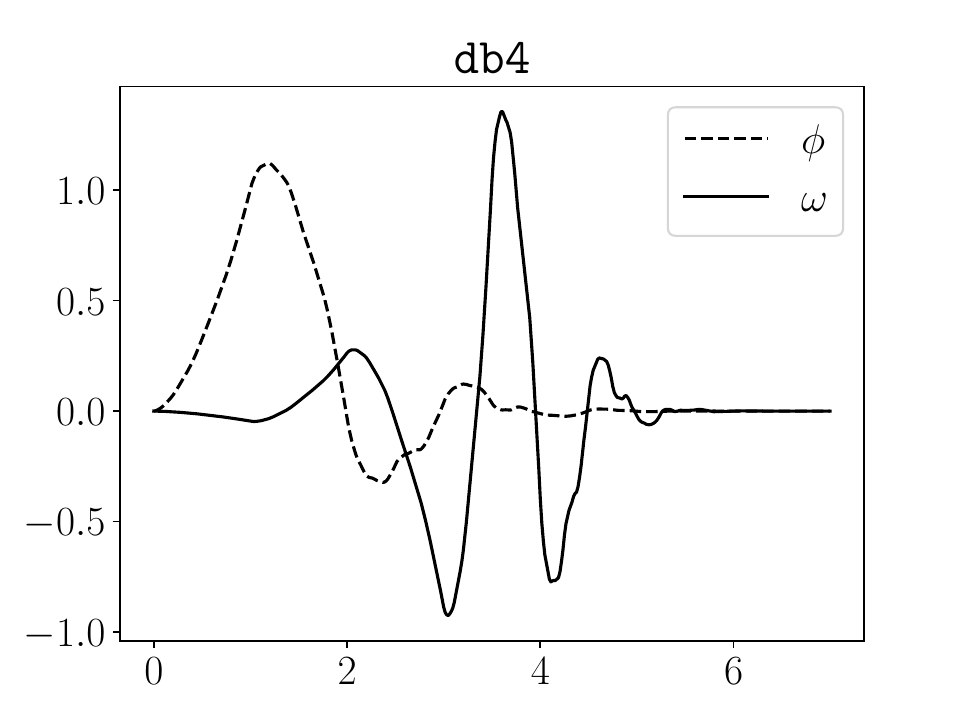} \\
		\includegraphics[width=.23\textwidth]{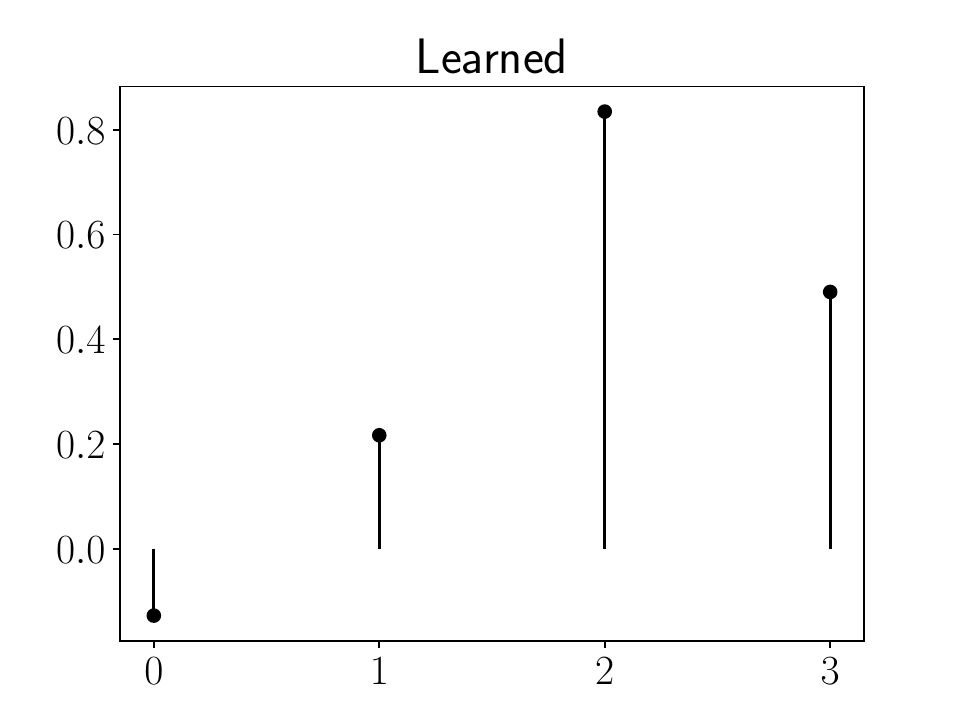} \includegraphics[width=.23\textwidth]{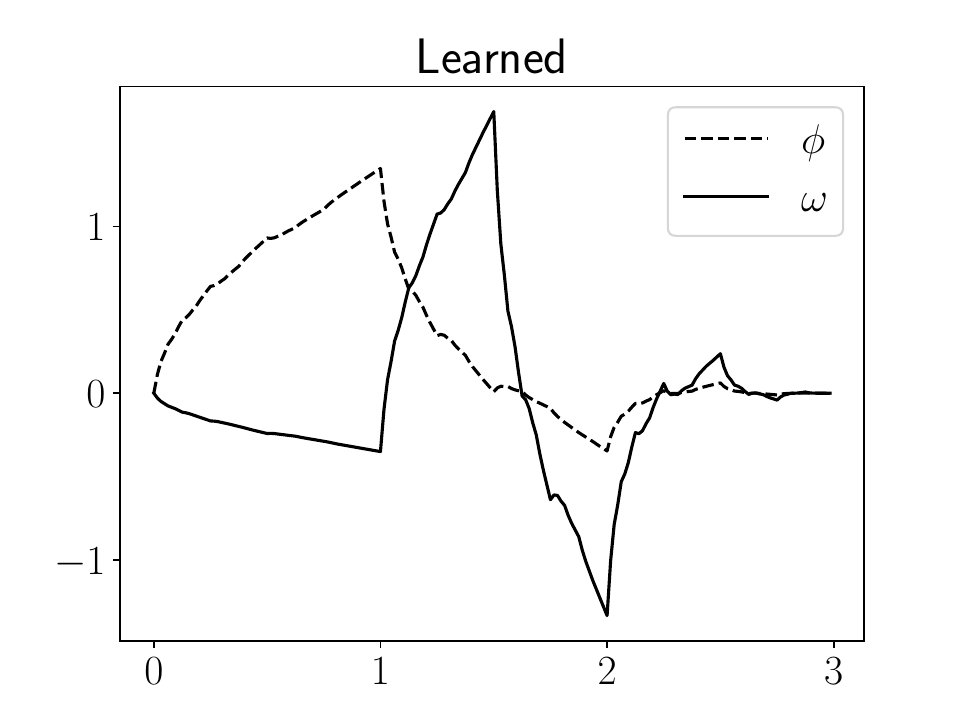} & \includegraphics[width=.23\textwidth]{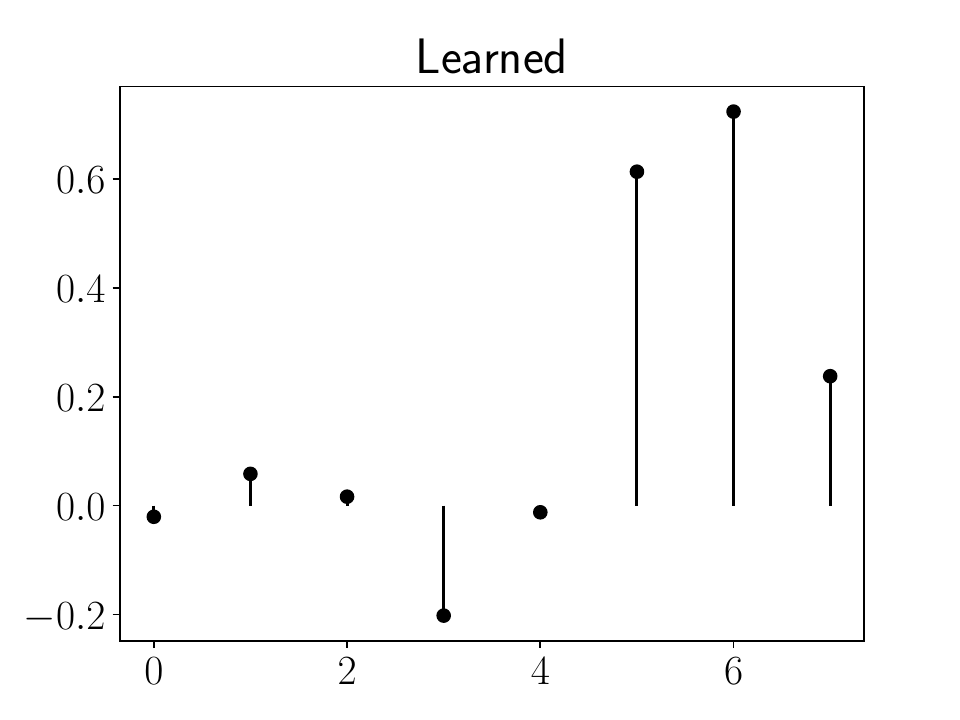} \includegraphics[width=.23\textwidth]{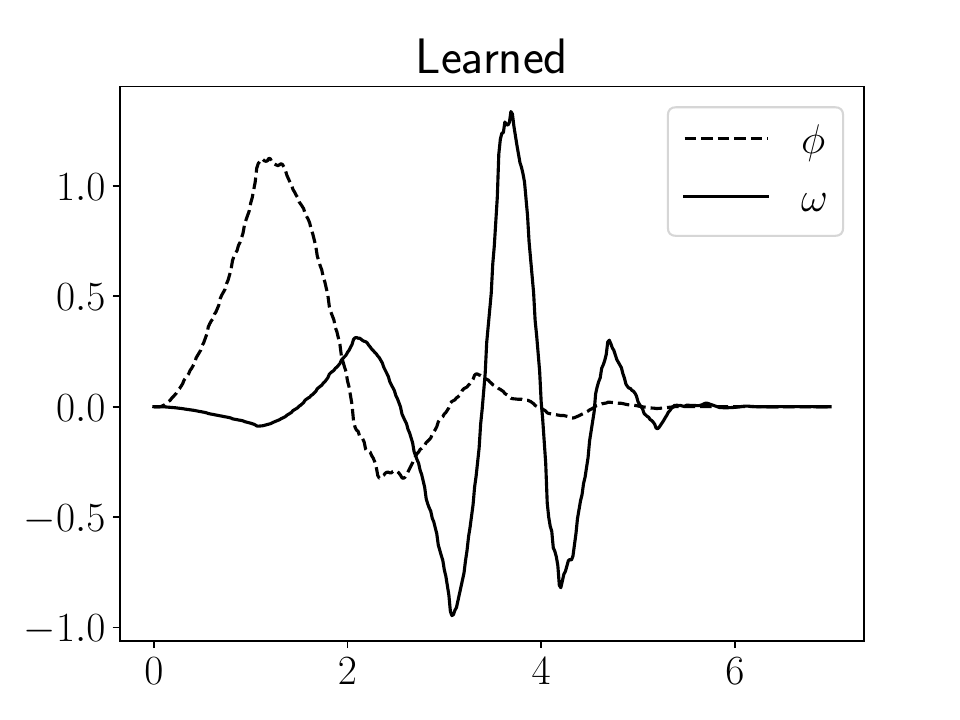} \\
		\includegraphics[width=.45\textwidth]{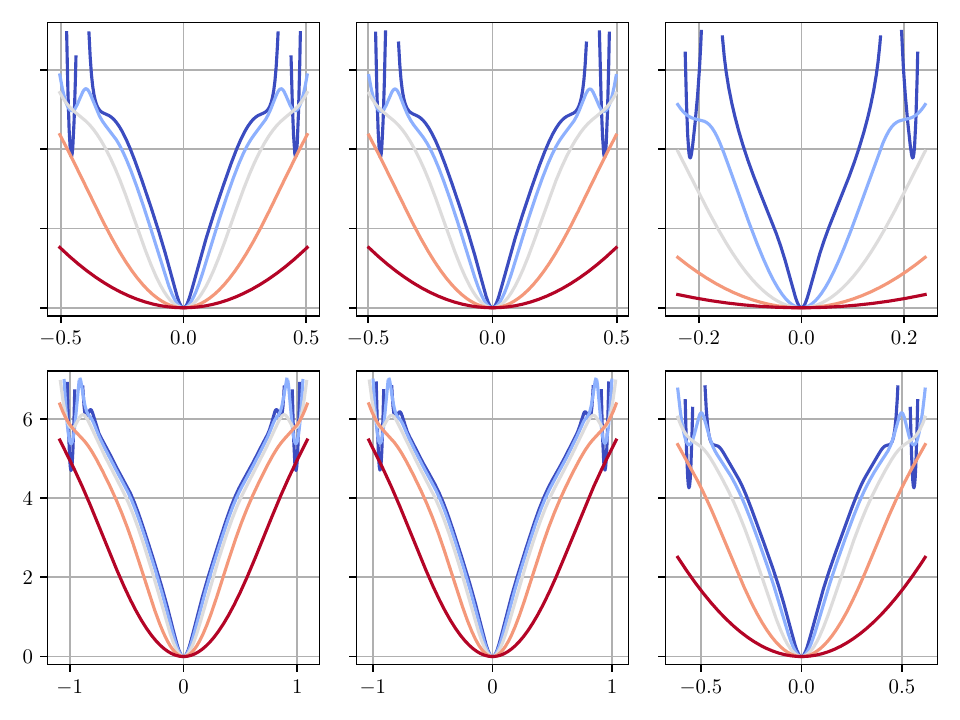} & \includegraphics[width=.45\textwidth]{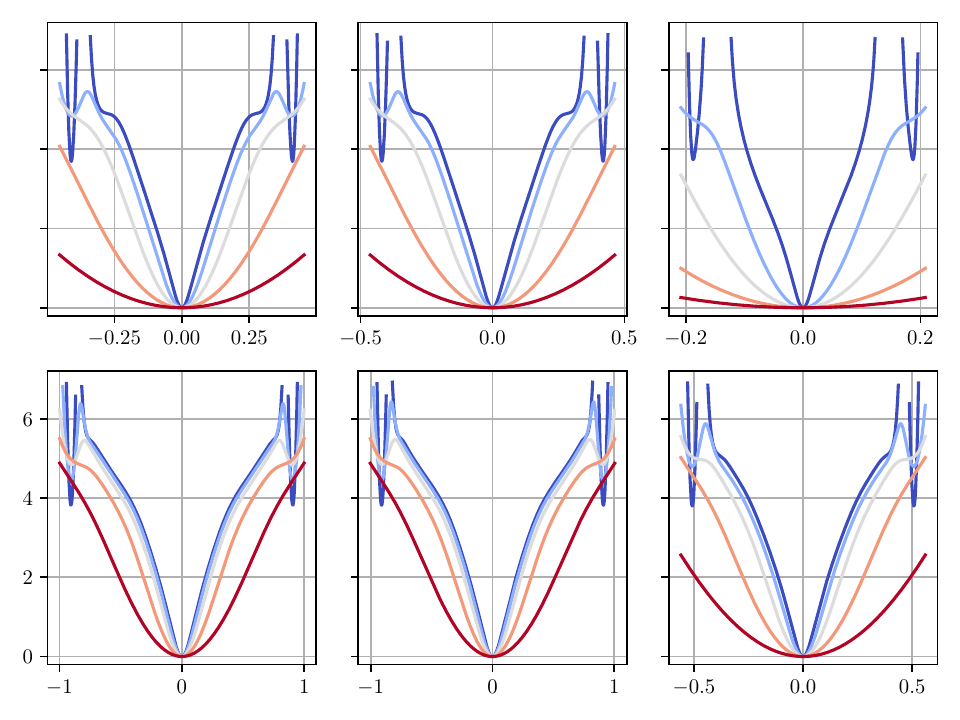} \\
		\includegraphics[width=.45\textwidth]{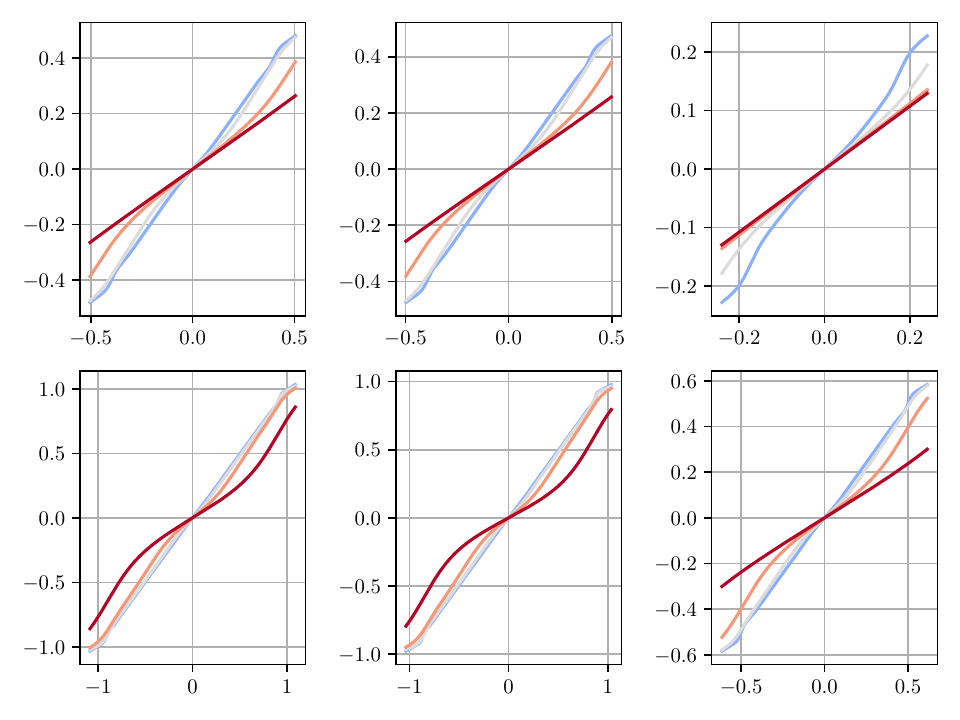} & \includegraphics[width=.45\textwidth]{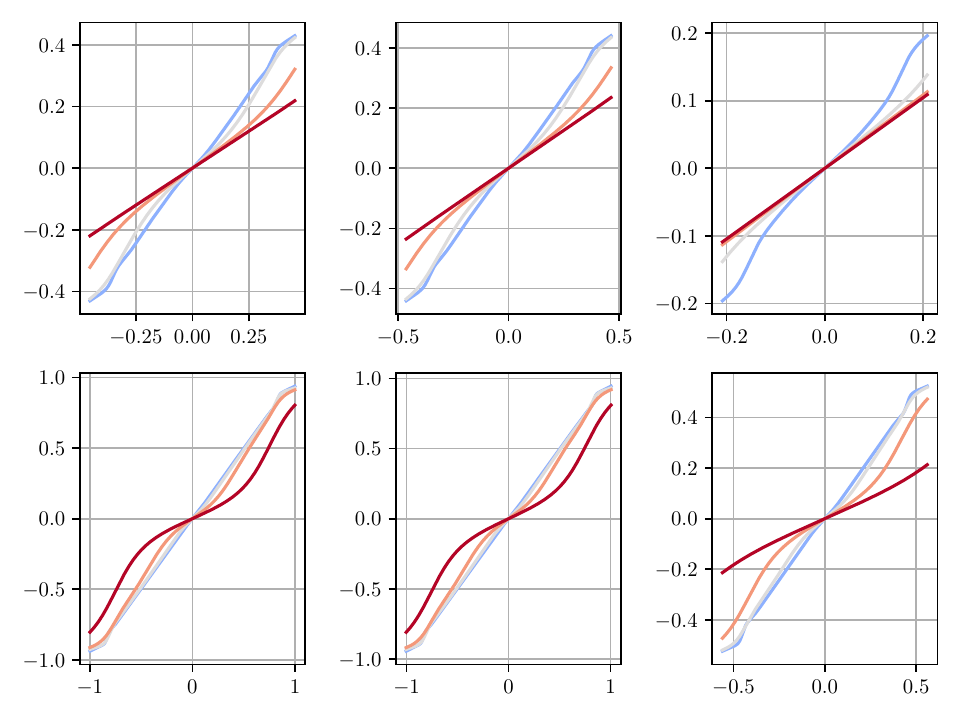}
	\end{tabular}%
	\vspace{-.2cm}
	\hspace*{11.2cm}{\resizebox{6cm}{!}{\(\sqrt{2t}=\)
		\foreach \ccolor\ssigma in {coolwarm1/0, coolwarm2/0.025, coolwarm3/0.05, coolwarm4/0.1, coolwarm5/0.2}
		{
			\tikz[baseline=.8*\hhheight]{\draw[\ccolor, ultra thick](0,0.3) -- (.5,.3);}\num{\ssigma}
		}
	}}
	\caption{%
		Left: \texttt{db2} initial generating sequence (\( K = 4\)).
		Right: \texttt{db4} initial generating sequence (\( K = 8\)).
		From top to bottom: Initial generating sequence \( h \) along with the corresponding scaling- and wavelet-functions \( \phi \) and \( \replaced{\omega}{\psi} \), learned generating sequence \( h \) along with the corresponding scaling- and wavelet-functions \( \phi \) and \( \replaced{\omega}{\psi} \), learned \replaced{potential functions}{potentials} \( -\log \psi(\argm, w, t) \) and the learned \gls{mmse} shrinkage functions \( y_t \mapsto y_t \replaced{+}{-} 2t \nabla \log \psi(y_t, w, t) \).
	}%
	\label{fig:wavelet pot}
\end{figure*}
\subsubsection{Learning Shearlets}
We initialize the one-dimensional low-pass filter \( h_1 \) and the two-dimensional directional filter \( P \) with the standard choices from~\cite{kutyniok_shearlab_2016}:
\( h_1 \) is initialized as maximally flat 9-tap symmetric low-pass filter\footnote{The \texttt{matlab} command \texttt{design(fdesign.lowpass(’N,F3dB’, 8, 0.5), ’maxflat’)} generates the filter.}, \( P \) is initialized as the maximally flat fan filter\footnote{The filter can be obtained in \texttt{matlab} using the Nonsubsampled Contourlet Toolbox by \texttt{fftshift(fft2(modulate2(dfilters(’dmaxflat4’, ’d’) ./ sqrt(2), ’c’)))}} described in~\cite{cunha_nonsubsampled_2006}.
Furthermore, \( \lambda_{j,k} \) is initialized as \( 1 \) for all scale levels \( j \) and shearings \( k \), and we set \( \eta_{j, k} = 0.5 \).

We enforce the following constraints on the parameter blocks:
The weighting parameters \( \lambda_{j,k} \) must satisfy non-negativity \( \lambda_{j, k} \in \R_{\geq 0} \).
The parameters \( h_1 \) and \( P \) specifying the shearlet system must satisfy \( h_1 \in \mathcal{H} \coloneqq \{ x \in \R^9 : \langle \mathds{1}_{\replaced{\R^9}{9}}, x \rangle_{\R^9} = 1 \} \) and \( P \in \mathcal{P} \coloneqq \{ x \in \R^{17 \times 17} : \norm{x}_1 = 1 \}\).

The projection operators can be realized as follows:
The projection onto the non-negative real line is just \( \proj_{\R_{\geq 0}}(x) = \max \{ x, 0 \} \).
The map \( \proj_{\mathcal{H}}(x) = x - \frac{\langle \mathds{1}_{\replaced{\R^9}{9}}, x \rangle_{\R^9} - 1}{9} \) realizes the projection onto the linear constrain encoded in \( \mathcal{H} \).
The projection onto the unit-one-norm-sphere is \( \proj_{\mathcal{P}}(x) = \operatorname{sgn}(x) \odot \proj_{\triangle^m}(|x|) \) (see e.g.~\cite{duchi_efficient_2008,Condat2015}), where we ignore the \deleted{case} degenerate case of projecting the origin where \( \proj_{\mathcal{P}} \) is not well defined.
Our implementation of the shearlet transformation is based on the ShearLab 3D~\cite{kutyniok_shearlab_2016} toolbox\footnote{See \href{http://shearlab.math.lmu.de/}{http://shearlab.math.lmu.de/}.}.
\begin{algorithm}[t]
	\DontPrintSemicolon
	\SetKwInOut{Output}{Output}
	\SetKwInOut{Input}{Input}
	\Input{\( x \in \R^m \)}
	\Output{\( y = \proj_{\triangle^m}(x) \)}
	\( u \coloneqq \operatorname{sort}(x) \)\tcp*{\( u_1 \geq \dotsc \geq u_m \)}
	\( K \coloneqq \max_{1\leq k \leq m} \{ k : \bigl( \sum_{r=1}^k u_r - 1 \bigr) / k < u_k \}\)\;
	\( \tau \coloneqq (\sum_{k=1}^K u_k - 1) / K \)\;
	\( y \coloneqq \max \{ x - \tau, 0 \} \)\tcp*{element-wise}
	\caption{%
		Simplex projection from~\cite{Held1974}.
	}%
	\label{alg:simplex proj}
\end{algorithm}

For the numerical experiments, we chose \( J = 2 \) scales and \( 5 \) shearings (\( k \in \{ -2, \dotsc, 2 \} \)).
We show the initial and learned filter weights \( \lambda_{j,k} \), the one-dimensional low-pass filter \( h_1 \), and the two-dimensional directional filter \( P \) in~\cref{fig:shearlet weights filters}.
The resulting shearlet system in the frequency- and time-domain, along with the learned potential functions, is shown in~\cref{fig:shearlet potentials}.
We again emphasize that the learned one-dimensional potential functions \( \psi_{j,k}(\argm, w_{j,k}, t) \) are distinctly different from the other models.
In particular, they exhibit multiple local minima, sometimes different from \( 0 \), such that certain image structures can be enhanced under this prior.
This is in stark contrast to the learned filter- and wavelet\added{-}responses, which show a single minimum at \( 0 \) and the classical heavy-tailed shape.

\cref{fig:shearlet potentials} also shows that the shearlet system only approximately fulfills the assumption~\eqref{eq:disjoint} and~\eqref{eq:constant}.
We analyze the shearlet system with respect to the assumption of disjoint support~\eqref{eq:disjoint} by visualizing the pair-wise cosine similarity of the magnitude of the spectra in~\cref{fig:shearlet inproduct}.
In detail, the figure shows \( \langle \frac{|\gamma_{\tilde{\jmath},\tilde{k}}|}{\norm{|\gamma_{\tilde{\jmath},\tilde{k}}|}}, \frac{|\gamma_{j,k}|}{\norm{|\gamma_{j,k}|}} \rangle_{\R^{n}} \), for \( \tilde{\jmath}, j \in \{ 1, 2 \} \) and \( \tilde{k}, k \in \{ -2,\dotsc, 2 \} \) and both cones.
Although less for the learned shearlet system, the plot is dominated by the main diagonal, indicating that the corresponding spectra are almost non-overlapping.
To meet the theoretical assumptions, it would be possible to penalize \( \langle \frac{|\gamma_{\tilde{\jmath},\tilde{k}}|}{\norm{|\gamma_{\tilde{\jmath},\tilde{k}}|}}, \frac{|\gamma_{j,k}|}{\norm{|\gamma_{j,k}|}} \rangle_{\R^n} \) for  \( \tilde\jmath \neq j \) and \( \tilde{k} \neq k \) during training.

The fact that the spectra are not constant over their support raises the question of how to choose \( \xi_a \) that best approximates~\eqref{eq:constant}.
During training and evaluation, we simply chose \( \xi_a = \norm{|\gamma_a|}_\infty \), where \( a \) is a two-index ranging over the chosen scale-shearing grid.
It remains an open question how the violation of the constraints~\eqref{eq:disjoint} and ~\eqref{eq:constant} influences the diffusion, and if there exists a better choice for \( \xi_a \).
\begin{figure}
	\centering
	\begin{subfigure}[b]{\columnwidth}
		\includegraphics[width=.47\columnwidth]{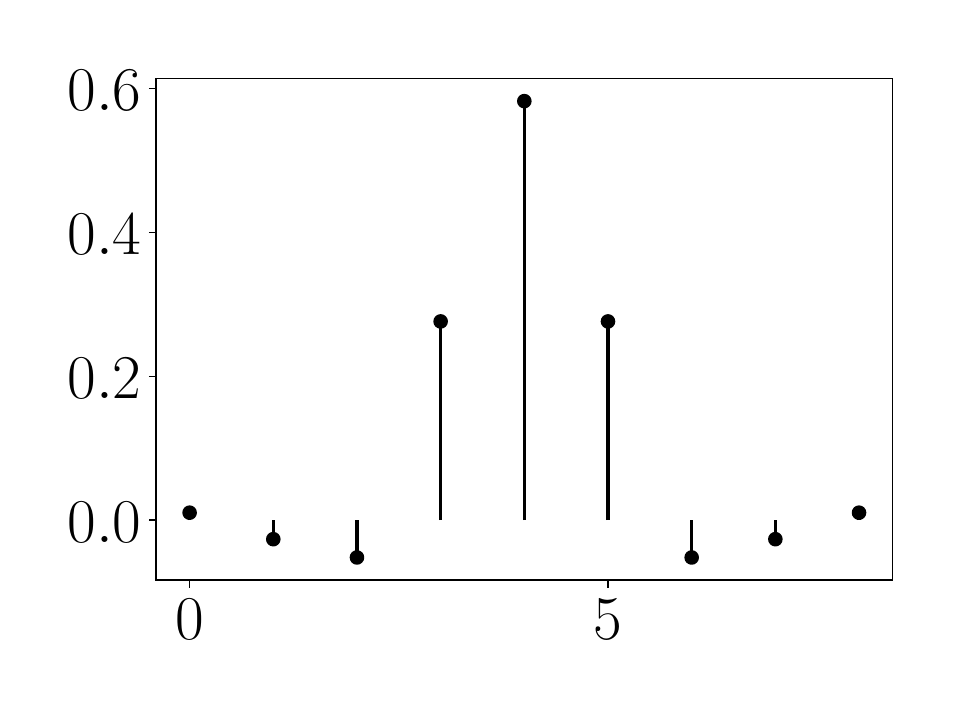}%
		\includegraphics[width=.47\columnwidth]{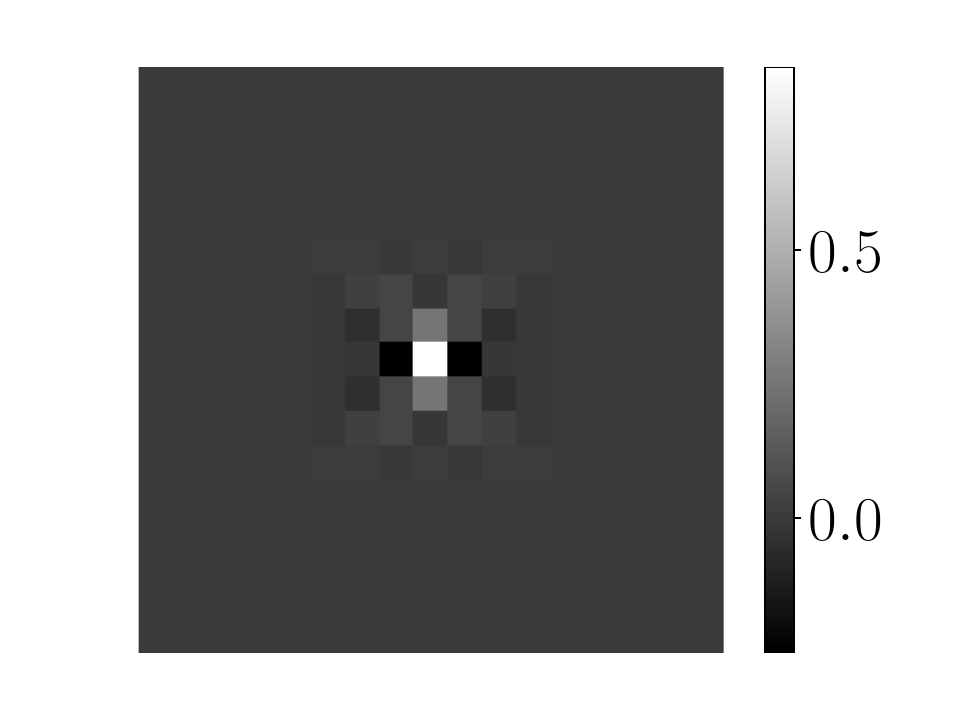}
		\caption{initial}\label{subfig:initial}
	\end{subfigure}
	\begin{subfigure}[b]{\columnwidth}
		\includegraphics[width=.47\columnwidth]{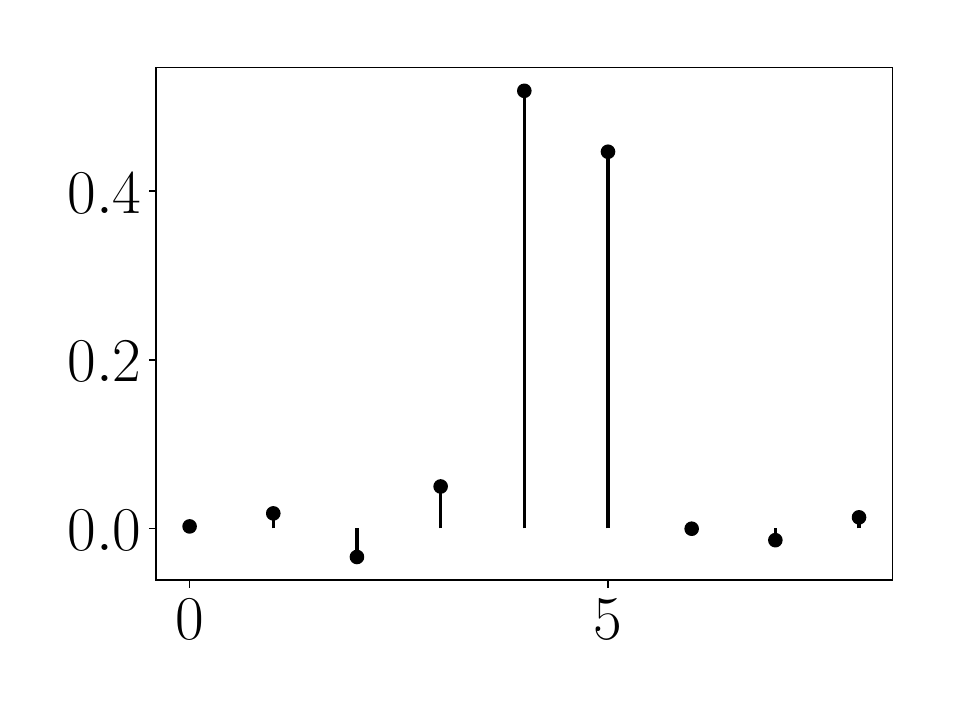}%
		\includegraphics[width=.47\columnwidth]{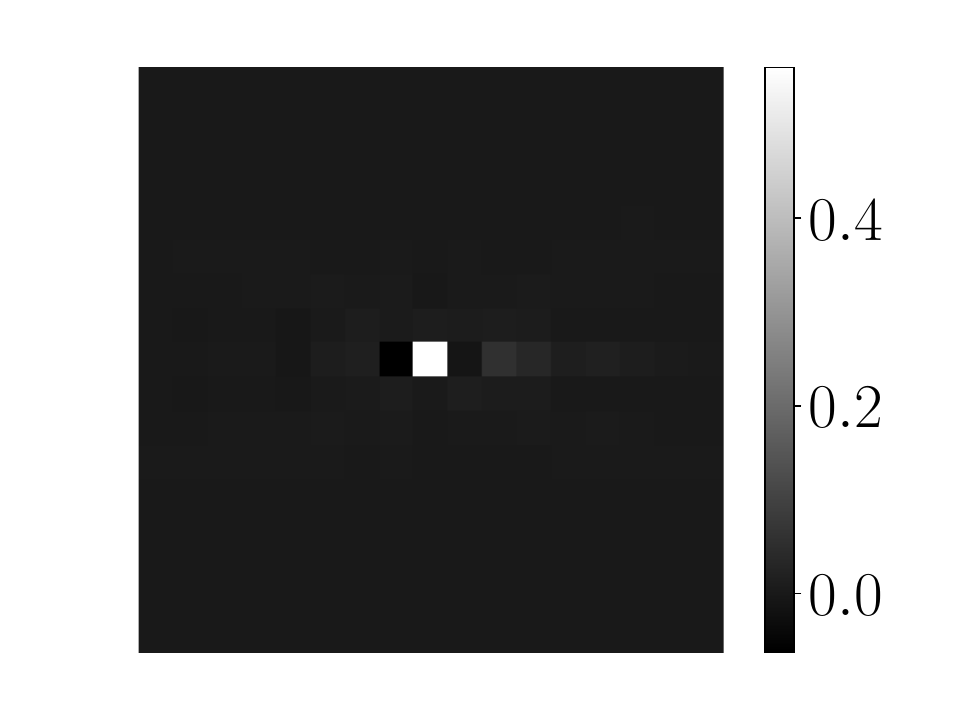}
		\caption{learned}\label{subfig:learned}
	\end{subfigure}
	\caption{%
		Initial (a) and learned (b) building blocks for the shearlet system:
		One-dimensional low-pass filter \( h_1 \) and the two-dimensional directional filter \( P \).
		The corresponding shearlet filters along with their frequency response, learned weights and learned potential functions are shown in~\cref{fig:shearlet potentials}.
	}%
	\label{fig:shearlet weights filters}
\end{figure}
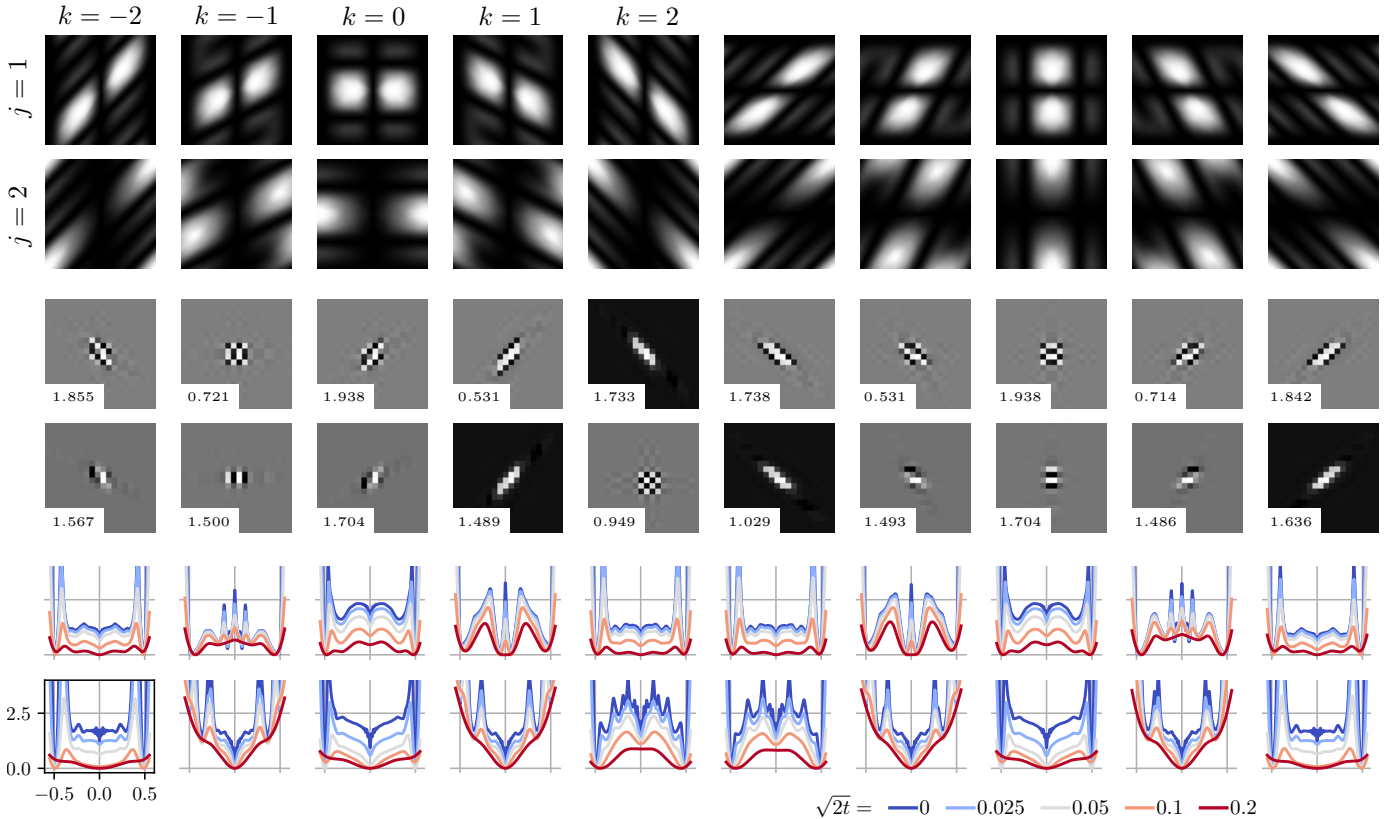
\begin{figure*}
	\centering
	\def\lamdas{{%
			1.8553598, 0.72088665, 1.9375001, 0.5313143, 1.7331071,%
			1.5667545, 1.4999282, 1.7035619, 1.4885511, 0.9492101,%
			1.7383913, 0.53091305, 1.9384255, 0.7136959, 1.8417741,%
			1.028708, 1.4931595, 1.70363, 1.4860928, 1.6355318%
	}}
	\begin{tikzpicture}
		\node at (.3, 3.5) {\includegraphics[trim=.4cm .3cm .4cm .3cm, clip, width=.95\textwidth]{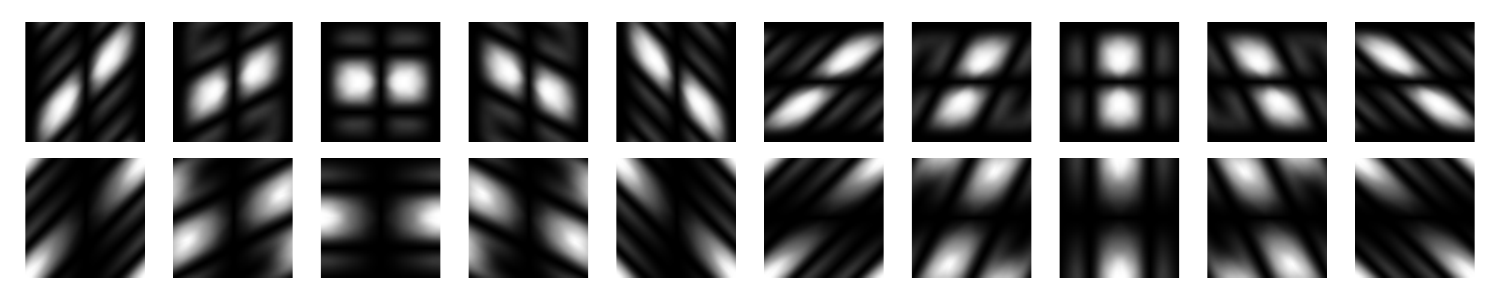}};
		\node at (.3, 0) {\includegraphics[trim=.4cm .3cm .4cm .3cm, clip, width=.95\textwidth]{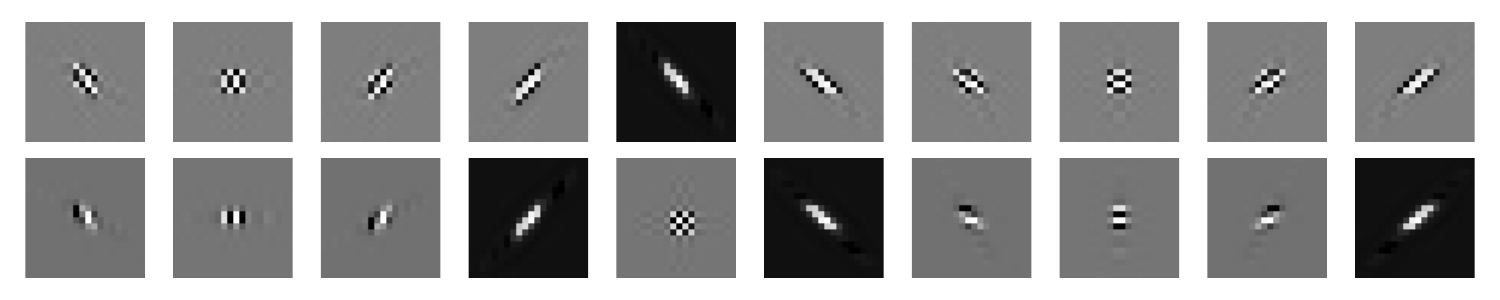}};
		\node at (0, -3.5) {\includegraphics[trim=.4cm .4cm .4cm .2cm, clip, width=.97\textwidth]{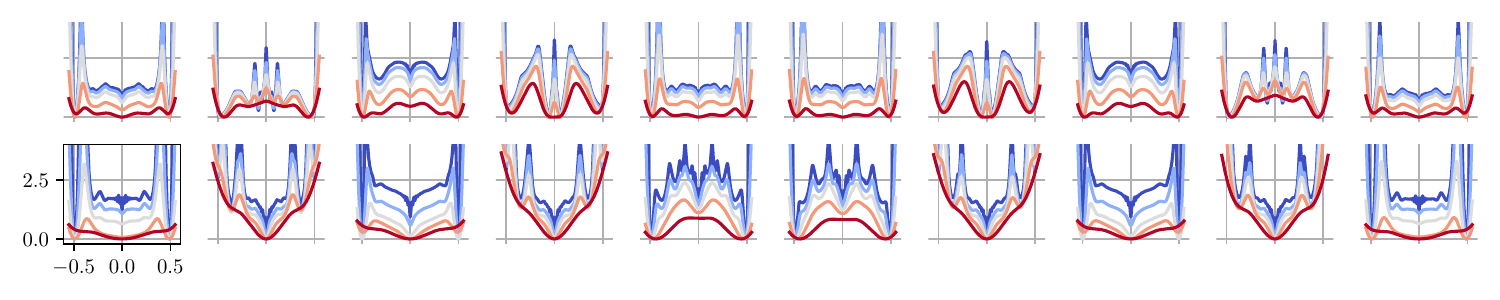}};
		\foreach \jj in {1, 2} {
			\pgfmathsetmacro{\yoff}{-(\jj * 1.7) + 6.05}
			\node [rotate=90] at (-8.8, \yoff) {\( j = \jj \)};
			\foreach \kk in {-2, -1, ..., 2} {
				\pgfmathsetmacro{\xoff}{\kk*1.80-4.15}
				\ifthenelse{\jj=1}{\node  at (\xoff, 5.3) {\( k = \kk \)};}{}
				\foreach \coneoff in {0, 5} {
					\pgfmathsetmacro{\yoff}{-(\jj * 1.65) + 1.9}
					\pgfmathsetmacro{\xoff}{(\kk + \coneoff)*1.78-4.54}
					\pgfmathsetmacro{\lammda}{\lamdas[(\jj - 1) * 5 + \kk + 2 + \coneoff * 2]}
					\node [fill=white] at (\xoff, \yoff) {\tiny \( \num[round-mode=places, round-precision=3]{\lammda} \)};
				}
			}
		}
	\end{tikzpicture}%
	\vspace{-.4cm}
	\hspace*{9.2cm}{\resizebox{6cm}{!}{\(\sqrt{2t}=\)
		\foreach \ccolor/\ssigma in {coolwarm1/0, coolwarm2/0.025, coolwarm3/0.05, coolwarm4/0.1, coolwarm5/0.2}
		{
			\tikz[baseline=.8*\hhheight]{\draw[\ccolor, ultra thick](0, 0.1) -- (.5, .1);}\num{\ssigma}
		}
	}}
	\caption{
		The five left columns show the frequency response, time-domain filters and corresponding potential functions of the first cone of the learned shearlet system.
		The inset numbers show the values of the corresponding weighting factor \( \lambda_{j, k} \).
		The five unlabeled right columns show the second cone.
	}%
	\label{fig:shearlet potentials}
\end{figure*}
\begin{figure}
	\centering
	\begin{subfigure}[b]{.49\columnwidth}
		\includegraphics[width=.95\columnwidth]{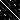}%
		\caption{initial}\label{subfig:inp initial}
	\end{subfigure}
	\begin{subfigure}[b]{.49\columnwidth}
		\includegraphics[width=.95\columnwidth]{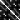}%
		\caption{learned}\label{subfig:inp learned}
	\end{subfigure}
	\caption{%
		Cosine similarity between the magnitude the spectra of different levels, shearings and cones of the initial (a) and learned (b) shearlet system.
		A system exactly fulfilling the assumption~\eqref{eq:disjoint} would be unpopulated off the main diagonal.
	}%
	\label{fig:shearlet inproduct}
\end{figure}
\subsection{Image Denoising}
To exploit our patch-based prior for whole-image denoising, following~\cite{zoran_learning_2011}, we define the expected patch log-likelihood of a noisy image \( y \in \R^n \) with variance \( \sigma^2(t) = 2t \) as
\begin{equation}
	\operatorname{epll}_\theta^{\mathrm{filt}}(y, t) = \sum_{j=1}^{\tilde{n}} p_j^{-1} \log f^{\mathrm{filt}}_\theta(P_j y, t).
\end{equation}
Here, \( \tilde{n} \) denotes the total number of overlapping patches (e.g.\ for \( n = 4 \times 4 \) and \( a = 3 \times 3 \), \( \tilde{n} = 4 \), ), \( P_i : \R^n \to \R^a \) denotes the patch-extraction matrix for the \( i \)-th patch and \( p_i = \bigl( \sum_{j=1}^{\tilde{n}} P_j^\top P_j \bigr)_{i, i} \) counts the number of overlapping patches to compensate for boundary effects (see~\cite[Appendix B]{romano_boosting_2017} for a more rigorous discussion).
The wavelet- and shearlet-based priors can act on images of arbitrary size.

Let \( \log f_\theta \) be either \( \operatorname{epll}_\theta^{\mathrm{filt}}, \log f_\theta^{\mathrm{wave}}\), or \( \log f_\theta^{\mathrm{conv}} \).
We consider two inference methods:
The one-step empirical Bayes estimate
\begin{equation}
		\hat{x}_{\mathrm{EB}}(y, t) = y + \sigma^2(t) \grad{1} \log f_\theta(y, t)
\end{equation}
corresponds to the Bayesian \gls{mmse} estimator.
Notice that, in the case of \( \log f_\theta = \operatorname{epll}_\theta^{\mathrm{filt}} \) the estimator
\begin{equation}
	\begin{aligned}
		\hat{x}_{\mathrm{EB}}(y, t) &= y + \sigma^2(t) \grad{1} \operatorname{epll}_\theta^{\mathrm{filt}}(y, t) \\
								  &= y + 2t \sum_{j=1}^{\tilde{n}} p_j^{-1} P_j^\top \grad{1} \log f^{\mathrm{filt}}_\theta(P_j y, t)
	\end{aligned}
	\label{eq:ebpa}
\end{equation}
computes patch-wise \gls{mmse} estimates and combines them by averaging.
This is known to be a sub-optimal inference strategy, since the averaged patches are not necessarily likely under the model~\cite{zoran_learning_2011}.
The discussion of algorithm utilizing patch-based priors for whole-image restoration is beyond the scope of this article.
We refer the interested reader to the works of~\cite{romano_boosting_2017,zoran_learning_2011} for a detailed discussion on this topic.
In addition, we refer to our previous conference publication~\cite{zach_explicit_2023}, in which we present a proximal gradient continuation scheme that slightly improves over the empirical Bayes estimate by allowing patch-crosstalk.

The second inference method we consider is the stochastic denoising algorithm proposed by~\cite{kawar_stochastic_2021} and summarized in~\cref{alg:stochastic image denoiser}.
In detail, this algorithm proposes a sampling scheme to approximately sample from the posterior of a denoising problem when utilizing diffusion priors.
This is achieved by properly weighting the \emph{score} \( \nabla \log f_\theta \) with the gradient of the data term while annealing the noise level.
Sampling from the posterior, as opposed to directly computing \gls{mmse} estimates with an empirical Bayes step, is known to produce sharper results when utilizing modern highly expressive diffusion models~\cite{kawar_stochastic_2021,Karras2022edm}.
We chose \( \epsilon = \num{5e-6} \), \( \sigma_{\replaced{C}{L}} = 0.01 \) and the exponential schedule \( \sigma_i = \sqrt{2t} \bigl(\frac{\sigma_{\replaced{C}{L}}}{\sqrt{2t}}\bigr)^{i/\replaced{C}{L}} \), using \( B = 3 \) inner loops and \( C = \replaced{100}{250} \) diffusion steps.
\begin{algorithm}[t]
	\DontPrintSemicolon%
	\SetKwInOut{Output}{Output}
	\SetKwInOut{Input}{Input}
	\Input{Variance schedule \( \{ \sigma_i \}_{i=1}^C \), noisy images \( y_t \in \R^m \), inner iterations \( B > 0 \in \mathbb{N} \), noise level \( \sigma_0 \) in \( y \)}
	\Output{Stochastically denoised image \( x_B \)}
	\For{\( i \in 1, \dotsc, \replaced{C}{L} \)}{
		\(\alpha_i \leftarrow \epsilon \sigma_i^2 / \sigma_{\replaced{C}{L}}^2 \)\;
		\For{\( b \in 1, \dotsc, B - 1 \)}{
			\(z_b \sim \mathcal{N}(0, \mathrm{Id}_m) \)\;
			\(g_b \leftarrow \nabla \log f_\theta(x_{b-1}, \sigma_i) + (y - x_{b-1}) / (\sigma_0^2 - \sigma_i^2) \)\;
			\(x_b \leftarrow x_{b-1} + \alpha_i g_b + \sqrt{2\alpha_i}z_b\)\;
		}
		\( x_0 \leftarrow x_B \)\;
	}
	\caption{%
		Stochastic image denoising algorithm from~\cite{kawar_stochastic_2021}.
	}%
	\label{alg:stochastic image denoiser}
\end{algorithm}

Let \( x \in \R^n \) denote a test sample from the distribution \( f_X \), and let \( \hat{x} \) denote the estimation of \( x \) given \( y_t = x + \sqrt{2t} \eta \) where \( \eta \sim \mathcal{N}(0, \mathrm{Id}_{\R^n}) \), through either of the discussed inference methods.
In~\cref{tab:denosing}, we show a quantitative evaluation \replaced{utilizing}{utilize} the standard metrics \gls{psnr} \( 10\log_{10}\frac{n}{\norm{\hat{x} - x}_2^2} \) and \gls{ssim}~\cite{zhou_ssim_2004} with a window size of \num{7} and the standard parameters \( K_1 = \num{0.01} \) and \( K_2 = \num{0.03} \).
The column with the heading \enquote{Patch-\acrshort{gsm}} utilized the Gaussian scale mixture parametrization discussed in~\cref{ssec:alternative parametrizations}.
\added{%
	The results are obtained for one run of the algorithms, i.e.\ we did not compute the expectation over the noise (neither in the construction of \( y_t \) nor during the iterations of the stochastic denoising algorithm).
	However, we did not observe any noteworthy deviation when performing different runs of the experiments.
}

The quantitative evaluation shows impressive results of the \replaced{model based on shearlet-responses}{shearlet model}, despite having very little trainable parameters.
In particular, it performs best across all noise \replaced{levels}{scales} and inference methods, with the exception of the one-step empirical Bayes denoising at \( \sigma = 0.2 \).
There, the patch-based model with \( a = 15 \times 15 \) performs best, but notably has about \num{50} times the number of trainable parameters.
By leveraging symmetries between the cones in the shearlet system, the number of trainable parameters could even be approximately halved.
These symmetries are strongly apparent in~\cref{fig:shearlet potentials}, where the \replaced{potential functions}{potentials} of the second cone (rightmost \num{5} potentials) are almost a perfect mirror image of the \replaced{potential functions}{potentials} of the first cone (leftmost \num{5} \replaced{potential functions}{potentials}).

Additionally, the table reveals that the empirical Bayes estimator beats the stochastic denoising in every quality metric.
This is not surprising, as --- in expectation --- it is the optimal estimator in the \gls{mmse} sense, which directly corresponds to \gls{psnr}.
Comparing the qualitative evaluation in~\cref{fig:denoising images} (empirical Bayes) to~\cref{fig:denoising images stochastic} (stochastic denoising) we do not observe that sharper images using the stochastic denoising algorithm; we are unsure why.
\begin{table*}[t]
	\centering
	\caption{%
		Quantitative denoising results in terms of \gls{psnr} and \gls{ssim} using one-step empirical Bayes denoising the stochastic denoising algorithm from~\cite{kawar_stochastic_2021}.
		\added{The intervals indicate the \num{0.95} confidence region, }
		\deleted{The quantitative evaluations of the noisy images differ since we used only a subset of the evaluation images for the stochastic denoising algorithm.}
		\replaced{b}{B}old typeface indicates the best method.
		\deleted{The symbol --- indicates that the algorithm did not finish due to numerical instability.}
	}%
	\label{tab:denosing}
	\begin{tabular}{ccS[table-format=1.3]*{7}{S[table-format=2.2(1),separate-uncertainty=true,retain-zero-uncertainty,table-align-uncertainty,tight-spacing]}}
		\toprule
		& & {\multirow{2}{*}{\( \sigma \)}} & {\multirow{2}{*}{\( y_t \)}} & \multicolumn{2}{c}{{Patch-GMM}} & {-GSM} & \multicolumn{2}{c}{{Wavelet}} & {\multirow{2}{*}{Shearlet}} \\\cmidrule(l{1em}r{1em}){5-6}\cmidrule(l{1em}r{1em}){8-9}
		& & & & {\( b = 7 \)} & {\( b = 15 \)} & {\( b = 7 \)} & {\( K = 4 \)} & {\( K = 8 \)} &  \\\midrule
		\parbox[t]{3mm}{\multirow{8}{*}{\rotatebox[origin=c]{90}{Empirical Bayes}}} & \parbox[t]{3mm}{\multirow{4}{*}{\rotatebox[origin=c]{90}{\gls{psnr}}}} & 
				   0.025 & 32.04(00) & 34.54(21) & 35.00(27) & 35.08(29) & 33.51(23) & 33.61(25) & \bfseries 35.30(38) \\
			& & 0.050 & 26.02(00)  &   30.44(32) & 30.79(37) & 30.80(37) & 29.32(29) & 29.48(31) & \bfseries 31.17(44) \\
			 & & 0.100 & 20.00(00)  &  27.03(42) & 27.27(46) & 27.20(44) & 25.72(36) & 25.90(36) & \bfseries 27.50(46) \\
			  & & 0.200 & 13.98(00)  & 24.24(48) & \bfseries 24.45(52) & 24.29(48) & 22.41(33) & 22.54(32) & 23.91(40) \\
		\cmidrule{2-10}
		&\parbox[t]{3mm}{\multirow{4}{*}{\rotatebox[origin=c]{90}{\gls{ssim}}}} &
		0.025 & 0.84(02) & 0.92(01) & 0.93(01) & 0.93(00) & 0.90(01) & 0.90(01) & \bfseries 0.94(00) \\
	  &&0.050 & 0.65(03) & 0.83(01) & 0.85(01) & 0.85(01) & 0.80(01) & 0.80(01) & \bfseries 0.87(01) \\
	  &&0.100 & 0.41(03) & 0.72(01) & 0.73(01) & 0.73(01) & 0.65(02) & 0.66(02) & \bfseries 0.75(01) \\
	  &&0.200 & 0.21(02) & 0.58(01) & \bfseries 0.60(01) & 0.59(01) & 0.48(02) & 0.49(02) & 0.55(01) \\
		\midrule
		\parbox[t]{3mm}{\multirow{8}{*}{\rotatebox[origin=c]{90}{Stochastic Denoising}}}&\parbox[t]{3mm}{\multirow{4}{*}{\rotatebox[origin=c]{90}{\gls{psnr}}}} &
				0.025 & 32.04(00) & 31.34(11) & 31.79(16) & 31.88(18) & 30.68(10) & 30.78(11) & \bfseries 32.40(27) \\
			  &&0.050 & 26.02(00) & 27.07(18) & 27.55(23) & 27.62(25) & 26.00(14) & 26.17(15) & \bfseries 28.46(40) \\
			  &&0.100 & 20.00(00) & 23.66(24) & 24.06(27) & 24.06(28) & 22.03(15) & 22.24(16) & \bfseries 24.94(44) \\
			  &&0.200 & 13.98(00) & 20.90(25) & \bfseries 21.24(29) & 21.12(28) & 18.60(13) & 18.71(13) & 21.10(34) \\
		\cmidrule{2-10}
		&\parbox[t]{3mm}{\multirow{4}{*}{\rotatebox[origin=c]{90}{\gls{ssim}}}} &
				0.025 & 0.84(02) & 0.84(02) & 0.85(01) & 0.86(01) & 0.82(02) & 0.82(02) & \bfseries 0.88(01) \\
			  &&0.050 & 0.65(03) & 0.69(02) & 0.71(02) & 0.72(02) & 0.65(02) & 0.65(02) & \bfseries 0.78(01) \\
			  &&0.100 & 0.41(03) & 0.52(02) & 0.54(02) & 0.54(02) & 0.46(03) & 0.47(03) & \bfseries 0.63(01) \\
			  &&0.200 & 0.21(02) & 0.36(02) & 0.37(02) & 0.37(02) & 0.29(02) & 0.30(02) & \bfseries 0.41(01) \\
		\midrule
			  & & &{\(\operatorname{dim} \theta \)} & {\num{5376}} & {\num{78400}} & {\num{3312}} & {\num{389}} & {\num{393}} & {\num{1642}} \\\bottomrule
	\end{tabular}
\end{table*}
\begin{figure*}
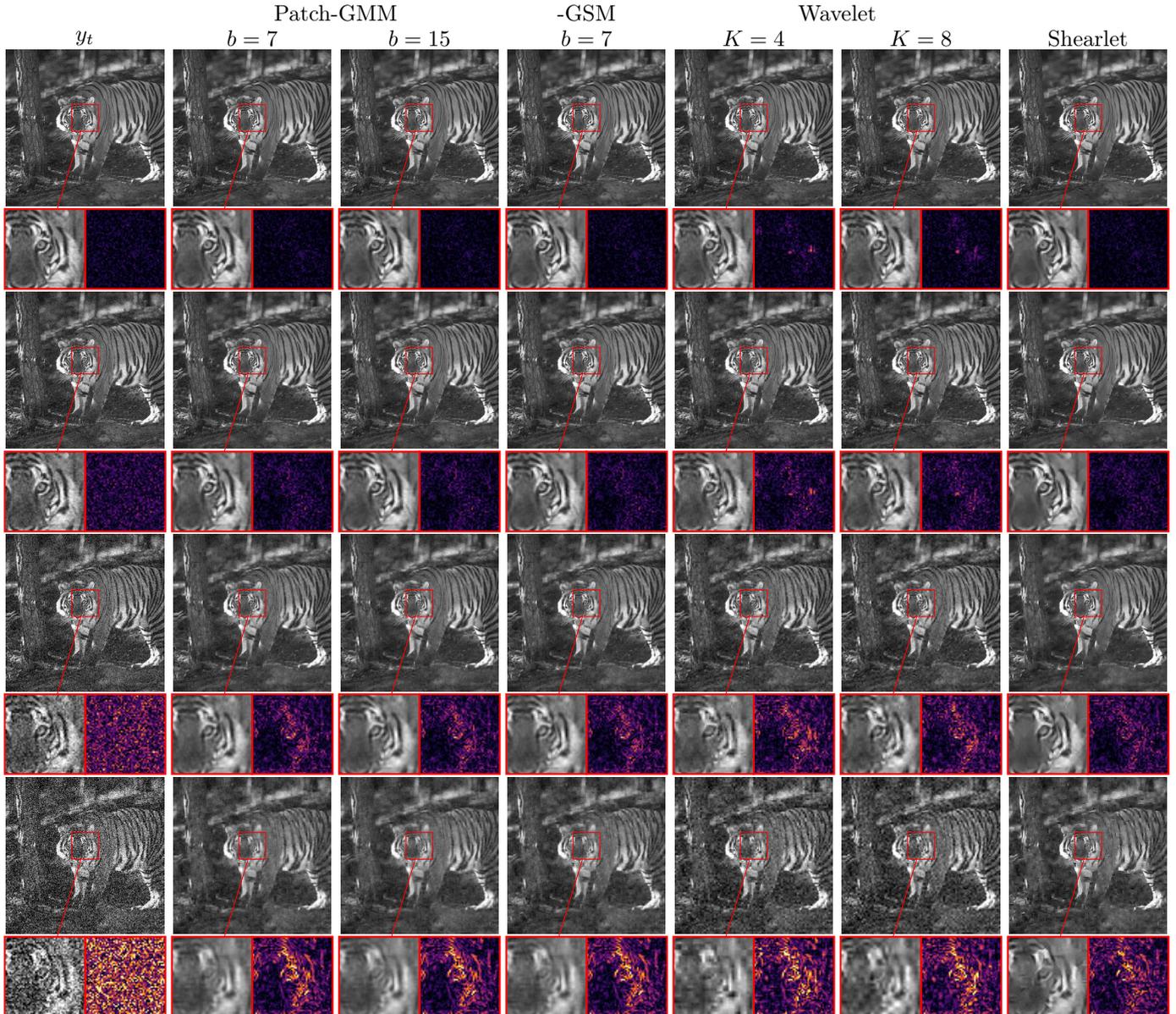

	\def\wwidth{2.45}
	\def\ppad{.15}
	\def\zzoomw{1.2}
	\def\spyxoff{.0}
	\def\spyyoff{-.15}
	\centering
	\begin{tikzpicture}
		\node at (2.5 * \wwidth + 2.5 * \ppad, -2) {Patch-\gls{gmm}};
		\node at (4 * \wwidth + 4 * \ppad, -2) {-\acrshort{gsm}};
		\node at (5.5 * \wwidth + 5.5 * \ppad, -2) {Wavelet};
		\foreach [count=\isigma] \ssigma in {0.025, 0.050, 0.100, 0.200} {
			\foreach [count=\imethod] \mmethod/\manno in {noisy/\(y_t\), gmm7/{\(b=7\)}, gmm15/{\( b = 15 \)}, gsm7/{\( b = 7 \)}, wavelet-db2/\( K = 4\), wavelet-db4/{\( K = 8 \)}, shearlet/{Shearlet}} {
				\ifthenelse{\isigma=1}{\node at (\imethod * \wwidth + \imethod * \ppad, -2.4) {\manno};}{}
				\coordinate (onn) at (\imethod * \wwidth + \imethod * \ppad + \spyxoff, -\isigma*\wwidth - \spyyoff - \isigma * \ppad - \isigma * \zzoomw);
				\begin{scope}[spy using outlines={rectangle, magnification=3, width=1.25cm, height=1.25cm}]
					\node at (\imethod * \wwidth + \imethod * \ppad, -\isigma * \wwidth - \isigma * \ppad - \isigma * \zzoomw) {\includegraphics[width=\wwidth cm]{./figures/denoising/tweedie/\ssigma/\mmethod/006_d.png}};
					\spy [red] on (onn) in node [left] at (\imethod * \wwidth + \imethod * \ppad + 1.25, -\isigma*\wwidth - \isigma * \ppad - \isigma * \zzoomw - 1.9);
				\end{scope}
				\begin{scope}[spy using outlines={rectangle, magnification=3, width=1.25cm, height=1.25cm, connect spies}]
					\node at (\imethod * \wwidth + \imethod * \ppad, -\isigma * \wwidth - \isigma * \ppad - \isigma * \zzoomw) {\includegraphics[width=\wwidth cm]{./figures/denoising/tweedie/\ssigma/\mmethod/006.png}};
					\spy [red] on (onn) in node [left] at (\imethod * \wwidth + \imethod * \ppad, -\isigma*\wwidth - \isigma * \ppad - \isigma * \zzoomw - 1.9);
				\end{scope}
			}
		}
	\end{tikzpicture}
	\caption{%
		Qualitative denoising results for one-step empirical Bayes denoising.
		In the rows, the noise standard deviation ranges in \( \sigma \in \{ 0.025, 0.05, 0.1, 0.2 \} \).
		The inlays show a zoomed region (magnifying factor \num{3}), and the absolute difference of the reconstruction to the ground truth image (\num{0}~\protect\drawcolorbar~\(\frac{1}{3}\)).
		The accompanying quantitative results are shown in~\cref{tab:denosing}.
	}%
	\label{fig:denoising images}
\end{figure*}
\begin{figure*}
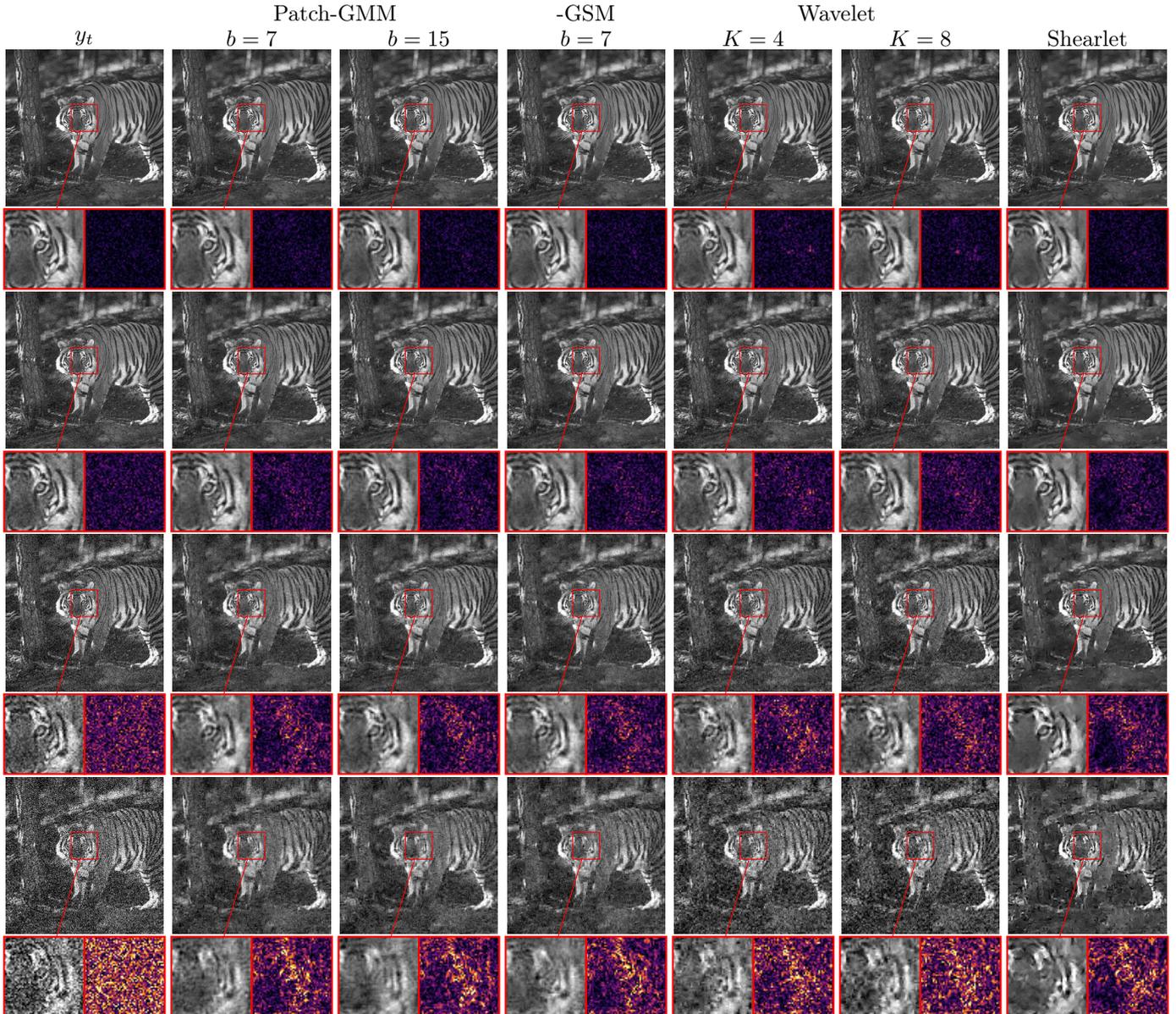

	\def\wwidth{2.45}
	\def\ppad{.15}
	\def\zzoomw{1.2}
	\def\spyxoff{.0}
	\def\spyyoff{-.15}
	\centering
	\begin{tikzpicture}
		\node at (2.5 * \wwidth + 2.5 * \ppad, -2) {Patch-\gls{gmm}};
		\node at (4 * \wwidth + 4 * \ppad, -2) {-\acrshort{gsm}};
		\node at (5.5 * \wwidth + 5.5 * \ppad, -2) {Wavelet};
		\foreach [count=\isigma] \ssigma in {0.025, 0.050, 0.100, 0.200} {
			\foreach [count=\imethod] \mmethod/\manno in {noisy/\(y_t\), gmm7/{\(b=7\)}, gmm15/{\( b = 15 \)}, gsm7/{\( b = 7 \)}, wavelet-db2/\( K = 4\), wavelet-db4/{\( K = 8 \)}, shearlet/{Shearlet}} {
				\ifthenelse{\isigma=1}{\node at (\imethod * \wwidth + \imethod * \ppad, -2.4) {\manno};}{}
				\coordinate (onn) at (\imethod * \wwidth + \imethod * \ppad + \spyxoff, -\isigma*\wwidth - \spyyoff - \isigma * \ppad - \isigma * \zzoomw);
				\begin{scope}[spy using outlines={rectangle, magnification=3, width=1.25cm, height=1.25cm}]
					\node at (\imethod * \wwidth + \imethod * \ppad, -\isigma * \wwidth - \isigma * \ppad - \isigma * \zzoomw) {\includegraphics[width=\wwidth cm]{./figures/denoising/stoch/\ssigma/\mmethod/006_d.png}};
					\spy [red] on (onn) in node [left] at (\imethod * \wwidth + \imethod * \ppad + 1.25, -\isigma*\wwidth - \isigma * \ppad - \isigma * \zzoomw - 1.9);
				\end{scope}
				\begin{scope}[spy using outlines={rectangle, magnification=3, width=1.25cm, height=1.25cm, connect spies}]
					\node at (\imethod * \wwidth + \imethod * \ppad, -\isigma * \wwidth - \isigma * \ppad - \isigma * \zzoomw) {\includegraphics[width=\wwidth cm]{./figures/denoising/stoch/\ssigma/\mmethod/006.png}};
					\spy [red] on (onn) in node [left] at (\imethod * \wwidth + \imethod * \ppad, -\isigma*\wwidth - \isigma * \ppad - \isigma * \zzoomw - 1.9);
				\end{scope}
			}
		}
	\end{tikzpicture}
	\caption{%
		Qualitative denoising results using the stochastic denoising algorithm from~\cite{kawar_stochastic_2021}.
		In the rows, the noise standard deviation ranges in \( \sigma \in \{ 0.025, 0.05, 0.1, 0.2 \} \).
		The inlays show a zoomed region (magnifying factor \num{3}), and the absolute difference of the reconstruction to the reference image (\num{0}~\protect\drawcolorbar~\(\frac{1}{3}\)).
		The accompanying quantitative results are shown in~\cref{tab:denosing}.
	}%
	\label{fig:denoising images stochastic}
\end{figure*}

\added{%
	The analysis of posterior variance is out of the scope of this paper.
	However, techniques for analyzing the posterior induced by diffusion models are also readily applicable to our models.
	In particular, we refer to~\cite{kawar_stochastic_2021} or related papers such as~\cite{chung2023diffusion} for an in-depth discussion of these techniques.
}
\subsection{Noise Estimation and Blind Image Denoising}
Within this and the following subsection, we describe two applications that arise as a byproduct of our principled approach: Noise estimation \added{(and, consequently, blind denoising)} and analytic sampling.
For both, we utilize the \replaced{model based on filter-responses}{patch-based model} as a stand-in but emphasize that similar results hold also for the \replaced{models based on wavelet- and shearlet-responses}{wavelet- and shearlet-based models}.

The construction of our model allows us to interpret \( f^{\mathrm{filt}}_\theta(\argm, t) \) as a time-conditional likelihood density.
Thus, it can naturally be used for noise \added{level} estimation:
\added{%
	We assume a noisy patch \( y \) constructed by \( y = x + \sigma \eta \), where \( x \sim f_X \), \( \eta \sim \mathcal{N}(0, \Id_{\R^n}) \) and \( \sigma \) is unknown.
	We can estimate the noise level \( \sigma \) by maximizing the likelihood of \( y \) w.r.t.\ to the diffusion time \( t \) --- \(\hat{t} = \argmax_t f^{\text{filt}}_\theta(y, t) \) --- and recover the noise level via \( \sigma = \sqrt{2\hat{t}} \).
}

\added{%
	To demonstrate the feasibility of this approach,
}
\cref{fig:noise estimation} shows the expected negative-log density%
\footnote{For visualization purposes, we normalized the negative-log density to have a minimum of zero over \( t \): \( l_\theta(x, t) = -\log f^{\mathrm{filt}}_\theta(x, t) - (\max_t \log f^{\mathrm{filt}}_\theta(x, t)) \).} %
\( \mathbb{E}_{p \sim f_X, \eta \sim \mathcal{N}(0, \mathrm{Id})} \added{\bigl[} l_\theta(p+\sigma\eta, t) \added{\bigr]} \) over a range of \( \sigma \) and \( t \).
The noise \added{level} estimate \( \sigma \mapsto \argmin_t \mathbb{E}_{p \sim f_X, \eta \sim \mathcal{N}(0, \mathrm{Id})} \added{\bigl[} l_\theta(p + \sigma\eta, t) \added{\bigr]} \) perfectly matches the identity map \( \sigma \mapsto \sqrt{2t} \).

\replaced{%
	In addition, we can leverage this noise level estimation procedure to perform blind heteroscedastic denoising with the same model as follows:
}{%
	We can utilize \emph{one} model for \emph{noise estimation} and \emph{heteroscedastic blind denoising} as follows:
}
First, for all \( \tilde{n} \) overlapping patches \( P_j y \) in the corrupted image, we estimate the noise level through \( \hat{t}_j = \argmax_t f^{\mathrm{filt}}_\theta(P_j y, t) \).
Given the noise \replaced{levels}{level map} \( \hat{t}_{\added{j}} \), we can estimate the clean image with an empirical Bayes step of the form
\begin{equation}
	\hat{x}_{\mathrm{blind}}(y) = y + 2 \sum_{j=1}^{\tilde{n}} \hat{t}_j  p_j^{-1} P_j^\top \grad{1} \log f^{\mathrm{filt}}_\theta(P_j y, \hat{t}_j),
	\label{eq:blind eb}
\end{equation}
where for each patch \( P_j y \) we utilize the estimated noise level \( \hat{t}_j \).

In~\cref{fig:blind denosing}, the original image is corrupted by heteroscedastic Gaussian noise with standard deviation \num{0.1} and \num{0.2} in a checkerboard pattern, which is clearly visible in the noise \added{level} map.
In the restored image and the absolute difference to the reference, the checkerboard pattern is hardly visible, indicating that the noise \added{level} estimation is robust also when confronted with little data.
\begin{figure*}[t]
	\centering
	\resizebox{1\textwidth}{!}{
	\begin{tikzpicture}
		\node at (0, 0) {\includegraphics[trim=.9cm 2cm 1cm 2cm, clip, width=.3\textwidth]{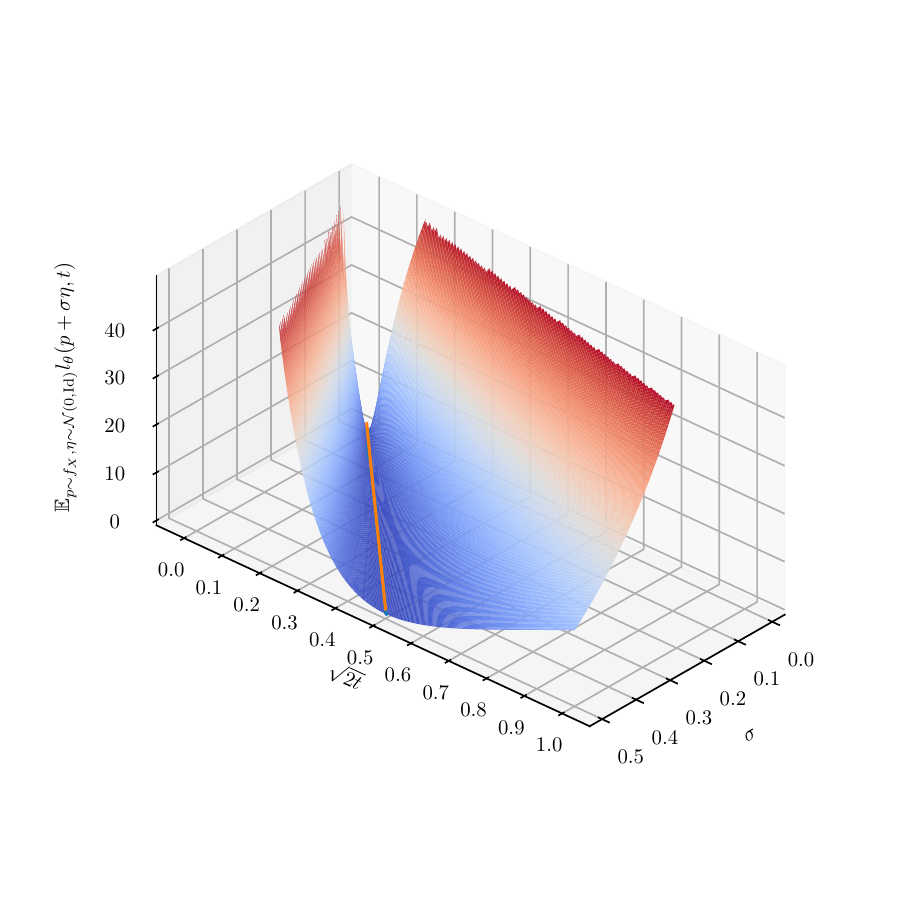}};
		\node at (6, -.3) {\includegraphics[trim=.5cm .1cm 1cm 1.2cm, clip, width=.3\textwidth]{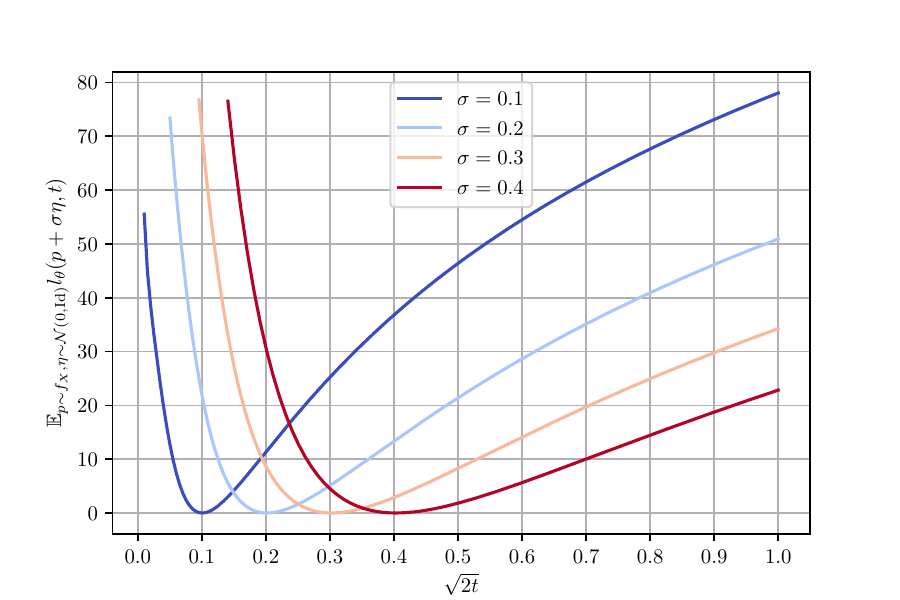}};
	\end{tikzpicture}}
	\caption{%
		Expected normalized negative-log density along with the noise estimate
		\protect\tikz[baseline=\hheight]\protect\draw[mplorange, thick] (0, 0.1) -- ++(0.4, 0); \( \sigma \mapsto \argmin_t \mathbb{E}_{p \sim f_X, \eta \sim \mathcal{N}(0, \mathrm{Id})}\added{\bigl[}l_\theta(p+\sigma\eta, t) \added{\bigr]} \), \protect\tikz[baseline=\hheight]\protect\draw [mplgreen, thick] (0, 0.1) -- ++(0.4, 0) node [right,black] {\( \sigma \mapsto \sqrt{2t} \)}; (left) and the slices at \( \sigma \in \{ 0.1, 0.2, 0.3, 0.4 \} \) (right).
	}%
	\label{fig:noise estimation}
\end{figure*}
\begin{figure*}[t]
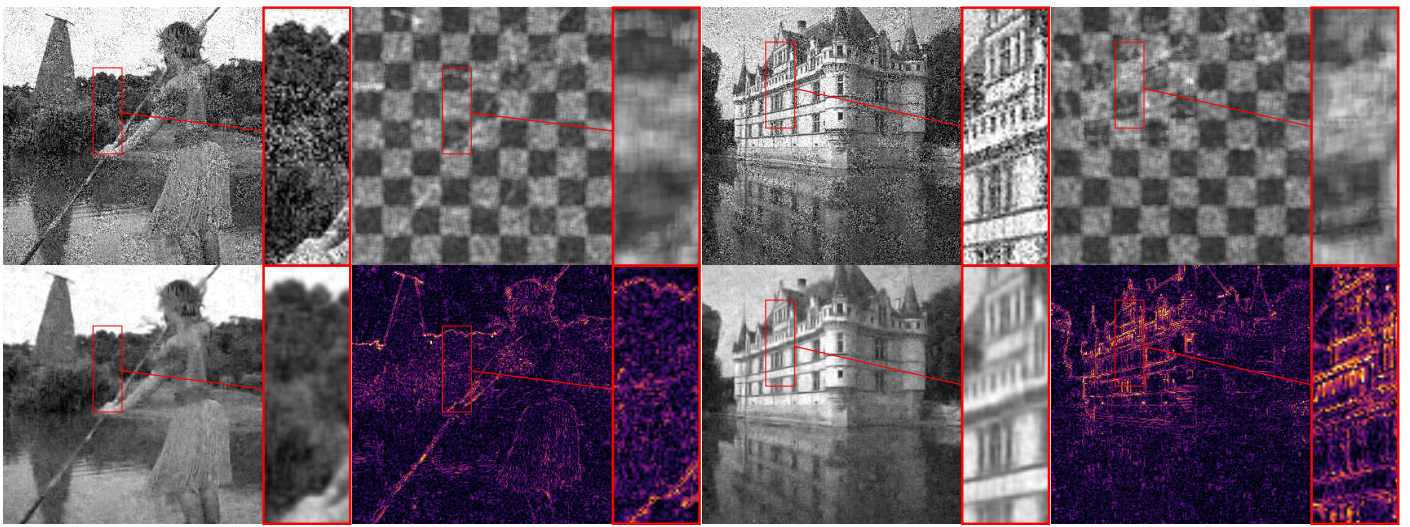

	\centering
	\resizebox{\textwidth}{!}{%
		\begin{tikzpicture}
			\def\wwidth{3cm}
			\foreach [count=\imcount] \cherry in {1, 2} {
				\foreach [count=\ii] \wwhich in {noisy, estimate} {
					\begin{scope}[spy using outlines={rectangle, magnification=3, width=1.cm, height=3cm, connect spies}]
						\node[rotate=-90] at (\imcount * 2 * \wwidth + \imcount * 2 * 1.03cm + \ii*\wwidth+\ii*1.03cm, 0) {\includegraphics[width=\wwidth]{./figures/blind/\cherry/\wwhich.png}};
						\spy [red] on (-\imcount * .3cm + \imcount * 2 * \wwidth + \imcount * 2 * 1.03cm + \ii*\wwidth+\ii*1.03cm, \imcount*.3cm) in node [left] at (2.5cm+\imcount * 2 * \wwidth + \imcount * 2 * 1.03cm + \ii*\wwidth+\ii*1.03cm, 0);
					\end{scope}
				}
			}
			\foreach [count=\imcount] \cherry in {1, 2} {
				\foreach [count=\ii] \wwhich in {denoised, diff} {
					\begin{scope}[spy using outlines={rectangle, magnification=3, width=1.cm, height=3cm, connect spies}]
						\node[rotate=-90] at (\imcount * 2 * \wwidth + \imcount * 2 * 1.03cm + \ii*\wwidth+\ii*1.03cm, -3cm) {\includegraphics[width=\wwidth]{./figures/blind/\cherry/\wwhich.png}};
						\spy [red] on (-\imcount * .3cm + \imcount * 2 * \wwidth + \imcount * 2 * 1.03cm + \ii*\wwidth+\ii*1.03cm, -3cm+\imcount*.3cm) in node [left] at (2.5cm+\imcount * 2 * \wwidth + \imcount * 2 * 1.03cm + \ii*\wwidth+\ii*1.03cm, -3cm);
					\end{scope}
				}
			}
		\end{tikzpicture}%
	}
	\caption{%
		Noise estimation and blind denoising.
		Top left: Image \( y \) corrupted with heteroscedastic Gaussian noise in a checkerboard pattern with standard deviation \num{0.1} and \num{0.2}.
		Top right: \replaced{Noise level map}{Patch-wise estimate \( \hat{t} \) of the noise level} (\num{0}~\protect\drawcolorbarbw~\num{0.5}).
		Bottom left: One-step empirical Bayes denoising result using~\eqref{eq:blind eb}.
		Bottom right: Absolute difference to the reference image (\num{0}~\protect\drawcolorbar~\(\frac{1}{3}\)).
	}%
	\label{fig:blind denosing}
\end{figure*}
\subsection{Sampling}%
\label{ssec:sampling}
A direct consequence of~\cref{cor:marginal} is that our models admit a simple sampling procedure:
The statistical independence of the components allows drawing random patches by 
\begin{equation}
	Y_t = \sum_{j=1}^J \frac{k_j}{\norm{k_j}^2} U_{j, t},
	\label{eq:analytic sampling}
\end{equation}
where \( U_{j, t} \) is a random variable on \( \R \) sampled from the one-dimensional \gls{gmm} \( \psi_j(\argm, w_j, t) \).
The samples in~\cref{fig:patch generation results} indicate a good match over a wide range of \( t \).
However, for small \( t \) the generated patches appear slightly noisy, which is due to an over-smooth approximation of the sharply peaked marginals around \( 0 \).
This indicates that the (easily adapted) discretization of \( \mu_l \) equidistant over the real line is not optimal.
We discuss alternative parametrizations in~\cref{sec:discussion}.
\begin{figure*}
	\centering
	\def\wwidth{3.6cm}
	\begin{tikzpicture}
		\foreach [count=\isigma] \ssigma/\llabel/\ccolor in {0.000/0/coolwarm1, 0.025/0.025/coolwarm2, 0.050/0.05/coolwarm3, 0.100/0.1/coolwarm4, 0.200/0.2/coolwarm5}{
			\foreach [count=\iwhich] \which in {true, analytical}{
				\node at (\isigma*\wwidth+\isigma*1, -\iwhich*\wwidth/2-\iwhich*5) {\includegraphics[width=\wwidth]{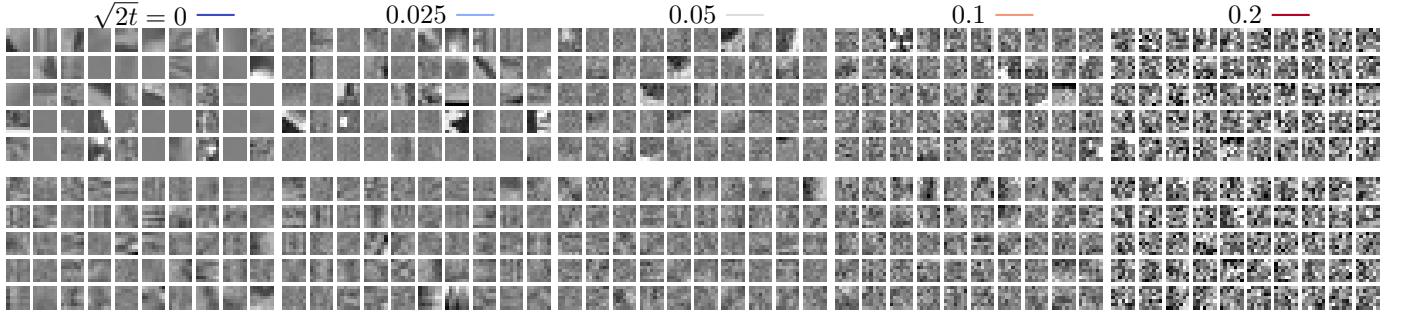}};
			}
			\node (anno\isigma) at (\isigma*\wwidth+\isigma, -.92){\ifthenelse{\isigma=1}{\( \sqrt{2t} = \llabel \)}{\( \llabel \)}};
			\draw [\ccolor, thick] (anno\isigma.east) -- ++(0.5, 0);
		}
	\end{tikzpicture}
	\caption{%
		Ground-truth samples from the random variable \( Y_t \) (top) and samples generated by the analytic sampling procedure~\eqref{eq:analytic sampling} (bottom).
	}%
	\label{fig:patch generation results}
\end{figure*}
\section{Discussion}%
\label{sec:discussion}
\subsection{Alternative Parametrizations}%
\label{ssec:alternative parametrizations}
The potential functions of the \replaced{models based on filter- and wavelet-responses}{patch- and wavelet-based models,} shown in~\cref{fig:patch results} and~\cref{fig:wavelet pot}\deleted{,} exhibit leptokurtic behavior, which has been noticed quite early in the literature~\cite{hua_statistics_1999,hinton_discovering_2001,teh_energy_2003,RoBl09,weiss_model_2007,Hyvrinen2009}.
To model these leptokurtic \replaced{potential functions}{potentials}, our parametrization relies on one-dimensional \glspl{gmm} with a-priori chosen equidistant means on the real line.
The \gls{gmm} is a very natural choice in our framework, as the Gaussian family is the only function family closed under diffusion (i.e.\ convolution with a Gaussian, cf.\ the central limit theorem).
However, as a consequence, the discretization of the means over the real line has to be fine enough to allow proper modeling of the leptokurtic marginals.
Thus, the majority of the learnable parameters are actually the weights of the one dimensional Gaussian mixtures.
This motivates the consideration of other \replaced{expert}{potential} functions \( \psi \).

An extremely popular choice for modeling \deleted{the potentials of} the distribution of filter\replaced{-}{s}responses on natural images is the Student-t \replaced{expert}{function}~\cite{hinton_discovering_2001,RoBl09}
\begin{equation}
	x \mapsto \biggl( 1 + \frac{x^2}{2} \biggr)^{-\alpha}.
\end{equation}
As outlined above, the convolution of this function with a Gaussian can not be expressed in closed form.
However, there exist approximations, such as the ones shown in~\cite{forchini_distribution_2008} or~\cite[Theorem 1]{berg2009density}, which we recall here for completeness:
Let \( X \) be a random variable on \( \R \) with density
\begin{equation}
	f_X(x) = \frac{\Gamma(\frac{\nu + 1}{2})}{\sqrt{\nu\pi}\Gamma(\frac{\nu}{2})} \bigl( 1 +  \frac{x^2}{\nu} \bigr)^{-\frac{\nu + 1}{2}},
\end{equation}
where \( \Gamma(z) = \int_0^\infty t^{z-1}\exp(-t)\,\mathrm{d}t \) is the Gamma function, and let \( Y_t \) be a random variable defined as previously.
Then, \( f_{Y_t} = \lim_{N\to \infty} f_{Y_t}^{(N)} \) where
\begin{equation}
	\begin{aligned}
		f_{Y_t}^{(N)}(y) &= \frac{\exp \bigl( -\frac{y^2}{4t} \bigr)\Gamma\bigl( \frac{\nu + 1}{2} \bigr)}{\sqrt{4t\pi}\Gamma\bigl(\frac{\nu}{2}\bigr) \bigl( \frac{4t}{\nu} \bigr)^{\frac{\nu}{2}}} \times \\
		\sum_{n=0}^{N}\biggl( &\frac{1}{n!} \Bigl( \frac{y^2}{4t} \Bigr)^n \Psi\Bigl( \frac{\nu + 1}{2}, \frac{\nu}{2} + 1 - n, \frac{\nu}{4t} \Bigr) \biggr)
	\end{aligned}
	\label{eq:forchini}
\end{equation}
with the confluent hypergeometric function of the second kind (also known as Tricomi's function, or the hypergeometric \( U \) function)~\cite{Tricomi1947} \( \Psi \).

We show \( - \log f_{Y_t}^{(N)} \) for different \( N \) and \( t > 0 \) in~\cref{fig:student t approx}, along with \( - \log f_{Y_t} \) which we computed numerically.
Notice that~\eqref{eq:forchini} is composed of two terms:
A Gaussian with variance \( 2t \) and an infinite polynomial in the even powers, filling up the tails of the distribution.
Thus, it is not surprising that the approximation fails to model the tails of the distribution when \( t \) is small, and becomes better as \( t \) increases and the distribution approaches a Gaussian.
\begin{figure*}
	\centering
	\begin{tikzpicture}
		\foreach [count=\ittt] \ttt in {0.1, 1.0, 3.0}
		{
			\node at (6 * \ittt, 0) {\includegraphics[width=.3\textwidth]{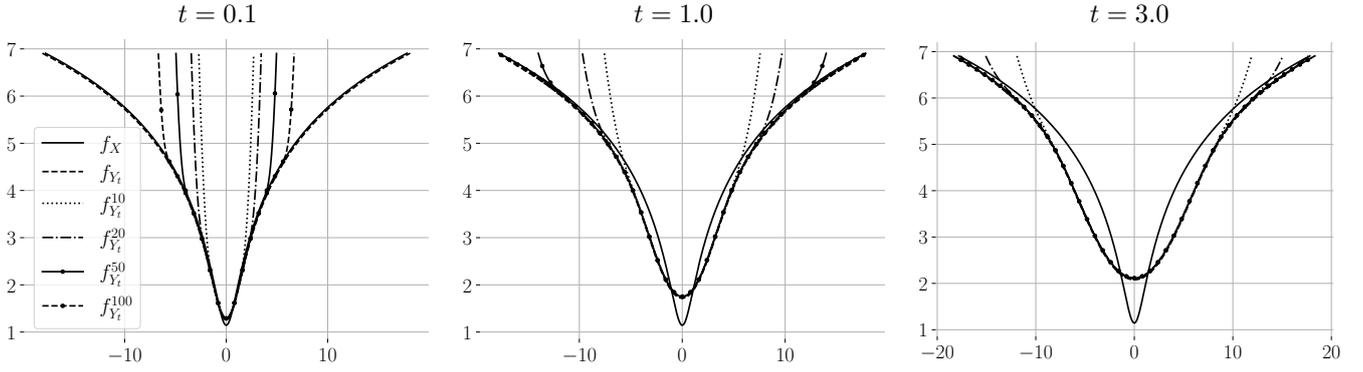}};
			\node at (6 * \ittt, 2.5) {\( t = \ttt \)};
		}
	\end{tikzpicture}
	\caption{%
		Forchini's~\cite{forchini_distribution_2008} approximation \( -\log f_{Y_t}^{(N)} \) (see~\eqref{eq:forchini}) of the density of the sum of a t- and a normally distributed random variable with standard deviation \( \sqrt{2t} \).
	}%
	\label{fig:student t approx}
\end{figure*}

Another popular expert function is the \gls{gsm}
\begin{equation}
	x \mapsto \int_{-\infty}^{\infty} (2\pi z^2\sigma^2)^{-\frac{1}{2}} \exp\biggl( -\frac{x^2}{2z^2\sigma^2} \biggr) f_Z(z)\, \mathrm{d}z
\end{equation}
which has been used in the context of modeling both the distributions of filter\added{-}\deleted{responses}~\cite{schmidt_generative_2010,qi_gao_generative_2012} as well as wavelet\replaced{-responses}{coefficients}~\cite{wainwright_scale_1999,portilla_image_2003}.
Here, \( f_Z \) is the \emph{mixing density} of the \emph{multiplier} \( Z \).
Thus, \glspl{gsm} can represent densities of random variables that follow
\begin{equation}
	X = ZU
\end{equation}
where \( Z \) is a scalar random variable and \( U \) is a zero mean Gaussian (see~\cite{andrews_scale_1974} for conditions under which a random variable can be represented with a \gls{gsm}).
In practice, for our purposes we model the mixing density as a Dirac mixture \( f_Z = \sum_{i=1}^I w_i \delta_{z_i} \) with \( (w_1,\dotsc,w_I)^\top \in \triangle^I \) and \( z_i \) a-priori fixed.
Then, the \gls{gsm} expert reads
\begin{equation}
	\psi_j^{\text{\gls{gsm}}}(x, w_j, t) = \sum_{i=1}^I w_{ji} (2\pi z_i^2(t))^{-\frac{1}{2}} \exp\biggl( -\frac{x^2}{2z_i^2(t)} \biggr),
\end{equation}
where without loss of generality we set \( \sigma = 1 \).

To show the practical merit of this parametrization in our context, we train a patch-model using \( b = 7 \) with the following choice of \( z_i \):
As the \gls{gmm} experiments indicated that the discretization of the means was a bit too coarse, we chose \( z_i = 0.01 \times 1.4^{i - 1} \), such that \( z_1 = 0.01 < 0.016 = \sigma_0 \).
The idea outlined in~\cref{ssec:patch model} naturally extends to such models:
Diffusion (for the \( j \)-th feature channel) amounts to rescaling \( z_i^2 \mapsto z_i^2 + 2t\norm{k_j}^2 \).

We show learned filters and \added{their} corresponding \added{potential functions} and \replaced{activation functions}{activations} \replaced{when utilizing}{for} a \gls{gsm} in~\cref{fig:gsm model}, where we used \( I = \num{20} \) scales.
The number of learnable parameters is \( (a - 1)(a + I) \), which is \num{3312} when modeling \( a = 7 \times 7 \) patches with our choice of \( I \).
This is considerably less than the \num{5376} parameters for the \gls{gmm}, which, as discussed in~\cref{ssec:sampling}, seems to still be discretized too coarsely.
This might indicate that a \gls{gsm} parametrization is more fit for this purpose.
Indeed, the quantitative analysis presented in~\cref{tab:denosing} shows superiority of the patch-based \gls{gsm} model over the patch-based \gls{gmm}.
However, note that the \gls{gmm} parametrization is strictly more versatile as it does not assume a maximum at \num{0}.
For instance, \glspl{gsm} can not model the potential functions of the \replaced{model based on shearlet-responses}{shearlet-based model} (\cref{fig:shearlet potentials}).
\begin{figure*}
	\centering
	\includegraphics[trim=4.5cm .7cm 4cm .6cm, clip, width=\textwidth]{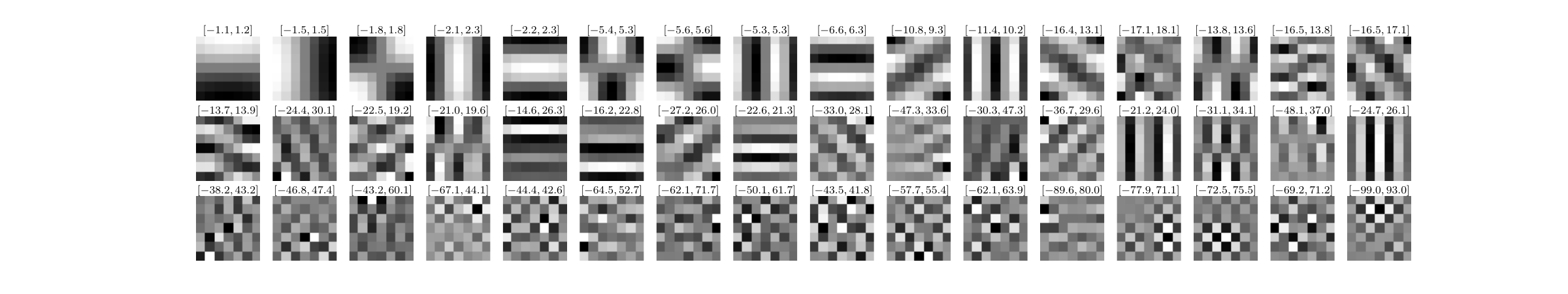}
	\includegraphics[trim=4.5cm .3cm 4cm .7cm, clip, width=\textwidth]{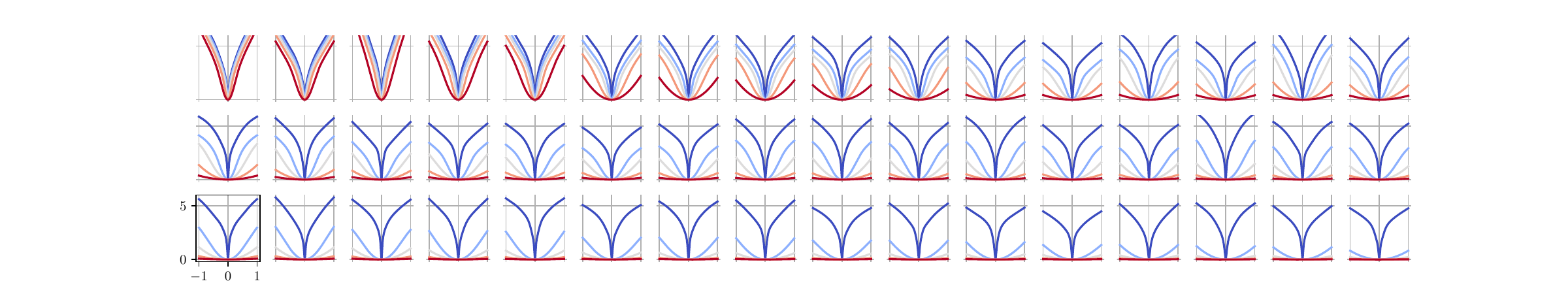}
	\vspace{-3mm}
	\includegraphics[trim=4.5cm .3cm 4cm .7cm, clip, width=\textwidth]{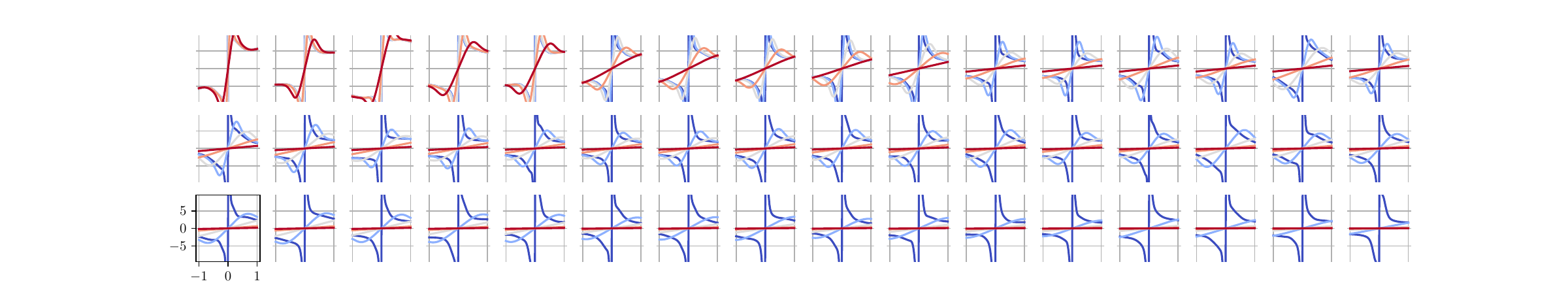}
	\hspace*{8.2cm}{\resizebox{6cm}{!}{\(\sqrt{2t}=\)
		\foreach \ccolor\ssigma in {coolwarm1/0, coolwarm2/0.025, coolwarm3/0.05, coolwarm4/0.1, coolwarm5/0.2}
		{
			\tikz[baseline=\hhheight*0.8]{\draw[\ccolor, ultra thick](0,0.1) -- (.5,.1);}\num{\ssigma}
		}
	}}
	\caption{%
		Learned filters \( k_j \) (top, the intervals show the values of black and white respectively, amplified by a factor of \num{10}), potential functions \( -\log \psi_j^{\mathrm{GSM}}(\argm, w_j, t) \) and activation functions \( -\nabla \log \psi_j^{\mathrm{GSM}}(\argm, w_j, t) \).
	}%
	\label{fig:gsm model}
\end{figure*}
\subsection{Designing more expressive models}
All architectures discussed until now are shallow in the sense that they model the distribution of filter\added{-}responses (either directly, or through wavelets or shearlets).
A possible extension of our work would be to consider deep networks, i.e., networks with more than one layer.
Indeed, many popular image restoration frameworks, such as trainable non-linear reaction diffusion~\cite{chen_trainable_2017} or the cascade of shrinkage fields~\cite{schmidt_shrinkage_2014} employ trainable Gaussian mixture \replaced{potential functions}{activations} (often referred to more \replaced{generally}{genreally} as radial basis splines).
However, they are typically trained as point estimators in a classic discriminative (task-specific) framework, and have not been studied in the context of diffusion priors.
We quickly note that the diffusion in the trainable non-linear reaction diffusion is a diffusion in \emph{image space}, whereas our framework considers diffusion in \emph{probability space}.
Extending the idea of diffusion in probability space to deep networks is non-trivial.
We believe that such models can only be tackled by approximating the diffusion \replaced{\gls{pde}}{process}.
\subsubsection{Wavelets: Modeling neighborhoods}
In essence, \replaced{the model based on wavelet-responses}{our wavelet model} described in~\cref{ssec:wavelet model} models the \replaced{histogram}{distribution} of wavelet coefficients in different sub-bands.
However, it does not take the spatial neighborhood (neither in its own sub-band nor of siblings or parents) into account.
There have been many attempts at making these types of models more powerful:
Guerrero-Colon et al.~\cite{guerrero-colon_image_2008} introduce mixtures of \glspl{gsm} to model the spatial distribution of wavelet coefficients in and across sub-bands.
The authors of~\cite{gupta_generalized_2018} extend this idea to mixtures of generalized Gaussian scale model mixtures.
We believe that these extensions can be used also in our work.
In particular, modeling disjoint neighborhoods leads to a block diagonal structure in the \replaced{product}{global} \gls{gmm}, which can be efficiently inverted.
However, modeling disjoint neighborhoods is known to introduce artifacts~\cite{portilla_image_2003}.
Still, such models can be globalized, e.g.\ by utilizing ideas similar to the expected patch log-likelihood~\cite{zoran_learning_2011}, which amounts to applying a local model to overlapping local neighborhoods individually and averaging the results.

Another interesting research direction with applications to generative modeling would be to condition the distribution of the wavelet coefficients on their parent sub-bands.
Notice that when utilizing conditioning and modeling local neighborhoods, we essentially recover the wavelet score-based generative model of Guth et al.~\cite{guth2022wavelet}.
Their model uses the score network architecture proposed in~\cite{nichol2021improved}, but we believe that modelling local neighborhoods could yield results that are close to theirs.
\subsection{Patch versus Convolutional Model}%
\label{ssec:conv v patch}
\added{%
	One of the major open questions in this work is the relationship between the models based on filter-responses and shearlet-responses.
	We again want to emphasize that they are distinctly different:
	The former \enquote{only} models the distribution of filter\added{-}responses, essentially forming a histogram.
	In particular, the distribution of filter\added{-}responses of natural images on arbitrary filters will \emph{always} exhibit leptokurtic behavior~\cite{hua_statistics_1999,weiss_model_2007}, with sharp peaks at \( 0 \) (see our learned potential functions of the model based on filter-responses in~\cref{fig:patch results}).
	The experts in the model based on shearlet-responses do \emph{not} model the marginal distribution of filter\added{-}responses, but takes into account the non-trivial correlation of overlapping patches
	This leads to significantly more complex expert functions with multiple minima, sometimes different from zero (see our learned potential functions of the model based on shearlet-responses in~\cref{fig:shearlet potentials}).
	Although quite well known in the literature~\cite{zhu_filters_1998,zoran_learning_2011,chen_trainable_2017,romano_boosting_2017}, this distinction is sometimes overlooked (e.g.\ when~\cite{RoBl09} chose the restrictive Student-t potential functions in their convolutional fields-of-experts model).
	To the best of our knowledge, this paper is the first in proposing strategies to learn patch-based and convolutional priors in a unified framework.%
}

\added{%
	The assumption of non-overlapping spectra of the filters in the convolutional model~\eqref{eq:disjoint} is in analogy to the assumption of pair-wise orthogonality of the filters in the patch model~\eqref{eq:ortho}:
	From~\eqref{eq:disjoint} immediately follows that \( \langle \mathcal{F}k_j, \mathcal{F}k_i \rangle_{\mathbb{C}^n} = 0 \) when \( i \neq j \).
	Thus, in some sense, the convolutional model becomes a patch-based model in Fourier space.
	However, the relationship remains unclear and deserves being investigated further.
}

\added{%
	The second assumption --- that the spectra are constant over their support~\eqref{eq:constant} --- restricts the space of admissible filters quite heavily.
	Unfortunately, we did not find a way to relax this constraint and we believe that it can not be relaxed without losing exact diffusion.
	However, we think that the constraint can be relaxed such that the diffusion \gls{pde} is fulfilled within some error bounds.
}

\section{Conclusion}%
\label{sec:conclusion}
In this paper, we introduced \glspl{pogmdm} as products of Gaussian mixture experts that allow for an explicit solution of the diffusion \gls{pde} of the associated density.
For \replaced{models acting on filter-, wavelet-, and shearlet-responses}{three different model classes --- patches-based models, wavelet- and convolutional models ---} we derive conditions for the associated filters and potential functions such that the diffusion \gls{pde} is exactly fulfilled.
Our explicit formulation enables learning of image priors simultaneously for all diffusion times using denoising score matching.
Numerical results demonstrated that \glspl{pogmdm} capture the statistics of the underlying distribution well for any diffusion time.
As a byproduct, our models \replaced{can naturally be used for noise estimation and}{are suitable for} blind heteroscedastic \deleted{image} denoising.

Future work will include the design of \deleted{design of} multi-layer architectures for which the diffusion can be expressed analytically, or approximated within some error bounds.
In addition, the learned models could be evaluated on more involved inverse problems such a deblurring or even medical imaging.
Further, the extensive evaluation of the \replaced{model based on filter-responses}{patch-based model} in terms of sampling the \replaced{distribution}{distribtuion} and performing heteroscedastic blind denoising can also be applied to the \replaced{models based on wavelet- and shearlet-responses}{wavelet- and shearlet-based models}.
\added{%
	Finally, the connection between the models based on filter- and shearlet-responses
	deserves being investigated further.
}
\bibliography{bibliography}
\end{document}